\documentclass[11pt]{article}
\usepackage{amssymb,amsthm,amsmath,float,subfigure}
\usepackage{appendix}
\usepackage{booktabs} 

\usepackage[total={6.5in,8.75in}, top=1.2in, left=0.9in, includefoot]{geometry}
\usepackage{graphicx}
\usepackage[pdftex,bookmarks, colorlinks, citecolor =blue, breaklinks]{hyperref}

\newcommand{\Eq}[1]{(\ref{eq:#1})}

\newcommand{\Lem}[1]{Lem.~\ref{lem:#1}}
\newcommand{\Cor}[1]{Cor.~\ref{cor:#1}}
\newcommand{\Con}[1]{Conj.~\ref{con:#1}}

\newcommand{\Sec}[1]{\S \ref{sec:#1}}
\newcommand{\Fig}[1]{Fig.~\ref{fig:#1}}
\newcommand{\Tbl}[1]{Table~\ref{tbl:#1}}
\newcommand{\App}[1]{App.~\ref{app:#1}}

\newcommand{\InsertFig}[4]
{\begin{figure}[h!t]
 \centerline{
 \includegraphics[width=#4]{#1}
 }
 \caption{{\footnotesize #2}
 \label{fig:#3}}
\end{figure}}

\newcommand{\InsertFigTwo}[5] {
\begin{figure}[h!t]
 \centerline{
 \includegraphics[width=#5]{#1}
 \hskip 0.1in
 \includegraphics[width=#5]{#2}
 }
 \caption{{\footnotesize #3}
 \label{fig:#4}}
\end{figure}}

\newcommand{\InsertFigFour}[7] {
\begin{figure}[h!t]
 \centerline{
\renewcommand{\arraystretch}{0.01}
 \begin{tabular}{cc}
 \includegraphics[width=#7]{#1}& \includegraphics[width=#7]{#2} \\
 \includegraphics[width=#7]{#3} & \includegraphics[width=#7]{#4}
 \end{tabular}
 }
 \caption{{\footnotesize #5}
 \label{fig:#6}}
\end{figure}}

\newcommand{\bC}{{\mathbb{ C}}}
\newcommand{\bN}{{\mathbb{ N}}}

\newcommand{\bQ}{{\mathbb{ Q}}}
\newcommand{\bR}{{\mathbb{ R}}}
\newcommand{\bS}{{\mathbb{ S}}}

\newcommand{\bZ}{{\mathbb{ Z}}}

\newcommand{\cC}{{\cal C}}

\newcommand{\cE}{{\cal E}}

\newcommand{\cG}{{\cal G}}

\newcommand{\cT}{{\cal T}}

\newcommand{\eps}{\varepsilon}
\newcommand{\vphi}{\varphi}

\newcommand{\tr}{\mathop{\rm tr}}

\newcommand{\fix}[1]{\mathop{\mathrm{Fix}(#1)}}

\newtheorem{thm}{Theorem}
\newtheorem{lem}[thm]{Lemma}
\newtheorem{cor}[thm]{Corollary}
\newtheorem{con}{Conjecture}

\newcommand{\beq}[1]{\begin{equation}\label{eq:#1}}
\newcommand{\eeq}{\end{equation}}

\newenvironment{se}[1]{\equation\label{eq:#1}\aligned}{\endaligned\endequation}
\newcommand{\bsplit}[1]{\begin{se}{#1}}
\newcommand{\esplit}{\end{se}}

\newenvironment{example}[1][]
 {
	\setlength \leftmargini {1.0em}		
	\setlength \topsep {0.5em}			
	\begin{quote}
	{\it Example#1} }
	{\end{quote}
 }
\newcommand{\bexam}[1][:]{\begin{example}[#1]}
\newcommand{\eexam}{\end{example}}

\title{Critical Invariant Circles in Asymmetric and Multiharmonic Generalized Standard Maps}
\author{Adam M.~Fox\footnote{Adam.Fox@colorado.edu}
 \ and\ 
 James D.~Meiss\footnote{James.Meiss@colorado.edu}
	\ \thanks
 {
 AMF and JDM and were supported in part by NSF grant DMS-1211350. Useful conversations with Renato Calleja, Rafael de la Llave, Keith Julien, Robert Easton, and Holger Dullin
 are gratefully acknowledged. 
 }
 \\
 	Department of Applied Mathematics\\
 University of Colorado \\
	Boulder, CO 80309-0526 \\
}
\date{\today}
\begin{document}
\maketitle

\begin{abstract}
\vspace*{1ex}
\noindent

Invariant circles play an important role as barriers to transport in the dynamics of area-preserving maps. KAM theory guarantees the persistence of some circles for near-integrable maps, but far from the integrable case all circles can be destroyed. 
A standard method for determining the existence or nonexistence of a circle, Greene's residue criterion, requires the computation of long-period orbits, which can be difficult if the map has no reversing symmetry. 
We use de la Llave's quasi-Newton, Fourier-based scheme to numerically compute the conjugacy of a Diophantine circle conjugate to rigid rotation, and the singularity of a norm of a derivative of the conjugacy to predict criticality. We study near-critical conjugacies for families of rotational invariant circles in generalizations of Chirikov's standard map. 

A first goal is to obtain evidence to support the long-standing conjecture that when circles breakup they form cantori, as is known for twist maps by Aubry-Mather theory.  The location of the largest gaps is compared to the maxima of the potential when anti-integrable theory applies.
A second goal is to support the conjecture that locally most robust circles have noble rotation numbers, even when the map is not reversible. We show that relative robustness varies inversely with the discriminant for rotation numbers in quadratic algebraic fields. 
Finally, we observe that the rotation number of the globally most robust circle generically appears to be a piecewise-constant function in two-parameter families of maps.

\end{abstract}

\section{Introduction}\label{sec:Intro}
Area-preserving maps are discrete-time analogs of Hamiltonian systems and represent the simplest of such systems to exhibit chaotic dynamics \cite{Meiss92}. A well-studied family are the generalized standard maps, $f_{\eps}: M \to M$, where $M = \bS\times \bR$, given by 
\bsplit{StdMap}
	x' &= x + \Omega(y+\eps g(x)) \mod~1, \\
	y' &= y+ \eps g(x) .
\esplit
We assume that the \emph{force}, $g(x)$, is a smooth, periodic function of the angle $x$ with period one and zero average, so that $f_\eps$ has ``zero net flux." The function $\Omega$ is the \emph{frequency map}, its dependence upon the momentum $y$ represents the anharmonicity of the dynamics. When $g$ and $\Omega$ are differentiable, \Eq{StdMap} is an area-preserving diffeomorphism. We will assume that $g$ and $\Omega$ are analytic.

Most studies of the dynamics of \Eq{StdMap} use the force
\beq{Chirikov}
	g_1(x) = \frac{1}{2\pi} \sin(2\pi x), 
\eeq
and the frequency map
\beq{Twist}
	\Omega_1(y) = y .
\eeq
We will refer to this choice as \emph{Chirikov's standard map} \cite{Chirikov79}. This map applies to many physical models, including the dynamics of a kicked rotor, a charged particle in electrostatic waves, the adsorption of a layer of atoms on a crystal surface, and the motion of a particle in a relativistic cyclotron, see the references in \cite{Meiss92}. 

When $\eps=0$, all of the orbits of \Eq{StdMap} lie on invariant circles, $\cT_y=\bS \times \{y\}$, and the dynamics on each circle is simply rigid rotation with rotation number $\omega = \Omega(y_0)$, i.e., $(x_t,y_t)=(x_0 +\omega t, y_0)$. We refer to any circle that is homotopic to $\cT_0$ as \emph{rotational}. Since the restriction of $f_\eps$ to a rotational invariant circle is a degree-one homeomorphism, each orbit on the circle has the same rotation number $\omega$. 

KAM theory predicts that invariant circles with ``sufficiently irrational"\footnote
{Namely, $\omega$ is Diophantine, see \App{Farey}.}
rotation numbers persist when $\eps$ is sufficiently small \cite{delaLlave01, Poshel01}. Alternatively when $\eps$ is sufficiently large, \emph{converse} KAM theory predicts that there are no rotational invariant circles \cite{MacKay85}. The study of the breakup of these circles was initiated in the work of Greene \cite{Greene79} and led to the development of a renormalization theory of breakup \cite{MacKay93}.

The standard method for studying this breakup is Greene's residue criterion, recalled in \App{Greene}. Its implementation requires accurate computation of periodic orbits and the linearization of the map about these orbits. This computation is much easier when the force $g$ is odd. In this case, \Eq{StdMap} is reversible: there is an orientation-reversing involution that conjugates $f_\eps$ to its inverse.\footnote
{Since $\Omega_1$ is odd, Chirikov's map also has a second reversor, so we call it a ``doubly reversible" map.}
 The fixed sets of this \emph{reversor} intersect every rotational circle, see \App{Symmetries}. A consequence is that every rotational circle is symmetric and can be approximated by a sequence of symmetric periodic orbits. Symmetric orbits can be found by a one-dimensional root-finding algorithm. Due to this additional structure, all previous studies of breakup have assumed odd forcing.

In this paper we will study the breakup of invariant circles for the generalized standard map when it lacks such symmetries. The mechanism for this breakup in nonreversible maps remains an open question, first posed by MacKay in his 1982 thesis, reprinted in \cite{MacKay93}. Our goal is to provide evidence for three conjectures.

\begin{con}\label{con:Cantorus} When a rotational invariant circle of \Eq{StdMap} is destroyed it becomes a cantorus.
\end{con}
\noindent
This is known to be true for nondegenerate twist maps, i.e., when $D\Omega(y) \neq 0$. For this case, the powerful theory of Aubry and Mather implies the existence of recurrent ``minimizing" invariant sets for each rotation number $\omega$ \cite{Mather82, Aubry83a}. These invariant sets always lie on Lipschitz graphs over $x$, and when $\omega$ is irrational they are either rotational invariant circles or invariant Cantor sets, Percival's \emph{cantori} \cite{Percival79}. Thus a Diophantine invariant circle for a twist map persists up to a critical parameter value, $\eps_{cr}$, where it loses smoothness, and then for larger $\eps$ becomes a cantorus.\footnote
{Of course, it may reform for larger values of $\eps$, and indeed the boundary of existence of an invariant circle for multi-parameter maps is often quite complex \cite{Bullett86, Wilbrink87, Ketoja89}.}

The notion of ``minimizing" arises from the Frenkel-Kontorova interpretation: orbits of a twist map are equivalent to equilibria of a one-dimensional, nearest-neighbor-coupled chain of particles. For example, define the periodic external potential $V$ by $DV(x) = g(x)$, and, for \Eq{Twist}, the spring potential energy $\tfrac12(x'-x)^2$ for each pair of neighboring particles at positions $x,x' \in \bR$, to give the energy
\beq{FK}
	E(x,x')=\tfrac{1}{2}(x'-x)^2 + \eps V(x).
\eeq
Any sequence $\{x_t: t \in \bZ\}$ that is a critical point of the (formal) sum $\cE = \sum_{t} E(x_t,x_{t+1})$ is equivalent to an orbit of \Eq{StdMap} for the frequency map \Eq{Twist}. Aubry-Mather sets are minima of the energy with respect to variations with compact support.

The creation of the gaps in a cantorus is most easily understood through the anti-integrable (AI) limit \cite{Aubry90}. As $\eps$ increases, the potential energy in \Eq{FK} will begin to dominate the spring energy and particles in the minimizing state will tend to fall into potential wells. For large enough $\eps$, there will be no particles in a neighborhood of the maxima of $V$, opening gaps in the circle. The orbits of each gap form a bi-infinite family that Baesens and MacKay called a \emph{hole}. More formally, AI theory studies the continuation of the critical points of $\lim_{\eps\to\infty} \frac{1}{\eps} \cE$, to finite $\eps$. In the limit, the particles simply sit on any sequence of critical points of $V$. If these critical points are nondegenerate and the ``acceleration" of the sequence is bounded, then such states may be continued to finite, large enough $\eps$. Ordered AI states with irrational $\omega$ continue to cantori \cite{MacKay92b}.

However, neither Aubry-Mather theory nor the Frenkel-Kontorova energy \Eq{FK} apply when the frequency map does not satisfy the twist condition. Thus, it is not known whether invariant circles of nontwist maps are destroyed by the formation of a family of gaps, nor whether there are remnant Cantor sets that remain. 

One well-studied example is the so-called standard nontwist map of Howard and Hohs \cite{Howard84}, where the frequency map is 
\beq{NonTwist}
	\Omega_2(y,\delta) = y^2 - \delta.
\eeq
KAM theory still applies to this map, providing---as for \Eq{NonTwist}---the curvature of $\Omega$ does not vanish when the twist vanishes \cite{Delshams00}. The nontwist map has rotational invariant circles that are meandering, in the sense that they are not graphs over $x$, and the nontrivial fixed point of the renormalization group for critical circles in this map has been extensively studied \cite{Apte03, Apte05, Fuchss06, Wurm11}. 

In \Sec{CriticalCircle} we will study the formation of gaps in invariant circles of generalized standard maps with a variety of frequency maps and multi-harmonic forces, giving evidence for \Con{Cantorus} in asymmetric and nontwist maps.

Greene conjectured that the last invariant circle of Chirikov's map has the rotation number $\gamma =\tfrac12(1+\sqrt{5})$, the golden mean. Though this conjecture has never been proved, it is supported by strong numerical evidence \cite{MacKay93}. More generally the golden invariant circle will not necessarily be globally most robust, but one expects that the last invariant circle in any neighborhood has a related rotation number:

\begin{con}\label{con:Noble}
Rotational invariant circles with noble rotation numbers are locally most robust.
\end{con}
\noindent
A frequency ratio $\omega = \nu_1/\nu_2$ is \emph{noble} if the vector $\nu$ is an integral basis for the quadratic field $\bQ(\gamma)$, or equivalently if the continued fraction of $\omega$ has a golden mean tail, see \App{Farey}. Strong numerical evidence for the local robustness of the nobles for Chirikov's map was given by MacKay and Stark \cite{MacKay92c}.
In \Sec{Nobles}, we will study the relative robustness of circles with rotation numbers from different quadratic fields for several examples of \Eq{StdMap}.

The local robustness conjecture implies that the critical function $\eps_{cr}(\omega)$ will have local maxima at each noble rotation number. The global maximum of this function for Chirikov's map occurs, according to Greene, at the golden mean. More generally, let $\omega_{max}$ be location of the global maximum of $\eps_{cr}$, i.e., the rotation number of the globally most robust circle. We will investigate the dependence of $\omega_{max}$ on  the force $g$.  

\begin{con}\label{con:Piecewise}
The rotation number of the globally most robust invariant circle is generically piecewise constant under continuous changes of \Eq{StdMap}.
\end{con}
\noindent
We are not aware of previous studies of this conjecture. Our preliminary evidence consists of studying two-parameter families $f_{\eps,\psi}$, obtained by introducing an additional parameter $\psi$ into $g$. The critical function will of course depend upon $\psi$ and if it were a typical, smooth function, the location of its global maximum would be a piecewise continuous function of $\psi$ with occasional jump discontinuities as different local maxima take over as the global maximum. As we will see in \Sec{Critical}, $\omega_{max}(\psi)$ seems instead to be locked to fixed noble rotation numbers over intervals of $\psi$. This conjecture does not hold when the map has special symmetries that cause $\eps_{cr}(\omega,\psi)$ to be a discontinuous function of $\psi$ for fixed $\omega$.

Each of these studies requires a technique for computing invariant circles of the generalized standard map that does not rely upon finding symmetric periodic orbits.
We will use, see \Sec{FourierMethod}, a method introduced by de la Llave and collaborators \cite{delaLlave05} to compute the Fourier series for the embedding $k: \bS \to M$ that conjugates the dynamics of $f_\eps$ on an invariant circle to the rigid rotation, $T_\omega: \bS \to \bS$,
\beq{RigidRotation}
	T_\omega(\theta) =\theta + \omega .
\eeq
This technique gives a numerically robust method for computing subcritical invariant circles. We will then estimate the critical value by continuation in $\eps$ and extrapolating to the critical point by estimating where the circle loses smoothness. In particular, a Sobolev norm of the conjugacy is singular at criticality \cite{Apte05, Olvera08,Calleja10a, Calleja10b}. We will see in \Sec{FindEpsCr} that a certain Sobolev seminorm has a power law blowup as $\eps \to \eps_{cr}$ that allows accurate estimates of the critical value.

\section{Computing the Conjugacy for an Invariant Circle}\label{sec:FourierMethod}

In this section we briefly recall an efficient algorithm to compute the embedding for invariant tori of maps, adapted from one developed by \cite{delaLlave05,Fontich09}. We specialize to the case that $f$ is an area-preserving map and the torus is a rotational circle, though much of the analysis carries over to higher-dimensional symplectic maps.

%

Suppose that $f$ is a real analytic, area-preserving map of the cylinder 
$M = \bS \times \bR$. We assume that $f$ has an analytic, rotational invariant circle $\cC \subset M$ on which the dynamics is conjugate to rigid rotation with Diophantine rotation number $\omega$, that is, we suppose there exists an analytic embedding $k:\bS \to M$ such that $k(\bS) = \cC$, and 
\beq{conjugacy}
	f \circ k=k \circ T_{\omega},
\eeq
where $T_\omega$ is the rigid rotation \Eq{RigidRotation}, see \Fig{reducibility}.

\InsertFig{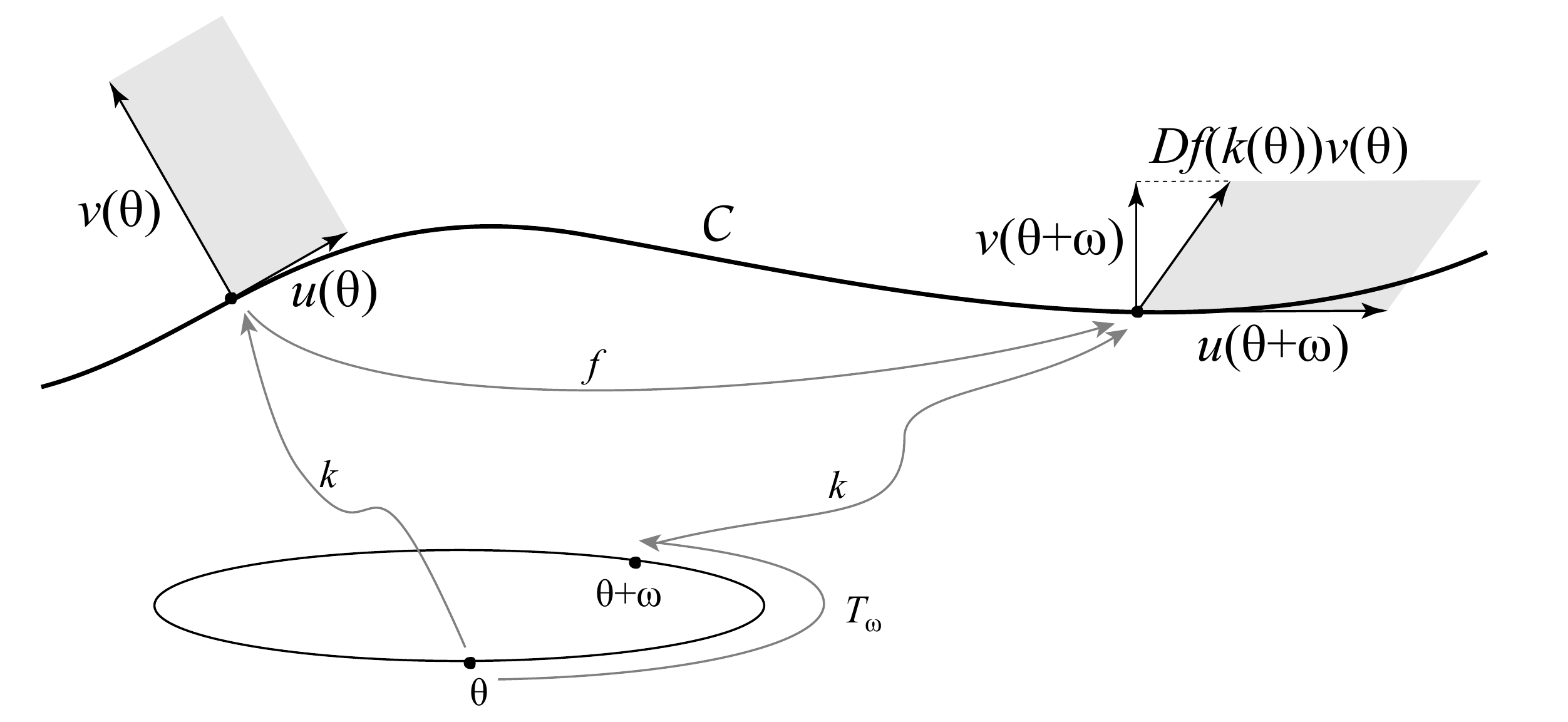}{A visualization of the conjugacy \Eq{conjugacy} and of automatic reducibility. Here the $u$ is the tangent vector to the embedded circle $\cC$, $v$ is a normal, and the shaded regions have unit area.}{reducibility}{5in}

Lifting the map $f$ to the universal cover of $M$, the condition that $\cC$ be rotational is
\[
	k(\theta+1)=k(\theta)+(1,0)^T,
\]
i.e., the angle component of $k$ is a degree-one circle map (called the ``hull" function in \cite{Aubry83a}) and the action component is periodic. We will say, in this case, that $k$ has \emph{degree-one}.
Since $k$ is analytic, its periodic part has a convergent Fourier series, and so it can be represented as
\beq{KFourier}
	k(\theta) = \begin{pmatrix} \theta \\ 0 \end{pmatrix}
				+ \sum_{j\in \mathbb{Z}} \hat{k}_{j} e^{2\pi i j \theta} ,
\eeq
with coefficients $\hat{k}_j = \hat{k}^*_{-j} \in \bC^2$.

Note that solutions of \Eq{conjugacy}, if they exist, are not unique: 
given a solution $k(\theta)$, then $k(\theta +\phi)$ is also a solution 
$\forall \phi \in \bR$. However when $\omega$ is irrational, continuous conjugacies are otherwise unique, apart from this shift in the origin of $\theta$. 

\begin{lem}[\cite{Apte05}]\label{lem:Uniqueness}
If $k \in C^0(\bS,M)$ solves \Eq{conjugacy} for an irrational rotation number $\omega$, then every other continuous solution of \Eq{conjugacy} for the same invariant circle is of the form $k(\theta+\phi)$ for some $\phi \in \bS$.
\end{lem}

\begin{proof}
Indeed, if $\tilde k$ is another continuous conjugacy, then, since $\tilde k(\bS) = k(\bS) = \cC$, there exists a $\phi$ such that, e.g., $\tilde k(0) = k(\phi)$, and by \Eq{conjugacy}, this implies $\tilde k(n\omega) = k(\phi+n\omega)$, $\forall n \in \bZ$. Thus $\tilde k(\theta) = k(\theta+\phi)$ on a dense set in $\bS$, and by continuity they agree everywhere.
\end{proof}

\subsection{Automatic Reducibility}\label{sec:reducibility}

Under the assumption that there exists an invariant circle with the given $\omega$, de la Llave and collaborators \cite{delaLlave05,Fontich09} developed an iterative, quasi-Newton scheme to find the conjugacy $k$. The algorithm is initialized with a guess $k$ such that 
\beq{conj1}
	f\circ k-k \circ T_{\omega}=e ,
\eeq
and is guaranteed to converge provided that the error, $e$, is sufficiently small. The iteration proceeds by inserting a corrected $k \to k + \Delta$ into \Eq{conjugacy} and expanding to give
\[
	f(k(\theta)) + Df(k(\theta))\Delta(\theta) -k(\theta+\omega) 
		-\Delta(\theta+\omega) = \mathcal{O}(\Delta^2) .
\]
Neglecting the second-order terms and using \Eq{conj1} gives the iterative equation
\beq{Nstep} 
	\Delta(\theta+\omega) - Df(k(\theta)) \Delta(\theta) = e(\theta) ,
\eeq
that can be viewed as determining $\Delta$. The resulting function $k+\Delta$ is then an approximate conjugacy in the sense of satisfying \Eq{conj1} with a new, presumably smaller, error $e$.

Computing $\Delta$ by a direct inversion of the \emph{cohomology} operator on the left hand side of \Eq{Nstep} is numerically expensive. However, this linear operator can be partially diagonalized through a process called \emph{automatic reducibility} by \cite{delaLlave05}. The idea is that there exists a change of variables $\Delta(\theta) = M(\theta)w(\theta)$ where $M(\theta)$ is a matrix with orthogonal columns, and such that for the new vector $w(\theta)$, 
\Eq{Nstep} takes the form
\beq{AutRed}
	w(\theta+\omega) - U(\theta)w(\theta) = M^{-1}(\theta+ \omega)e(\theta) 
	 \equiv \tilde e(\theta),
\eeq
where $U(\theta)$ is a special upper-triangular matrix. 

To find $U$ and $M$, one must solve the matrix system
\beq{reducingMatrix}
	Df(k(\theta)) M(\theta) = M(\theta+\omega) U(\theta) .
\eeq
To solve \Eq{reducingMatrix} the columns of $M$ are chosen to be tangent and normal vector fields of the (approximate) circle $\cC$, recall \Fig{reducibility}. Note that if $k$ solves \Eq{conjugacy}, then differentiation implies
\beq{tangent}
	Df(k(\theta))Dk(\theta)=Dk(\theta+\omega) ,
\eeq
which is the statement that the tangent vector field, $Dk$, to $\cC$ is invariant under $f$. Since the function $k$ in \Eq{Nstep} will never be an exact conjugacy, $Dk$ will only approximately solve \Eq{tangent}. Nevertheless, we may use \Eq{tangent} in the Newton iteration \Eq{Nstep} incurring only error at second order, see \cite[Prop.~7]{delaLlave05}.
The matrix $M$ is now chosen to be
\[
	M(\theta)= \begin{pmatrix} u(\theta) &v(\theta) \end{pmatrix} 
\]
where $u$ and $v$ are tangent and (scaled) normal vector fields
\bsplit{unitVectors}
	u(\theta) &= Dk(\theta) , \\ 
	v(\theta) &=\frac{1}{\|u(\theta)\|^2}Ju(\theta), 
\esplit
and $J = \begin{pmatrix} 0 & -1 \\ 1 & 0\end{pmatrix}$. The normalization implies that the rectangle formed from $u$ and $v$, see \Fig{reducibility}, has unit area, since $\det M(\theta) = (Ju)^T v =  1$.

With this choice, \Eq{reducingMatrix} becomes
\beq{reduc1}
	Df(k(\theta))M(\theta) =
	\begin{pmatrix} Df(k(\theta))u(\theta) &Df(k(\theta))v(\theta) \end{pmatrix} = 
	 \begin{pmatrix}	
	 	u(\theta + \omega) & v(\theta + \omega) 
	\end{pmatrix}
	\begin{pmatrix}
			1 & a(\theta) \\
			0 & 1 \\
	\end{pmatrix} ,
\eeq
which must be solved for the off-diagonal term of $U$, the function $a(\theta)$.
Approximate invariance of the tangent vector field, \Eq{tangent}, implies that the first column of \Eq{reduc1} is an identity, at least to second order. The second column
gives
\beq{Aeq1}
	a(\theta)u(\theta+\omega)+ v(\theta+\omega)=Df(k(\theta))v(\theta) .
\eeq
Taking the inner product of this equation with $u(\theta+\omega)$ determines $a$:
\beq{Aeq}
	a(\theta)=\frac{u(\theta+\omega)^T Df(k(\theta))v(\theta)}{ \|u(\theta+\omega)\|^2}
	   = v^T(\theta+\omega) JDf(k(\theta)) v(\theta).
\eeq
Similarly, taking the inner project of \Eq{Aeq1} with $Ju(\theta+\omega)$ and using the definition \Eq{unitVectors}
gives the consistency condition
\begin{align*}
	1 &= (Ju(\theta+\omega))^T Df(k(\theta)) v(\theta) 
	 = \left(Df(k(\theta))^T J u(\theta+\omega)\right)^T v(\theta) \\
	 &= \left(J (Df(k(\theta)))^{-1} u(\theta+\omega)\right)^T v(\theta) 
	 = (Ju(\theta))^T v(\theta) = 1 ,
\end{align*}
where we have used the symplectic condition $Df^T J = J Df^{-1}$ and
\Eq{tangent}. Simply put, this condition states that since area is preserved, the image of the original rectangle also has unit area, see \Fig{reducibility}.

%

The two rows of \Eq{AutRed} now yield skew coupled equations for the components of the vector $w$:
\begin{eqnarray}
 	w_1(\theta+\omega) -w_1(\theta) &=& \tilde{e}_1(\theta)+a(\theta)w_2(\theta) ,\label{eq:w1} \\
	w_2(\theta+\omega)- w_2(\theta)&=& \tilde{e}_2(\theta) \label{eq:w2},
\end{eqnarray}
where $a$ is defined by \Eq{Aeq}, and $\tilde e$ by \Eq{AutRed}.

These two equations can be solved easily in Fourier space. 
Indeed, each is of the form of a cohomology equation,
\[
	w\circ T_\omega -w=e ,
\]
that is diagonalized by Fourier transformation. Indeed, when $e$ is analytic and $\omega$ is Diophantine \Eq{Diophantine}, then $w$ is analytic \cite{Moser66} and its Fourier coefficients are
\beq{CohoSoln}
	\hat{w}_j=\frac{\hat{e}_j}{e^{2\pi i j\omega}-1} , \; j \neq 0 ,
\eeq 
provided that $e$ satisfies the solvability condition
 \[
 	\hat{e}_0 = \int_0^1 e(\theta) d\theta = 0 , 
\]
i.e., that its average vanish.
Since the average, $\hat w_0$, is in the kernel of the cohomology operator, it can be chosen freely.

\subsection{Summary of the Algorithm}\label{sec:Algorithm}

Beginning with a guess for the conjugacy $k$, we compute the error $e$ from \Eq{conj1}, and the vector fields $u$ and $v$ from \Eq{unitVectors}. Now the transformed error $\tilde e$ can be computed from \Eq{AutRed} and $a$ from \Eq{Aeq}. At this point the cohomology equation \Eq{w2} can be solved for $w_2$ using \Eq{CohoSoln}, under the assumption that the average of $\tilde e_2$ vanishes. Even though this assumption is not generally true for an approximate conjugacy, ignoring it does not interfere with convergence of the method; indeed as $\tilde e \to 0$, so does the error induced by this inconsistency \cite{Fontich09}.
The arbitrary coefficient $\hat w_{20}$ can be selected so that the average of the right hand side of \Eq{w1} is zero. At this point \Eq{w1} is consistent and can be solved for $w_1$. The average of $w_1$ is arbitrary; indeed, since this gives a contribution $u(\theta) \hat w_{10}$ to $k$, it contributes to a shift along the circle, which corresponds to the non-uniqueness of the solution, recall \Lem{Uniqueness}. We set $\hat w_{10} = 0$ for simplicity. Once $w$ is found, $k$ can be updated using $k \to k+M(\theta)w(\theta)$. This process is repeated until the $L^2$-norm of the error \Eq{conj1} is satisfactorily low. 

This algorithm is generally robust and invariant circles can be computed for moderate values of $\eps$ and many Diophantine $\omega$. In some cases the algorithm fails well before $\eps_{cr}$, even upon using extrapolation in $\eps$ to update the initial guess. This failure is often the result of aliasing in the Fourier spectrum. Aliasing occurs whenever one computes the discrete Fourier transform of a nonlinear function, such as $g \circ k_1$, of a finite series approximation to $k$. This error can be ameliorated by application of an anti-aliasing filter \cite{Trefethen00}. We use the quadratic filter 
\[
	\hat k_j \to \frac{\hat k_j}{j^2-J^2} , \quad |j| > J
\]
to scale the amplitudes of the Fourier coefficients beyond a threshold $J$. For the twist case, $\Omega_1(y)$,
we set $J = \tfrac14 N$ for $N$ Fourier modes, but we set $J = \tfrac18 N$ for $\Omega_2$ and other nonlinear frequency maps to help compensate for the additional nonlinearity. This filter is applied to $k$ once per iteration, at the beginning of the Newton step, prior to the computation of the error $e$ with \Eq{conj1}.

To compute invariant circles for large $\eps$, we use continuation from $\eps = 0$, where $k$ is trivial. We begin with the increment $\Delta\eps = 0.01$ and with $N = 2^8$ Fourier modes. The algorithm is deemed to converge when the $L^2$-norm of the error $e$ is less than the tolerance $10^{-12}$. An initial guess for the conjugacy at $\eps +\Delta\eps$ is obtained by linear extrapolation. This is continued until the algorithm fails to converge within the specified tolerance. Upon this first failure, the step size is reduced to $\Delta \eps = 0.0005$ and the number of Fourier modes is doubled. At each successive failure the number of Fourier modes is again doubled; however, we found it significantly faster and more accurate to keep the step size constant.
For reasonable accuracy and speed, we let the algorithm exit when the number of Fourier modes exceeds $2^{13}$ (conjugacies with millions of Fourier modes have been computed \cite{Apte05, Olvera08}). This failure occurs below the critical value, $\eps_{cr}$, at which the invariant circle is first destroyed. Indeed, a comparison with Greene's criterion (see \App{Greene}) for Chirikov's map, as we discuss in \Sec{FindEpsCr}, shows that convergence typically fails when $\eps \sim \eps_{cr} - 10^{-3}$.

\section{Detecting Critical Circles}\label{sec:FindEpsCr}
The quasi-Newton scheme discussed above provides a method to compute an analytic invariant circle when it exists. We will use the scheme of \Sec{FourierMethod} to estimate the value $\eps_{cr}(\omega)$ such that there is an invariant circle with rotation number $\omega$ for all $0 \le \eps \le \eps_{cr}(\omega)$. As in \Con{Cantorus}, we expect that when $\eps$ exceeds $\eps_{cr}$ the invariant circle is replaced by a cantorus. 

Since the conjugacy loses continuity at $\eps_{cr}(\omega)$, this transition can be detected by the blowup of a Sobolev norm \cite{Apte05, Olvera08, Calleja10a, Calleja10b}.
To detect the transition, we use a related seminorm defined as the $L^2$-norm of the $m^{th}$ derivative:
\beq{Sobolev}
	\|k\|_m^2 \equiv \|D^m k\|_{L^{2}}^2=\sum_j(2\pi |j|)^{2m} |\hat{k}_j |^2 ,
\eeq
where $\hat{k}$ are the Fourier coefficients of $k$. As $\eps \to \eps_{cr}$, $\|k\|_m \to \infty$, and we assume it does so asymptotically as
\beq{AsyNorm}
	\|k\|_m \sim \frac{A}{(\eps_{cr}-\eps)^b} .
\eeq
To compute the three parameters in the asymptotic form \Eq{AsyNorm}, we use three consecutive $(\eps,\|k\|_m)$ pairs from the $\Delta\eps$ steps of the continuation method. One reason for using a fixed $\Delta\eps$, as described in \Sec{Algorithm}, is that we find the estimate of the pole position to be more accurate than if a variable step size were used. We typically compute \Eq{Sobolev} for the angle component of $k$, though the action component gave similar results. 

Numerical computations for the rotation number $\gamma^{-2}$ and several examples of \Eq{StdMap} indicate that \Eq{AsyNorm} is a good fit for $m=2$ when $\eps_{cr} -\eps \lesssim 10^{-2}$, see \Fig{fitError}(a). More generally we find that \Eq{AsyNorm} applies for other values of $m$ with $b \approx m-1$.


The choice of $m$ significantly affects the error in the approximation of $\eps_{cr}$. In 
\Fig{fitError}(b) we compare the estimates of $\eps_{cr}$ from \Eq{AsyNorm} for Chirikov's map with those obtained numerically from Greene's criterion (recall \App{Greene}). For the invariant circles with rotation numbers $\gamma$ and $\sqrt{2}$, and using orbits of periods up to $30,000$, Greene's method gives
\bsplit{GoldenEpsCR}
	\eps_{cr}(\gamma) 	&= 0.9716353(1) ,\\
	\eps_{cr}(\sqrt{2}) 	&= 0.957447(6) , 
\esplit
where the error (the number in parentheses is the uncertainty in the last digit) is estimated from the extrapolation.\footnote
{Quadruple precision computations give $\eps_{cr}(\gamma) = 0.9716354063(2)$ \cite[\S 4.4]{MacKay93}.}
We observe that the smallest numerical error for the seminorm fit occurs with the choice $m = 2$. Fixing the maximal number of Fourier modes to $2^{13}$, mainly for computational speed, we typically obtain $\eps_{cr}$ with an error less than $10^{-4}$ (more details will be given in \Sec{Nobles}).

\InsertFigTwo{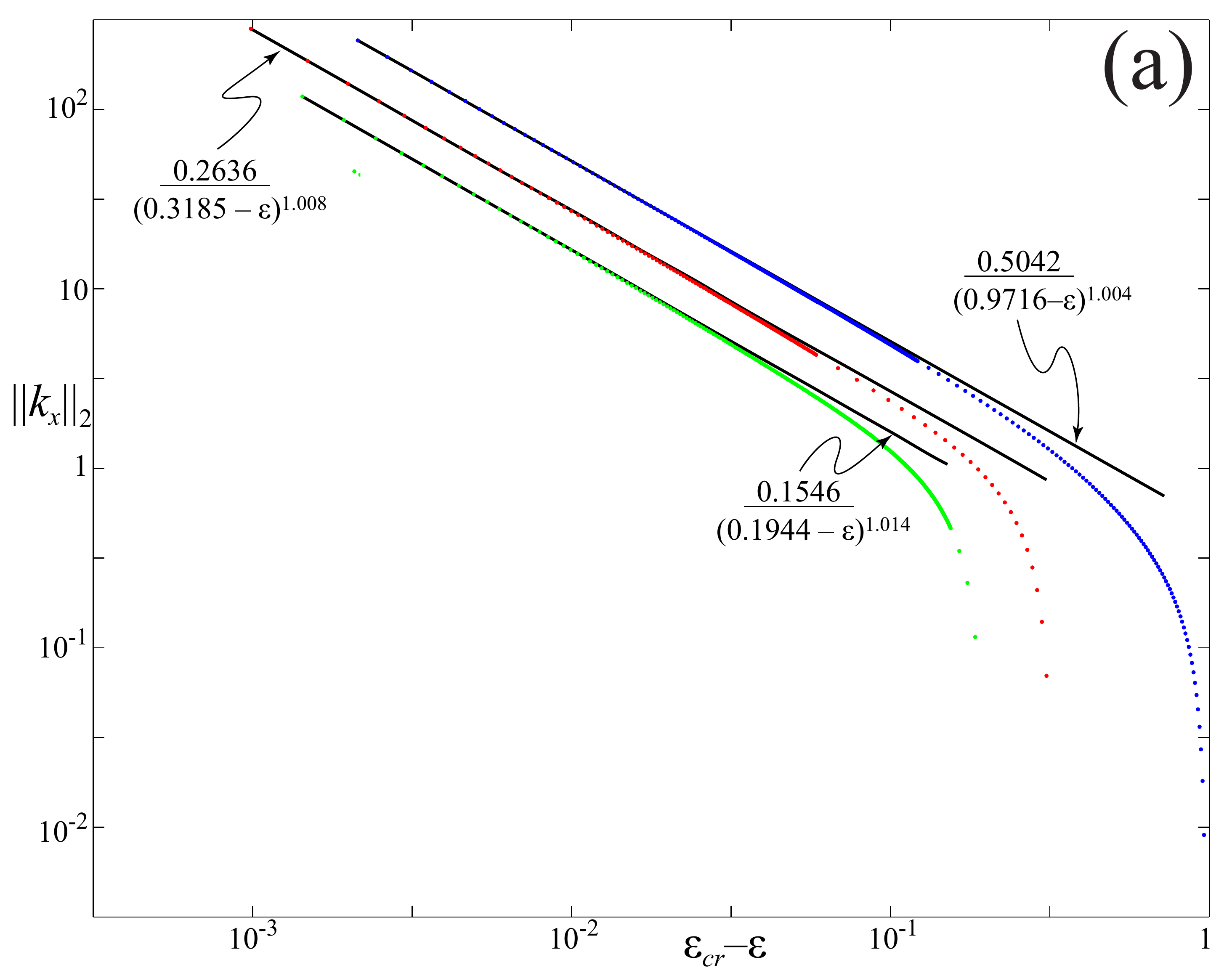}{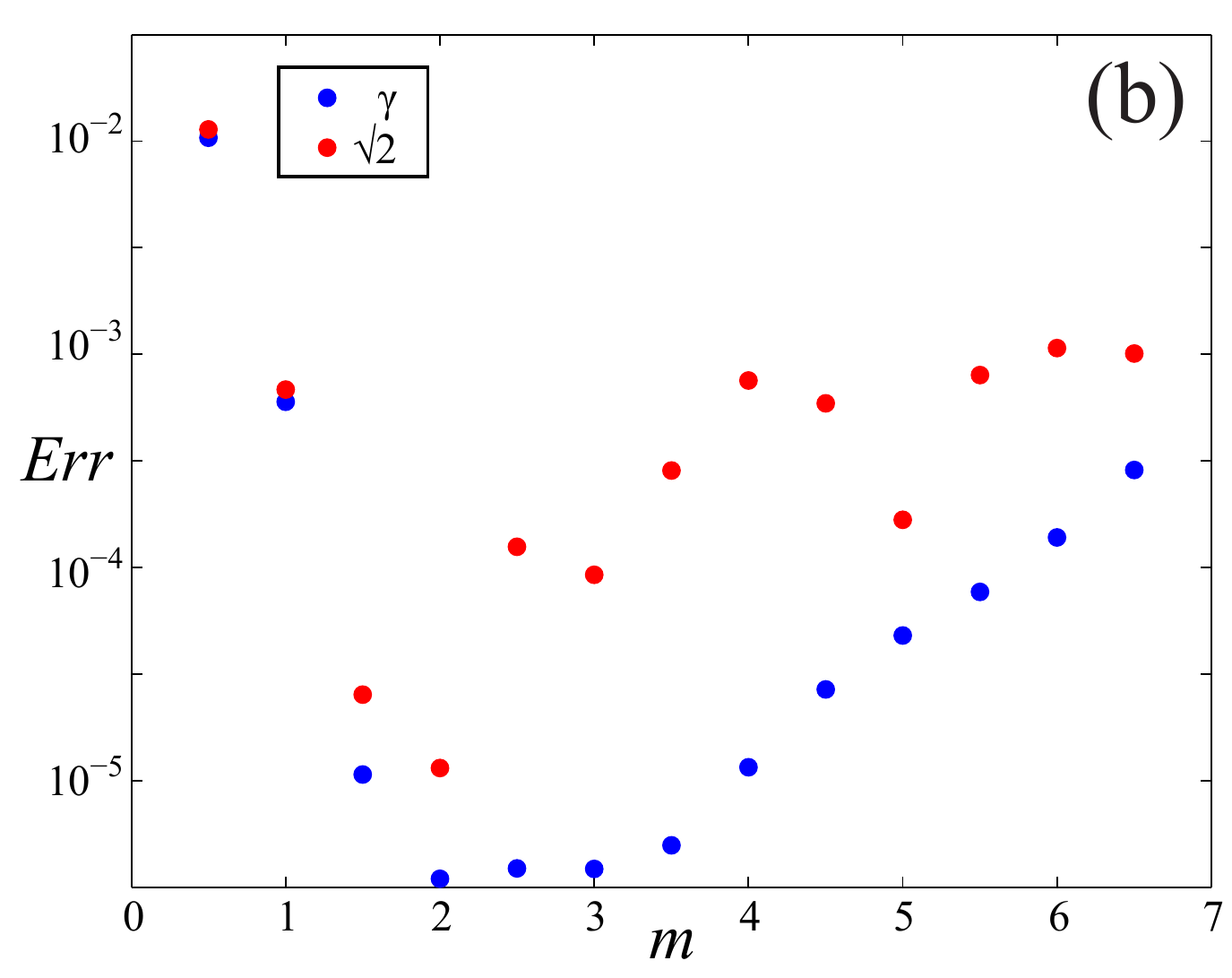}{ (a) Blowup of the seminorm $\|k_x\|_2$ for the $\omega = \gamma^{-2}$ invariant circle for three cases of \Eq{StdMap}. The upper (blue) curve corresponds to Chirikov's standard map. The lower two curves have the force \Eq{G3} with $\psi = \pi/4$. The middle (red) curve is for the frequency map \Eq{Twist} and the bottom (green) is for \Eq{NonTwist} with $\delta = 0.3$. In each case the horizontal axis is logarithmic based on the best estimate of $\eps_{cr}$ from \Eq{AsyNorm}. (b) Error in the pole location for fits to \Eq{AsyNorm} for seminorms with varying $m$ for Chirikov's map for two rotation numbers, compared to $\eps_{cr}$ from Greene's criterion. For both rotation numbers, the smallest error occurs with $m=2$.}{fitError}{3in}

When the invariant circle is not symmetric, we cannot compare with the results from Greene's residue criterion (a periodic orbit finder would require a two-dimensional search and be much less well-behaved). Nevertheless, we select $m=2$ for the computations in the rest of this paper. We estimate the error in the calculation of $\eps_{cr}$ as the difference between the last two, three-step approximations obtained with the step size $\Delta \eps = 0.0005$. Note, however, that this occasionally underestimates the true error in $\eps_{cr}$ by a factor of $2-8$.

\section{Conjugacies of Near-Critical Circles}\label{sec:CriticalCircle}
In this section we compute the conjugacy \Eq{conjugacy} for circles with golden mean rotation number for a variety of maps of the form \Eq{StdMap}. We use the seminorm technique to determine when these circles are destroyed, comparing to Greene's criterion when the map is reversible, but also generalizing to maps that do not have twist or reversing symmetries. 

We begin by computing the embedding of the oft-studied golden mean invariant circle for Chirikov's map, i.e., for $g=g_1$, \Eq{Chirikov}, and $\Omega = \Omega_1$, \Eq{Twist}. 
The algorithm of \Sec{FourierMethod} converges up to $\eps = 0.9695$, which from \Eq{GoldenEpsCR} is $\eps_{cr}(\gamma) - 0.0021$. The resulting embedding at this parameter value, as shown in \Fig{StdMapEmbed}, has error $e(\theta)$, \Eq{conj1}, with $L^2$-norm of $(10)^{-12}$ and $L^\infty$-norm of $1.7(10)^{-11}$.

\InsertFigTwo{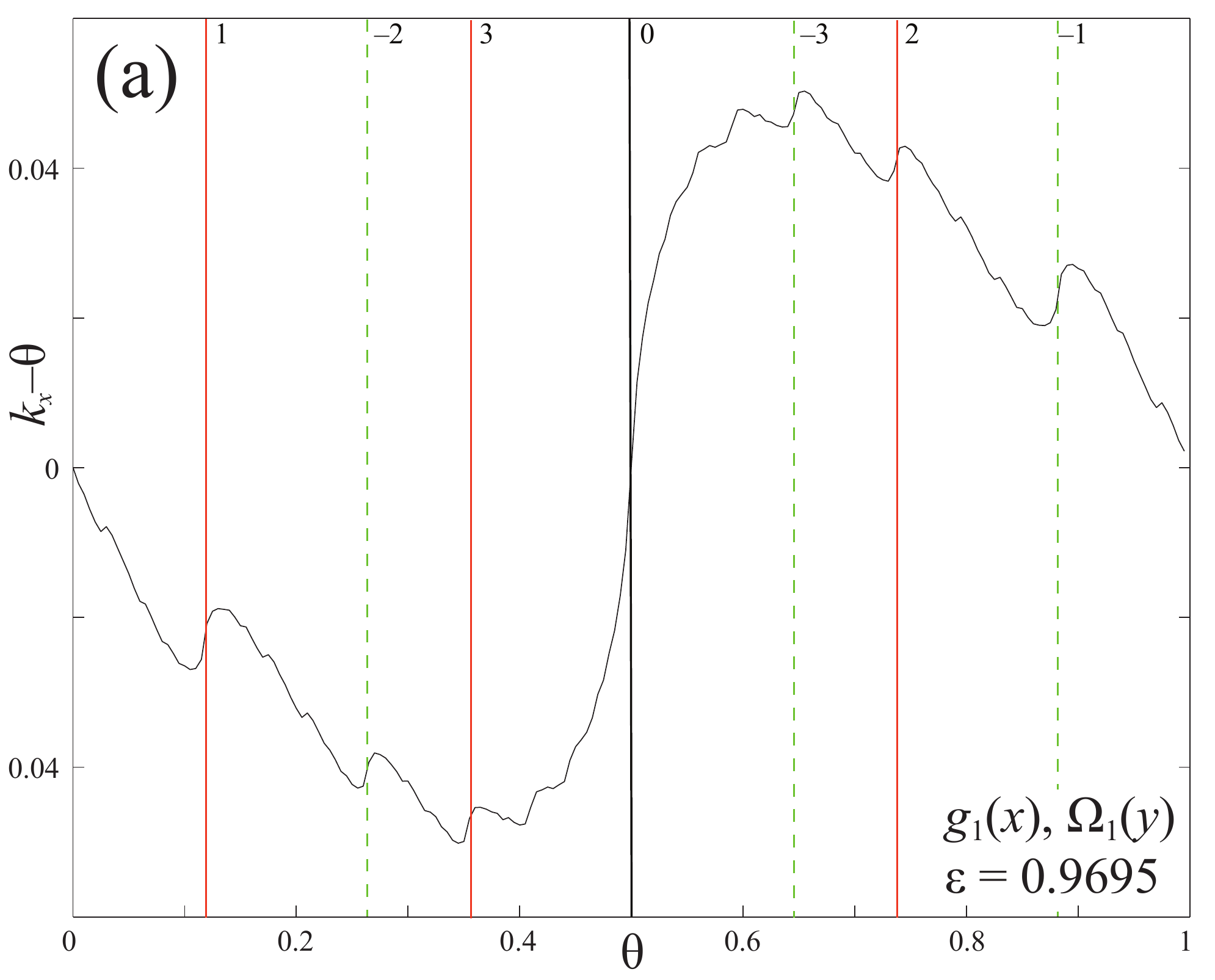}{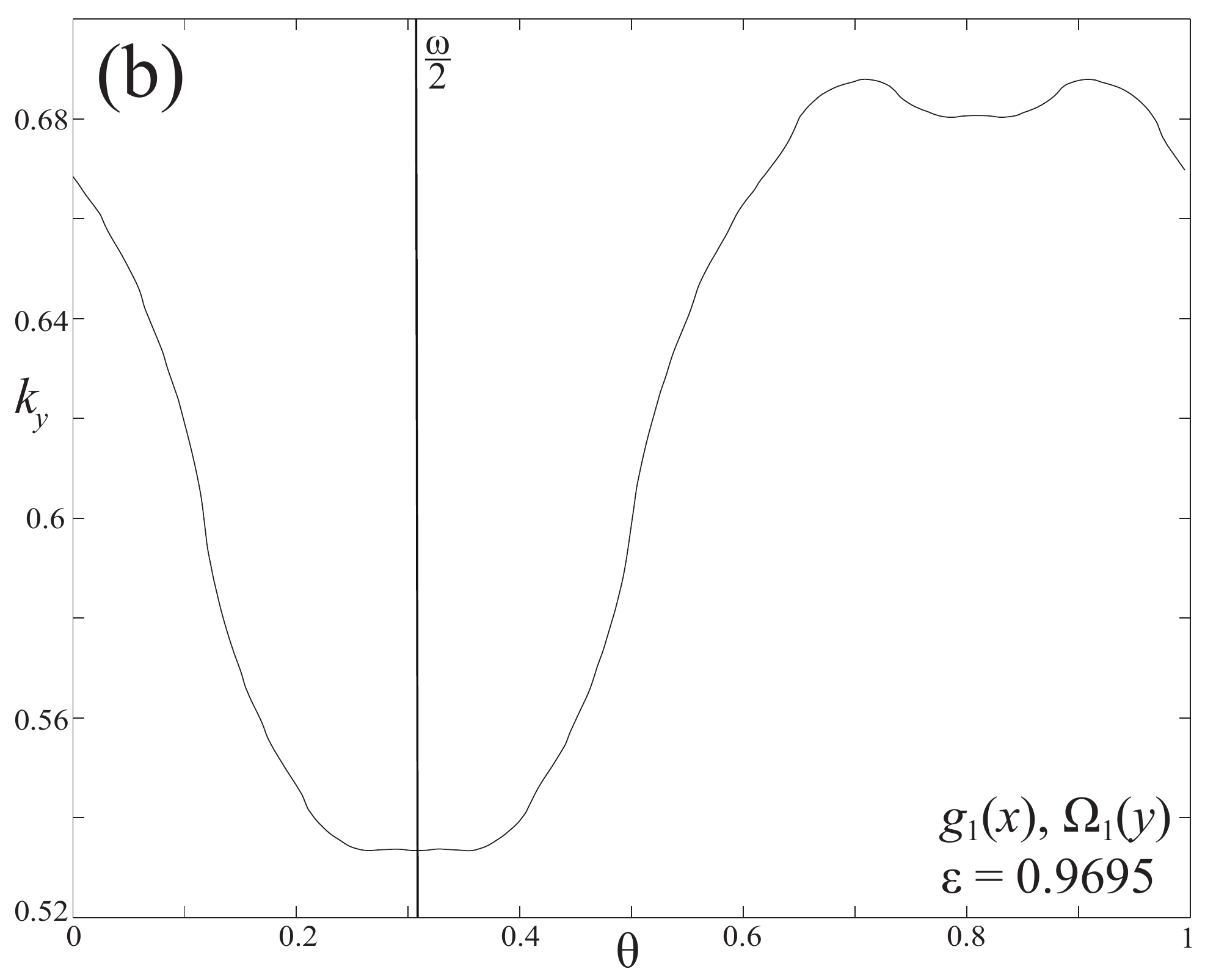}{Components of the conjugacy for the $\omega = \gamma^{-1}$ invariant circle of Chirikov's standard map from a computation with $2^{13}$ Fourier modes. (a) The function $k_x(\theta)-\theta$ and the first three forward and backward images of $\theta = \tfrac12$, indicated by the solid (red) and dotted (green) lines, respectively. (b) The action component $k_y(\theta)$ is even about $\omega/2$. }{StdMapEmbed}{3.2in}

Since Chirikov's map has twist, Aubry-Mather theory implies that each invariant circle becomes a cantorus upon destruction. Moreover, these sets are ordered in the sense that $\theta \mapsto k_x(\theta)$ is a monotone circle map. The implication is that upon destruction of the circle, the extension of the conjugacy to the Cantor set is the inverse of a devil's staircase; consequently, $k_x(\theta)-\theta$ must develop positive jump discontinuities at a dense set of $\theta$ values corresponding to the location of the gaps in the cantorus. Although the Fourier method does not converge for the super-critical case, the beginning of this metamorphosis can be seen as $\eps \to \eps_{cr}$. Indeed the largest, incipient discontinuity seen in \Fig{StdMapEmbed}(a) occurs at $\theta = \tfrac12$, which, since the computations give $k_x(\tfrac12) \approx \tfrac12$, also corresponds to $x \approx \tfrac12$. The remaining localized regions with large slope occur along the orbit $\theta_t = \tfrac12 + t\omega$. 

The locations of the gaps can be explained by recalling that an invariant circle corresponds to a minimizing state of the Frenkel-Kontorova energy \Eq{FK}.
The potential for this map, $V(x)=-\tfrac{1}{4\pi^2} \cos(2\pi x)$, has a single maximum at $x=\tfrac12$; thus, the AI theory implies that the largest gap should form around this maximum, and its images form a bi-infinite family (one ``hole") that corresponds to the gaps in the cantorus. The first few forward and backward iterates in this family are indicated by the vertical lines in \Fig{StdMapEmbed}(a). Continuing this process to larger iterates appears to account for all the local peaks in the derivative of $k_x$; thus it appears that this cantorus has one hole \cite{MacKay84}, a fact that can be proven near the AI limit \cite{Baesens94}.

The standard map has two independent families of reversors, see \App{Symmetries}: it is \emph{doubly-reversible}. The first reversor arises from the oddness of $g_1$ and the second from the oddness of $\Omega_1$. As is shown in \Cor{SymConj} in the appendix, the first reversor implies that the function $k_x(\theta)$ is odd about some point $\vphi$; using \Eq{FourierPhase}, we computed $\vphi = -\hat k_{x0} \approx -8.130(10)^{-5}$. Since $\vphi$ is so small, the graph in \Fig{StdMapEmbed}(a) looks like it is odd about $0$. Corollary~\ref{cor:SymConj} also implies that the action $k_y(\theta)$ is even about $\vphi + \tfrac12 \omega$, as is visually apparent in \Fig{StdMapEmbed}{b}. More precisely, we found that the identities \Eq{FourierPhase} are satisfied up to an error with $L^\infty$-norm of $9.9(10)^{-13}$, comparable to the accuracy of the computation of $k$ itself.

A second, doubly-reversible version of \Eq{StdMap} is the generalized, two-harmonic twist map with the force
\beq{G2}
	g_2(x;\psi)=\tfrac{1}{2 \pi} \left(\sin(\psi) \sin(2\pi x) + 
				\cos(\psi) \sin (4\pi x)\right),
\eeq
keeping the frequency map, $\Omega_1$.
The breakup of the golden circle for this map was studied first by Greene et al \cite{Greene87}, and later the formation of cantori and turnstiles was studied in \cite{Ketoja89, Ketoja94, Baesens93, Baesens94, Lomeli06}.\footnote
{In these papers, the parameterization $k_1 = \eps \sin(\psi)$ and $k_2 = -2\eps \cos(\psi)$ was used for $\eps g_2$.}
This two-harmonic map reduces to Chirikov's map when $\psi = \tfrac{\pi}{2}$. The two maps are also equivalent when $\psi=0$ or $\pi$ under the rescaling transformation $(x,y) \mapsto 2 (x,y)$, so that the parameter and rotation number are effectively doubled.
Finally, since $g(x+\tfrac12,-\psi) = g(x,\psi)$, it is sufficient to study the parameter range $\psi \in [0,\pi]$. 

\InsertFigTwo{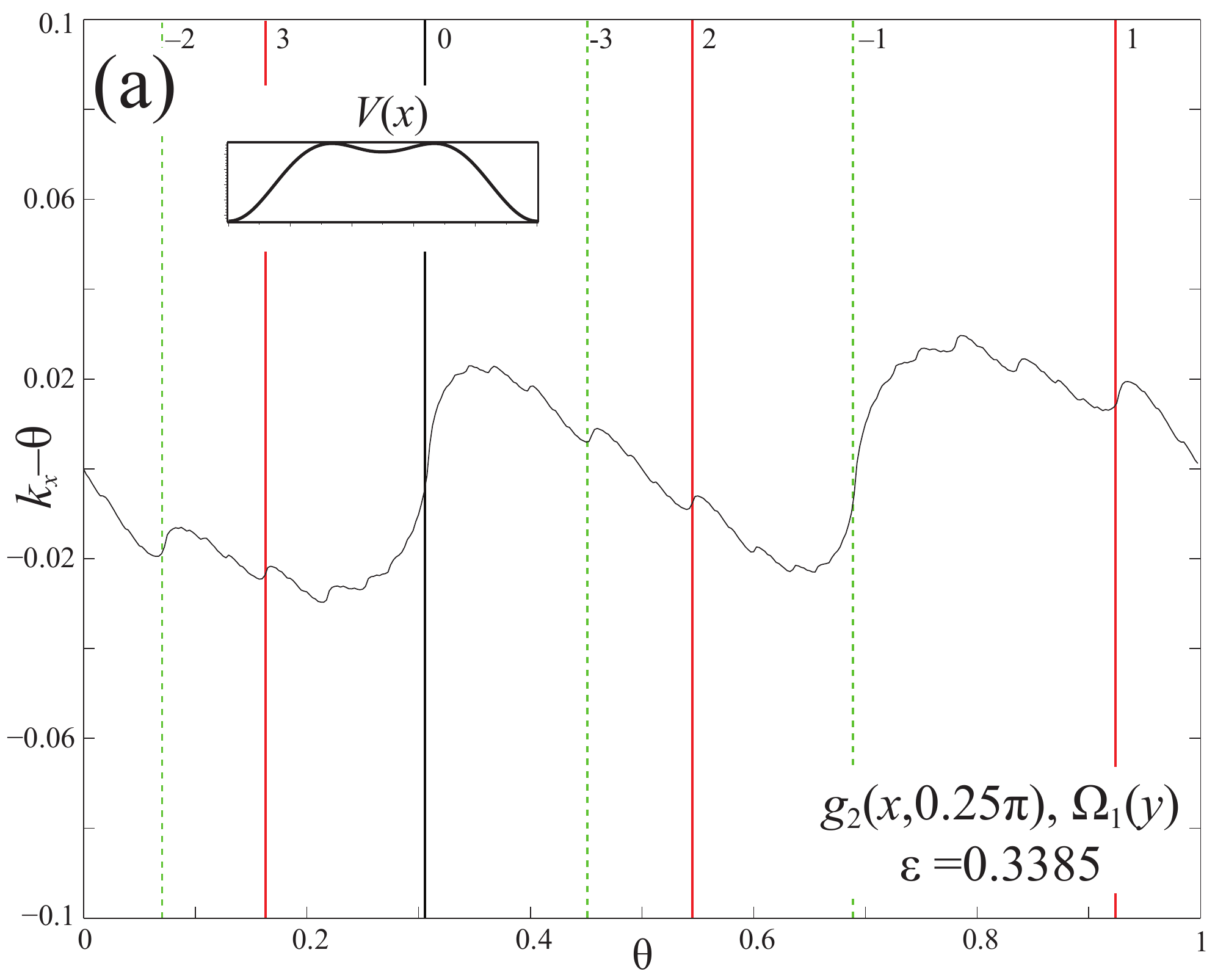}{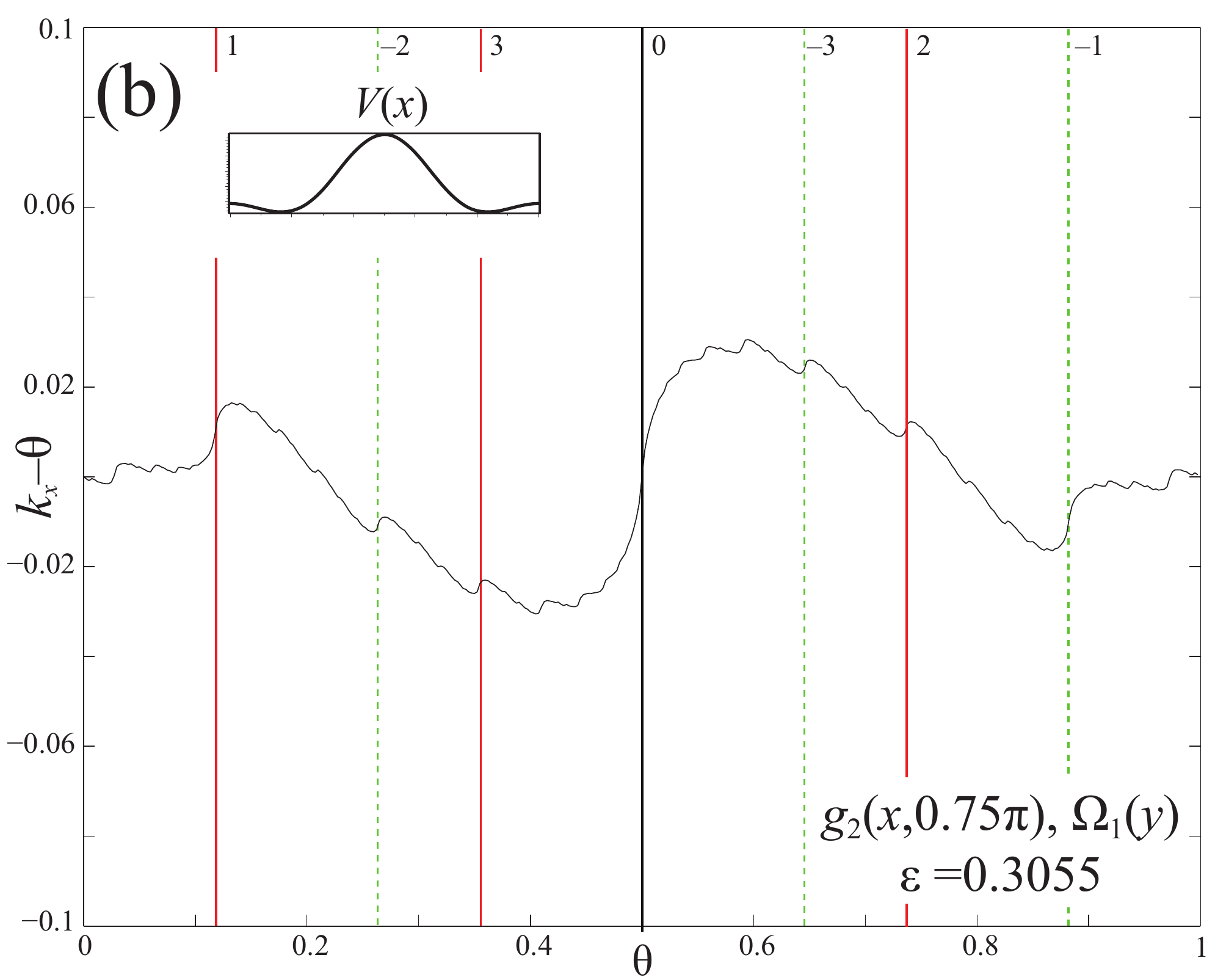}{Plots of $k_x(\theta)-\theta$ for near-critical invariant circles for \Eq{StdMap} with \Eq{Twist} and \Eq{G2} and $\omega = \gamma^{-1}$. Also shown are the first three forward (red, solid) and backward (green, dashed) images of the largest gap. Insets show the shape of the potential $V(x)$ for the two values of $\psi$.}{O1G2TwoPsis}{3.2in}

Two examples of near-critical invariant circles for the two-harmonic twist map are shown in \Fig{O1G2TwoPsis}. Since this map has the same reversors as Chirikov's map, $k_x$ is still odd (since $k_x(0) \approx 0$, essentially about $\theta = 0$). When $\psi$ is small, as in \Fig{O1G2TwoPsis}(a), it appears that there is a symmetric pair of two equally large incipient gaps (near $\theta_1 = 0.3090$ and $ 0.6910$), while for larger $\psi$, \Fig{O1G2TwoPsis}(b), there is a single, largest gap near $\theta = 0.5$. In both cases, it appears that there is a single orbit, $\theta_i + t\omega$, of gaps: for example, two largest peaks in \Fig{O1G2TwoPsis}(a) are separated by one iterate. Thus the resulting cantori will have one ``hole." 

The difference between the two cases is correlated to the structure of the potential $V(x) = \int g_2 dx$---see the insets in \Fig{O1G2TwoPsis}---which has critical points at $x = 0$, $\tfrac12$ and any solutions of $\tan \psi = -2\cos(2\pi x)$. When $0 < \psi < \arctan 2 \approx 0.35\pi$, $V$ has two wells, the deepest at $x=0$, separated by maxima of equal height; this is the case for \Fig{O1G2TwoPsis}(a). The maxima collide at $\tan\psi = 2$ annihilating the well at $x=\tfrac12$. For $|\tan\psi| \ge 2$, $V$ has single well. Finally, when $0.65\pi \approx \pi -\arctan 2 <\psi <\pi$ the potential has two equally deep wells separated by differing maxima, the largest at $x=\tfrac12$, as shown in \Fig{O1G2TwoPsis}(b). The AI limit, $\eps \to \infty$, for a double well potential has infinitely many cantori, selected by the fraction of points that lie in each well \cite{Ketoja94, Baesens94}. For these cases, there are cantori with two families of gaps; however, the second hole is formed beyond the initial breakup of the invariant circle. Note that when $\tan\psi = \pm 2$, the potential has degenerate critical points and persistence of AI states with these points occupied cannot be guaranteed. 

When the maxima of the potential collide at $\tan \psi=2$, we expected that the symmetric pair of largest gaps shown in \Fig{O1G2TwoPsis}(a) would also merge. 
However \Fig{O1G24Frame} shows that this merger does not occur until $\psi \approx 0.46\pi$, well after the collision. Note that as the largest gaps coalesce, they still appear to lie on a single orbit, though the number of iterates between them changes with $\psi$. For example, for $\psi = 0.44\pi$, in \Fig{O1G24Frame}(a), the largest gaps are two iterates apart, but for $\psi = 0.46\pi$, in \Fig{O1G24Frame}(c), they are five iterates apart.

\InsertFigFour{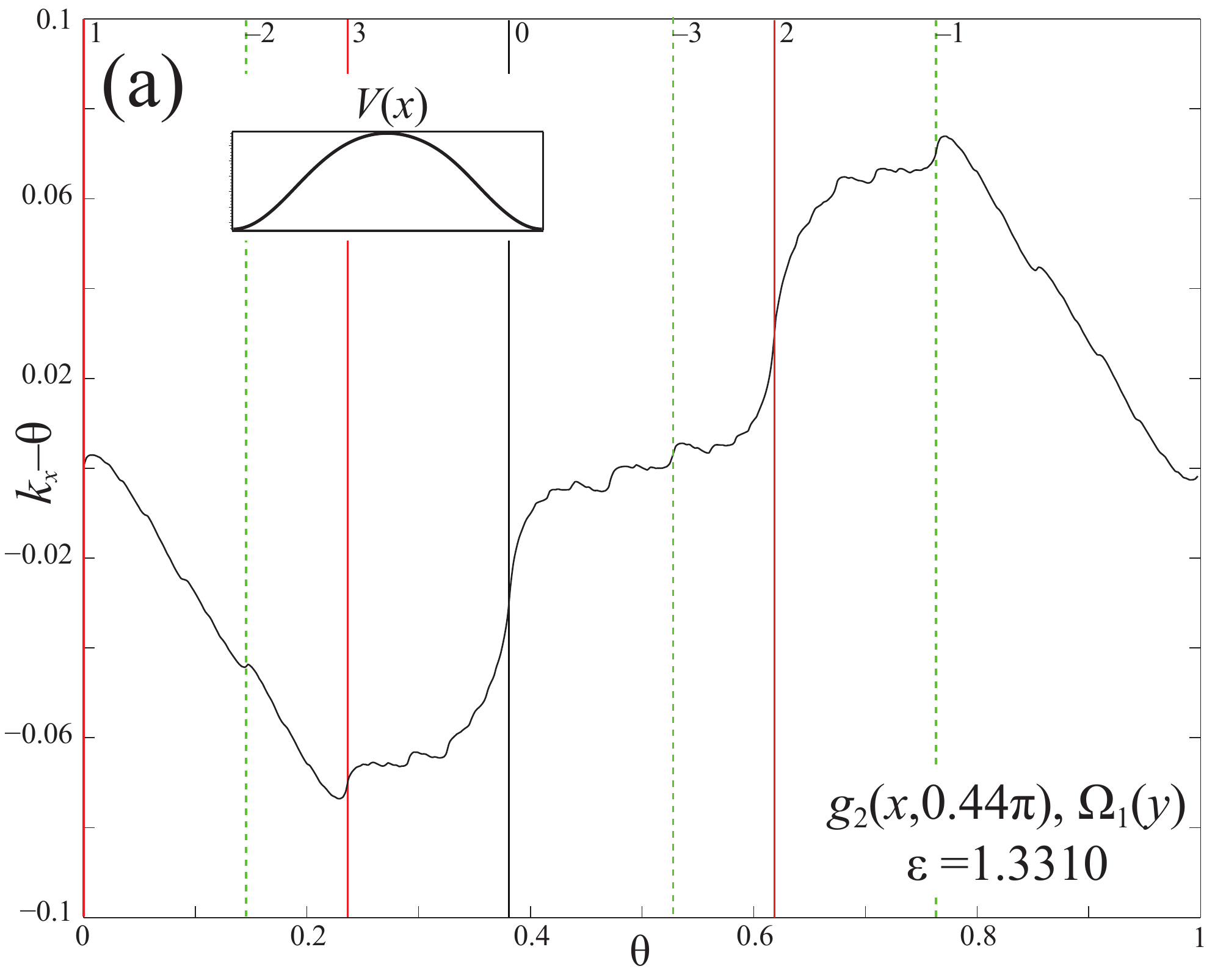}{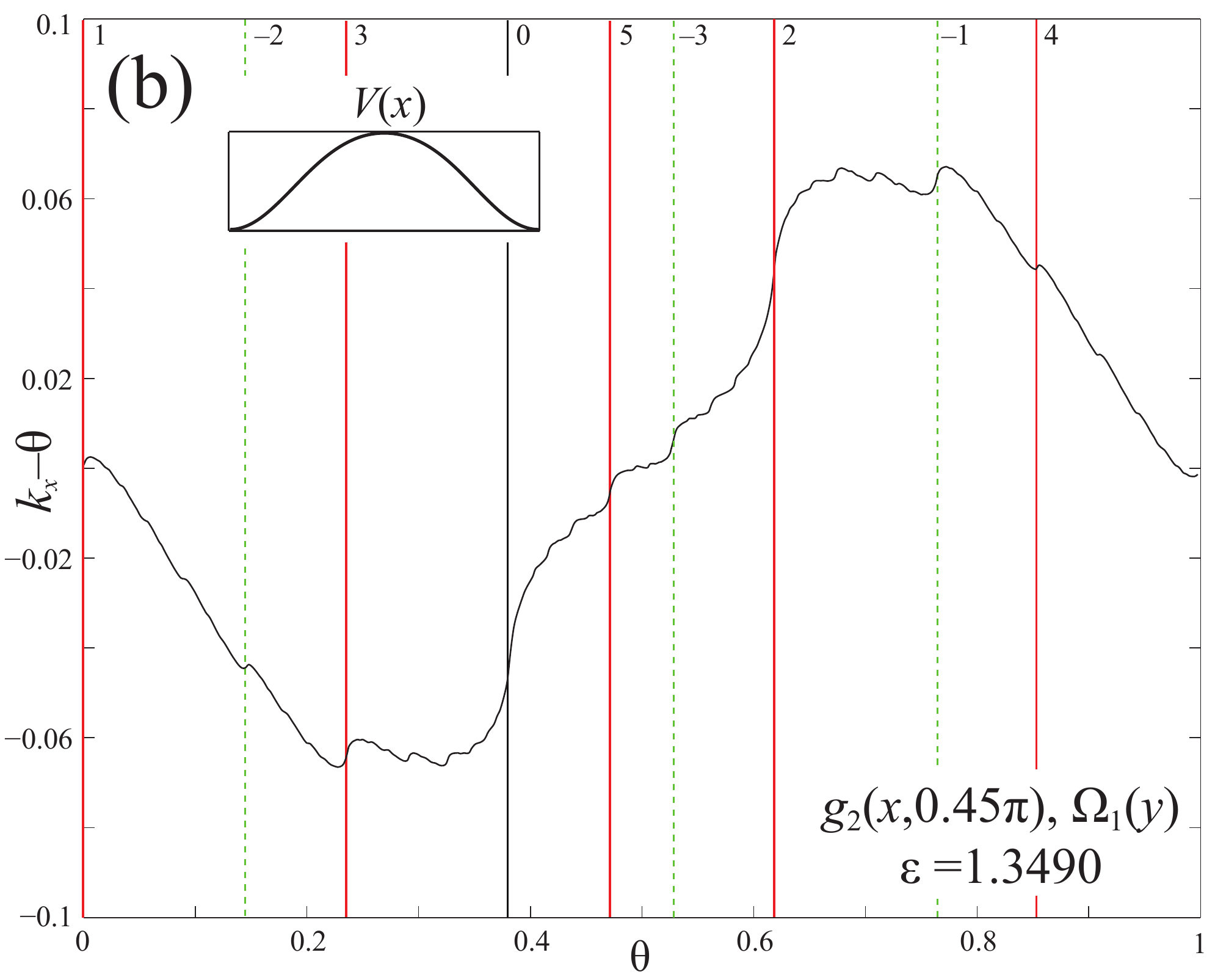}{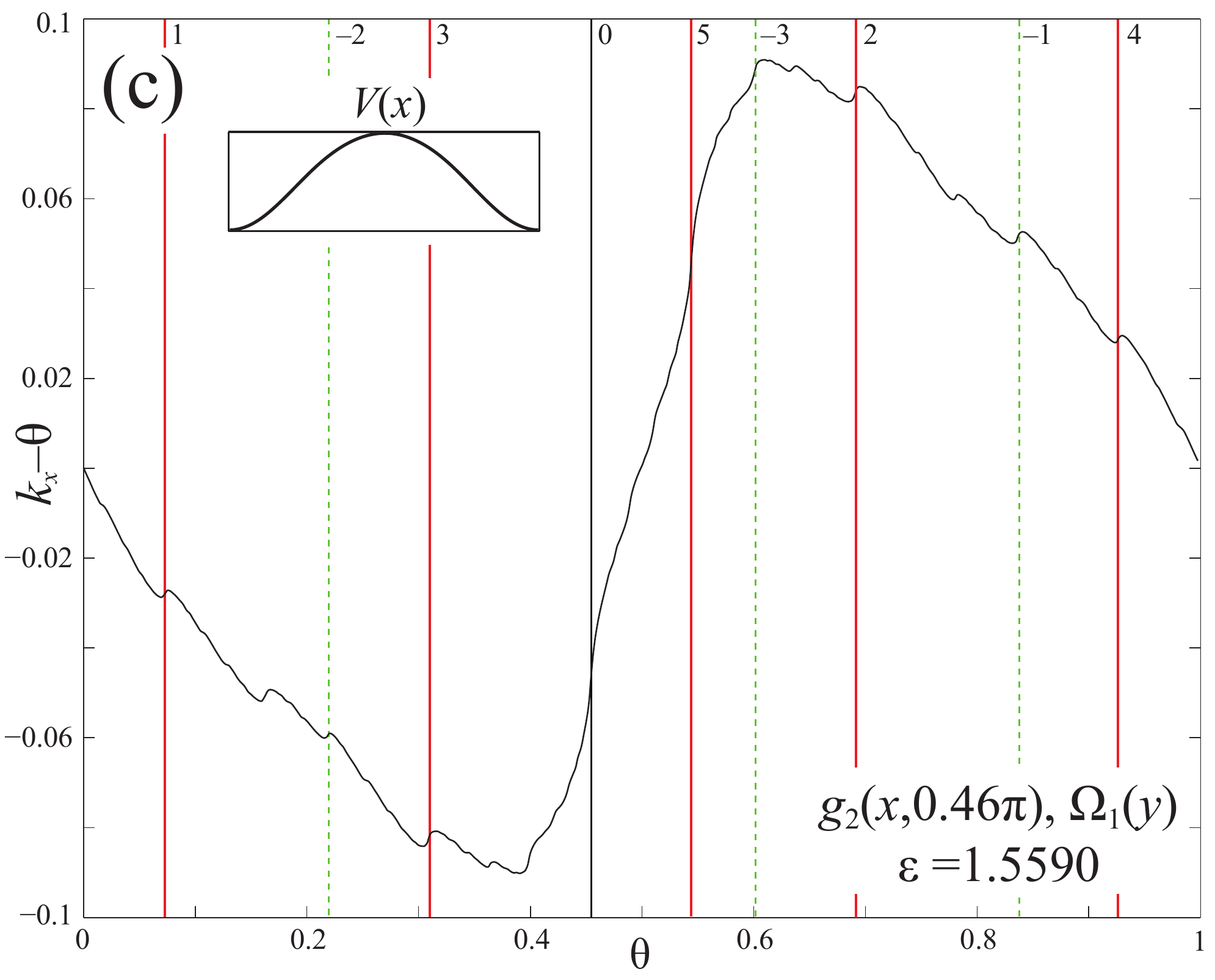}{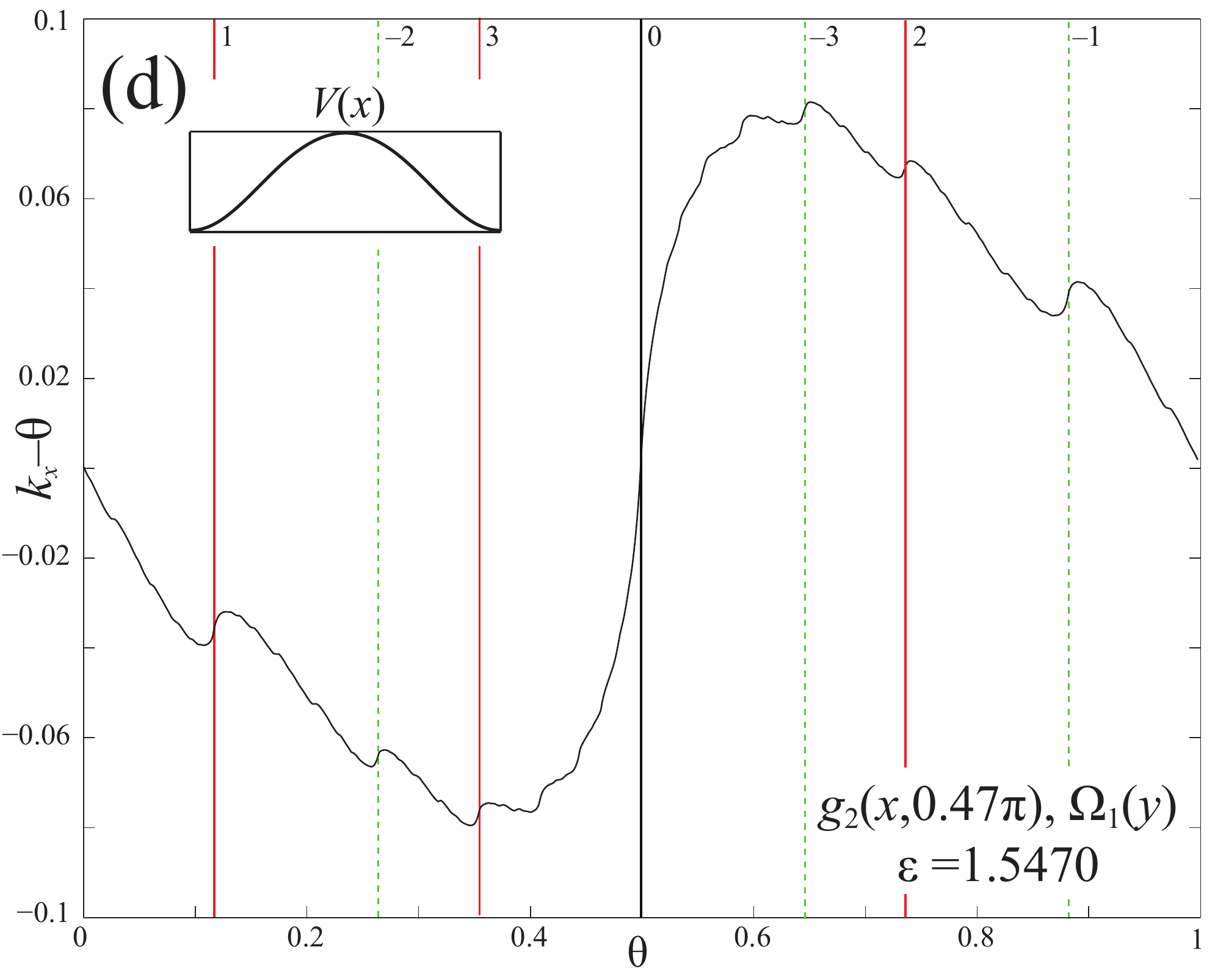}{Near-critical conjugacies, $k_x(\theta)-\theta$, for the $\gamma^{-1}$ circle of \Eq{StdMap} with $g_2$ and $\Omega_1$, and the values of $\psi$ and $\eps$ indicated. Insets show the potential $V(x)$ for these parameter values. Vertical lines show forward (solid) and backward (dashed) iterates of the largest gap.}{O1G24Frame}{3.2in}

The set of critical parameters $(\eps_{cr},\psi)$ for the golden circle of this map, computed using the seminorm method, is shown in \Fig{CritCurves}(a). For $\psi = \tfrac{\pi}{2}$, $\eps_{cr} \approx 0.971635$ since $g_2(x;\tfrac{\pi}{2}) = g_1(x)$. In addition $\eps_{cr} \approx 0.40236$ when $\psi = 0$ or $\pi$---this is half the critical parameter for the $\omega = \tfrac{\gamma}{2}$ circle of Chirikov's map \cite{Greene87}. When $\psi \in [0,\tfrac{\pi}{2}]$, the critical set exhibits a Cantor set of cusps as was first observed using Greene's criterion in \cite{Ketoja89}. The cusps are related to a complex set of symmetry breaking bifurcations of periodic orbits that limit onto the invariant circle \cite{Ketoja94}.

\InsertFigTwo{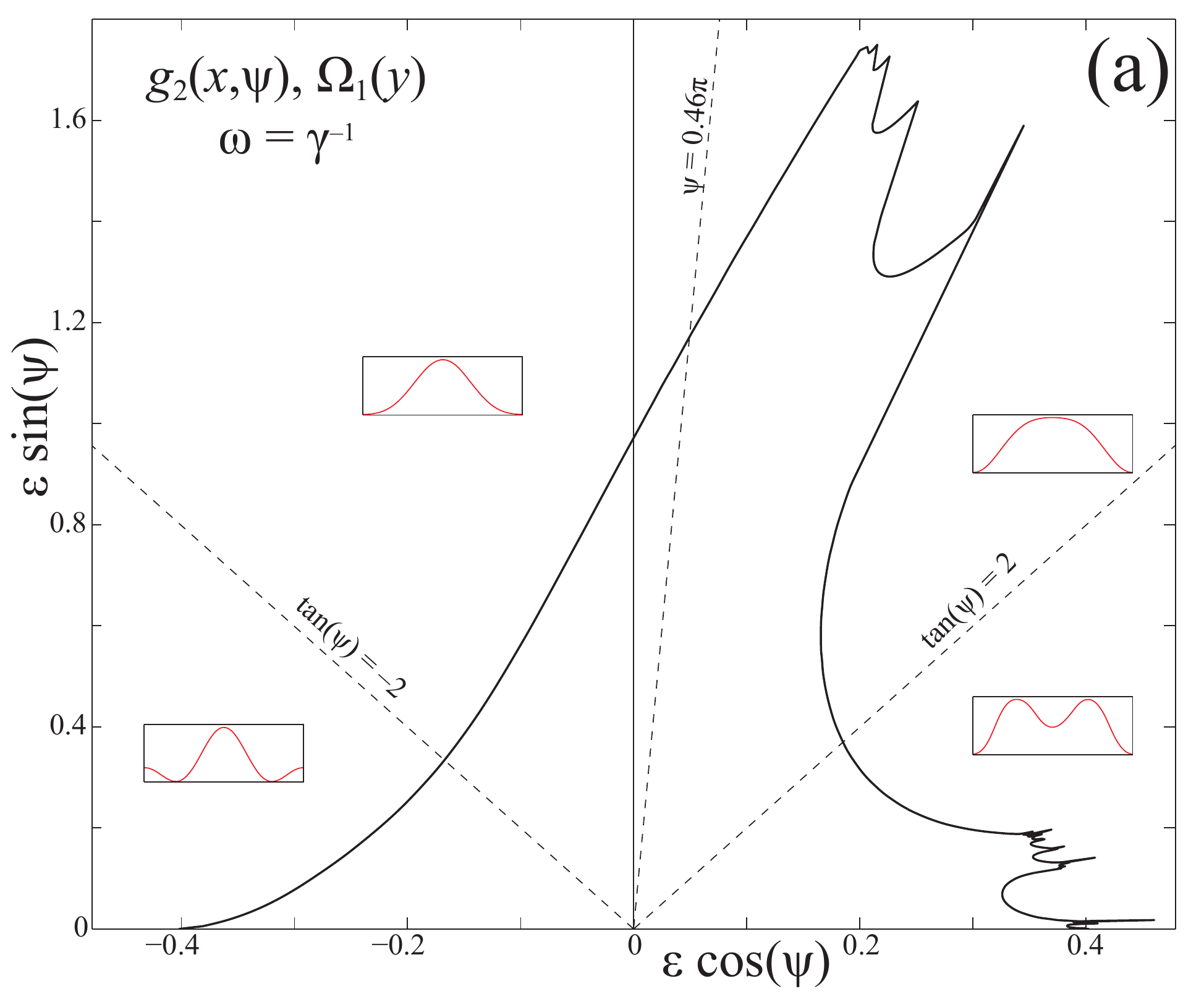}{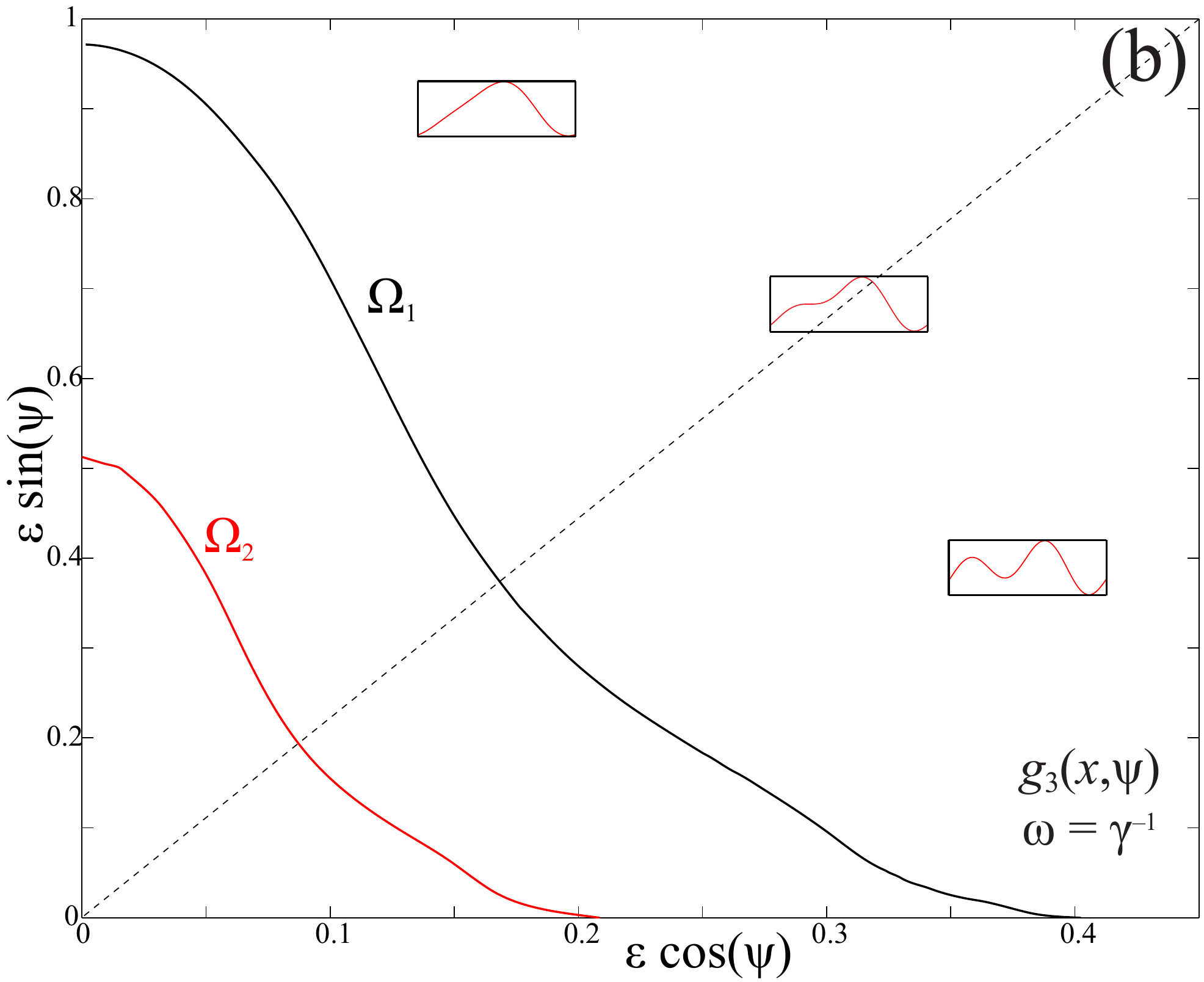}{Critical parameter set for the golden circle of \Eq{StdMap} for the frequency maps and two-harmonic forces shown, computed using the seminorm. The axes are the amplitudes of the two Fourier modes of $\eps g$. The dashed lines correspond to degenerate AI limits and for (a) to $\psi = 0.46\pi$ where the two symmetric gaps coalesce in \Fig{O1G24Frame}. Insets show representative $V(x)$.
}{CritCurves}{3in}

To begin an exploration of the breakup of circles in maps with less symmetry, consider a two-harmonic force similar to \Eq{G2}, but with a phase shift so that it is no longer odd: 
\beq{G3}
	g_3(x;\psi)=\tfrac{1}{2\pi}(\sin(\psi) \sin (2 \pi x) + \cos(\psi) \cos (4\pi x)) .
\eeq
With the odd frequency map $\Omega_1$, the map still has one reversor $S_2$, \Eq{S2Reversor}, which conjugates the two circles with rotation numbers $\pm\omega$. As before, $g_3(x+\tfrac12;-\psi) = g_3(x;\psi)$, so the parameter $\psi$ can be restricted to the range $[0,\pi]$. Moreover, since $g_3(-x;\pi-\psi) = -g_3(x;\psi)$, the transformation $(x,y,\psi) \to (-x,-y,\pi-\psi)$ leaves the dynamics invariant, so we need consider only $\psi \in [0, \frac{\pi}{2}]$. 



\InsertFigTwo{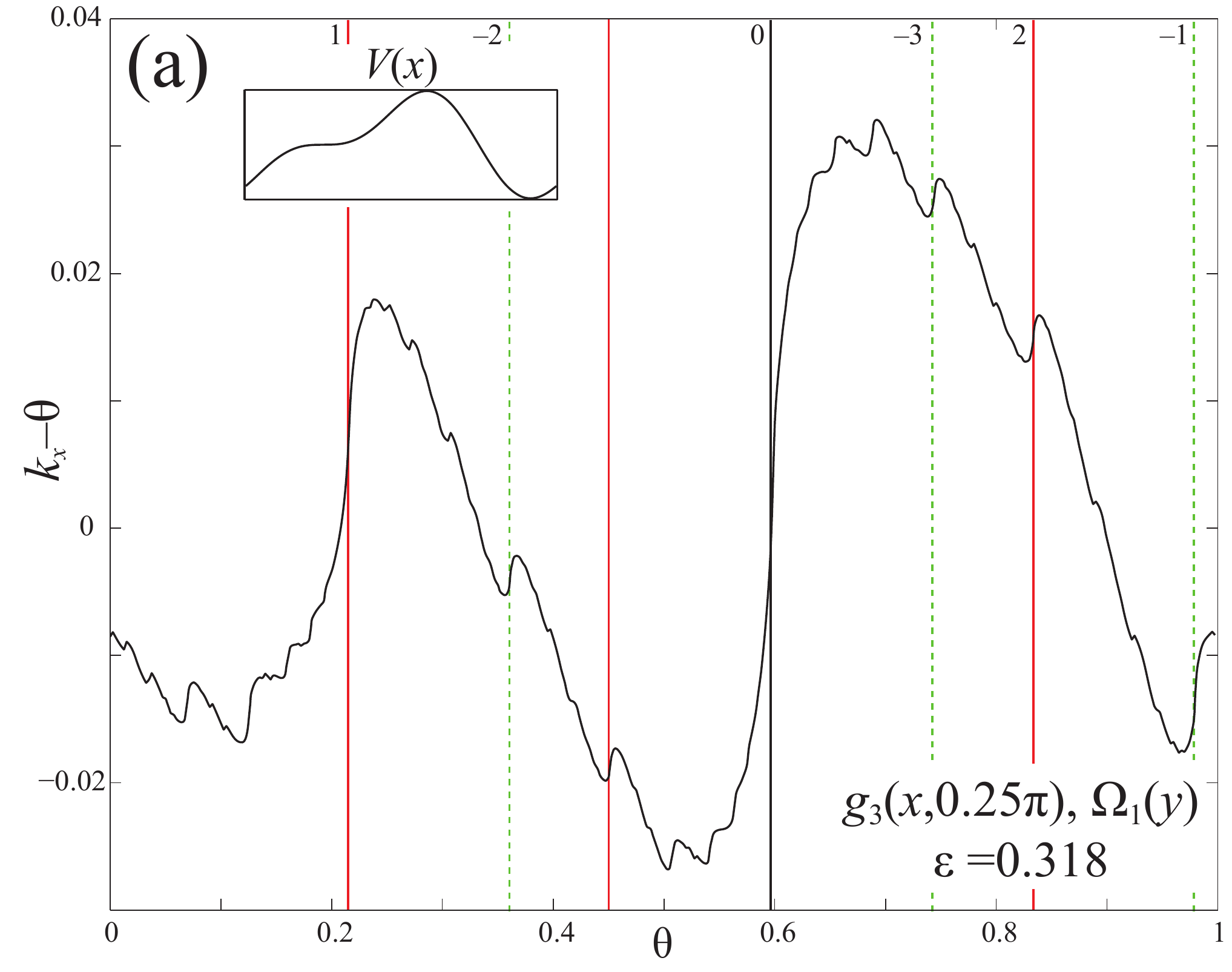}{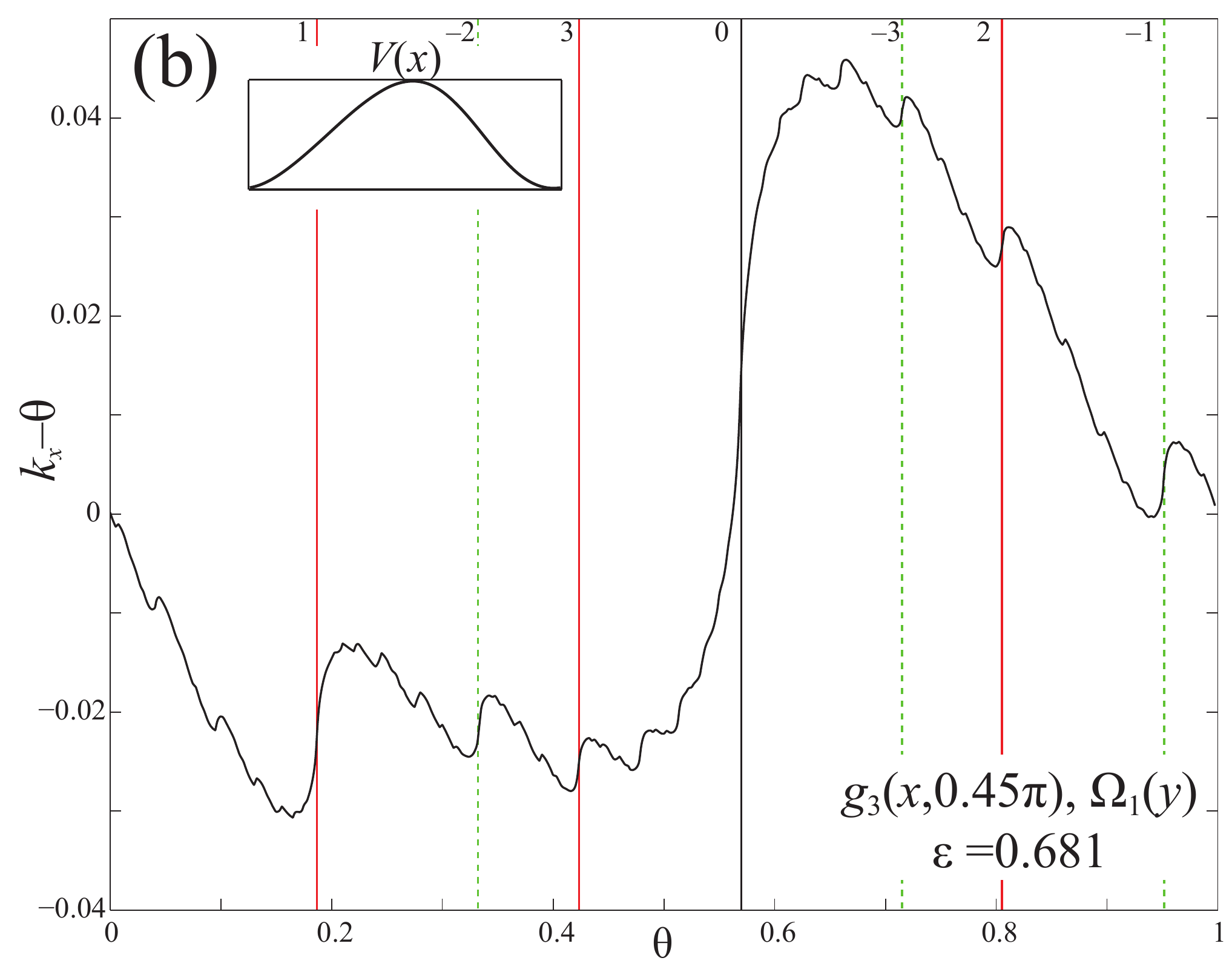}{Angle component of the conjugacy for a near-critical circle of the generalized standard map with $g_3$ and $\Omega_1$, for $\omega = \gamma^{-1}$ and two values of $\psi$. These embeddings do not exhibit the symmetry seen for the reversible case, but do display an orbit of incipient gaps.}{O1Ge3TwoPsi}{3in}

Two examples of near-critical invariant circles for this map are shown in \Fig{O1Ge3TwoPsi}. Since the force \Eq{G3} is not odd, the map \Eq{StdMap} no longer has the reversor \Eq{S1Reversor}, and the resulting invariant circles will not be invariant under $S_1$. Indeed, the conjugacies shown in the figure no longer show this symmetry and the relations \Eq{FourierPhase} no longer hold.

The potential for $g_3$ has two wells when $\psi \in [0,\tfrac{\pi}{4})$ and a single well for $\psi \in [\tfrac{\pi}{4},\pi]$. In both panes of the figure, the potential has a single well and the largest gap in $k_x$ is formed near the maximum of $V$. For $\psi < \tfrac{\pi}{4}$, the next-largest gap is formed near the smaller maximum of $V$. This gap can still be seen near $\theta = 0.2$ for the degenerate case of \Fig{O1Ge3TwoPsi}(a), and the corresponding $x = k_x(\theta)$ is near the degenerate critical point of $V$ at $x = 0.25$. This gap is the foward image of the largest gap and is much larger the remaining gaps. This secondary gap is still present in \Fig{O1Ge3TwoPsi}(b), but is significantly smaller, giving some evidence for the influence of the smaller maximum on the size of the gaps.

The critical set for the golden circle of this map is shown in \Fig{CritCurves}(b) (the black curve) for the range $\psi \in [0,\tfrac{\pi}{2}]$; this set is symmetric under reflections about the horizontal and vertical parameter axes. Just as for \Fig{CritCurves}(a), the points at $\psi = 0$ and $\psi = \frac{\pi}{2}$ correspond to Chirikov's map. However, this critical curve no longer has a Cantor set of cusps---indeed, it seems to be smooth apart from cusps at $\psi = 0$ and $\pi$, corresponding to the doubled Chirikov map.

We now consider several maps with the frequency map $\Omega_2$, \Eq{NonTwist}, for which the twist reverses at $y=0$. For this case, there are typically two circles for each rotation number in the range of $\Omega_2(y,\delta)$. When $\omega \gg -\delta$ one circle is contained in the positive twist region, $y>0$, and one in the negative twist region, $y<0$. The most interesting circles, however, cross $y=0$ so that the twist condition is locally violated. The breakup of these \emph{shearless} circles for the standard nontwist map (with $g_1$) has been much studied \cite{Apte03, Apte05, Fuchss06, Wurm11}, so we do not examine this map in detail here. 

Instead we consider a nontwist map with force $g_3$. It would seem that destroying the oddness of both $g$ and $\Omega$, would eliminate both reversors of \Eq{StdMap}. However, since $g_3$ is even about $\frac14$, $g_3(\tfrac14 +x,\psi) = g_3(\tfrac14-x,\psi)$, and $\Omega_2$ is even about $y=0$, this map has a reversor, the map $S_3$ of \Eq{S3Reversor}. This reversor conjugates the positive and negative twist circles; in particular they have the same $\eps_{cr}$. Moreover, the transformation $(x,y,\psi) \to (x,-y,\pi-\psi)$ leaves the dynamics invariant, and combining this with the reversor allows us to restrict $\psi \in [0, \frac{\pi}{2}]$.

Invariant circles of this map with $\omega \gg -\delta$ are very similar to those of a twist map, see for example the positive twist circle in \Fig{NonRevGM}. The conjugacy is very close to that for the twist case, $\Omega = \Omega_1$, shown in \Fig{O1Ge3TwoPsi}. Indeed, since $\langle k_y(\theta)\rangle \approx 0.9581$ and varies by only $\pm 0.023$, the twist is nearly linear over the range of oscillation of $y$. Moreover, since $D\Omega_2(y,\delta) = 2y \approx 1.9161$, the twist is almost a factor of two larger than that for $\Omega_1$---this accounts for the fact that $\eps_{cr}$ has been decreased by a factor very close to two from the $\Omega_1$ case.

\InsertFig{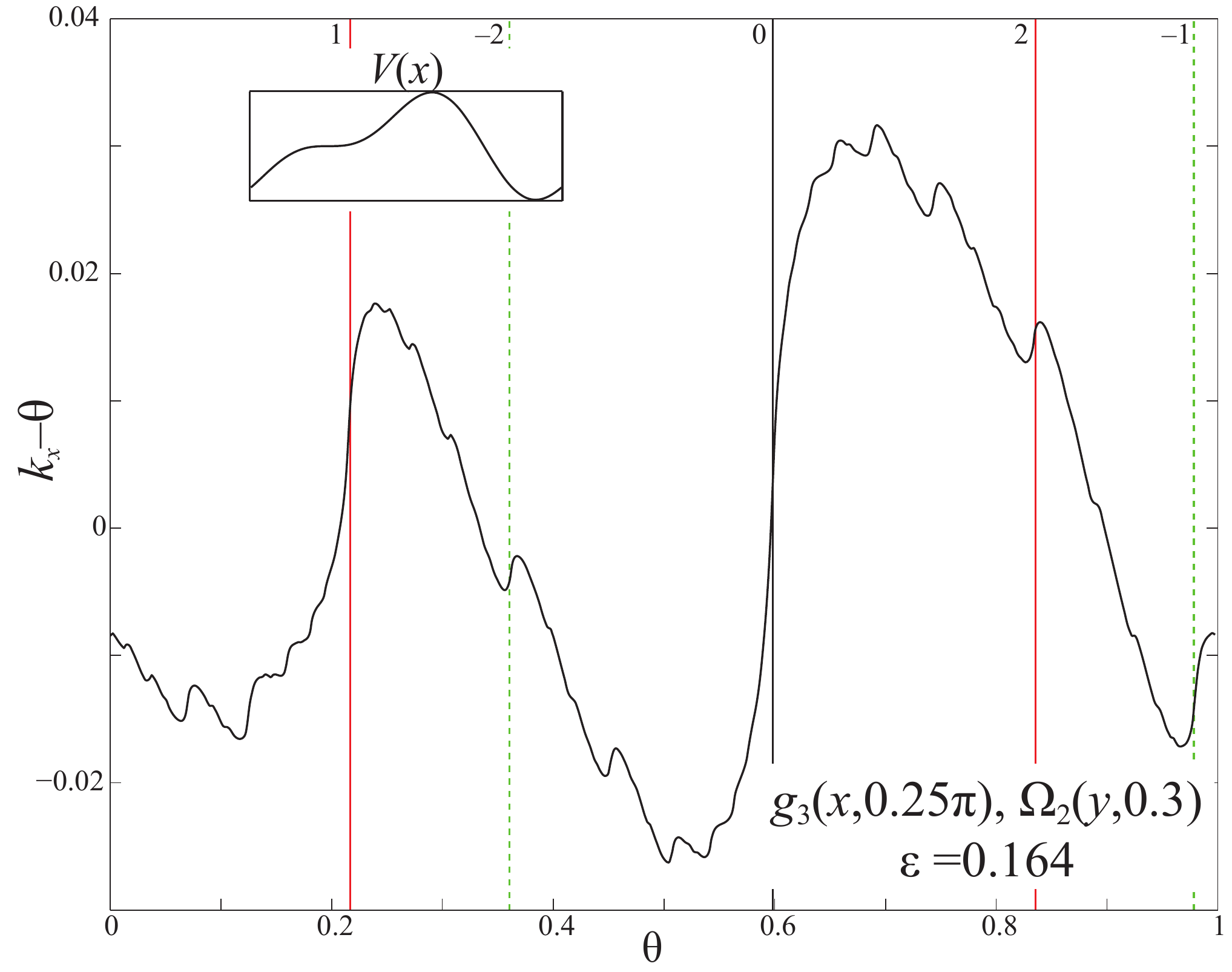}{Angle component of the conjugacy for near-critical circle of the generalized standard map with $g_3$ and $\Omega_2$, for $\omega = \gamma^{-1}$ and $y>0$.}{NonRevGM}{3in}


Circles that experience a stronger variation in twist display significantly different behavior than those of the twist case. This is especially true of those circles that cross the $y=0$ axis, where the twist vanishes. Figure~\ref{fig:BothTwist} compares circles for $\Omega_1$ and $\Omega_2$ with the same noble rotation number $\omega
\approx -0.2793$. Both still have one orbit of gaps; however for $\Omega_2$, the dominant gap is no longer associated with the largest maximum of $V$. Thus, though it appears that the invariant circle will become a ``one hole" cantorus upon destruction, traditional anti-integrable theory may not provide insight into its formation. Note also that $\eps_{cr}$ for the nontwist case is $4.02$ times larger than that for the twist map, no doubt related to the fact that $\langle D\Omega_2(y,0.3)\rangle = 2 \hat{k}_{y0} =0.2083 $ is $4.8$ times smaller than $D\Omega_1 = 1$.

\InsertFigTwo{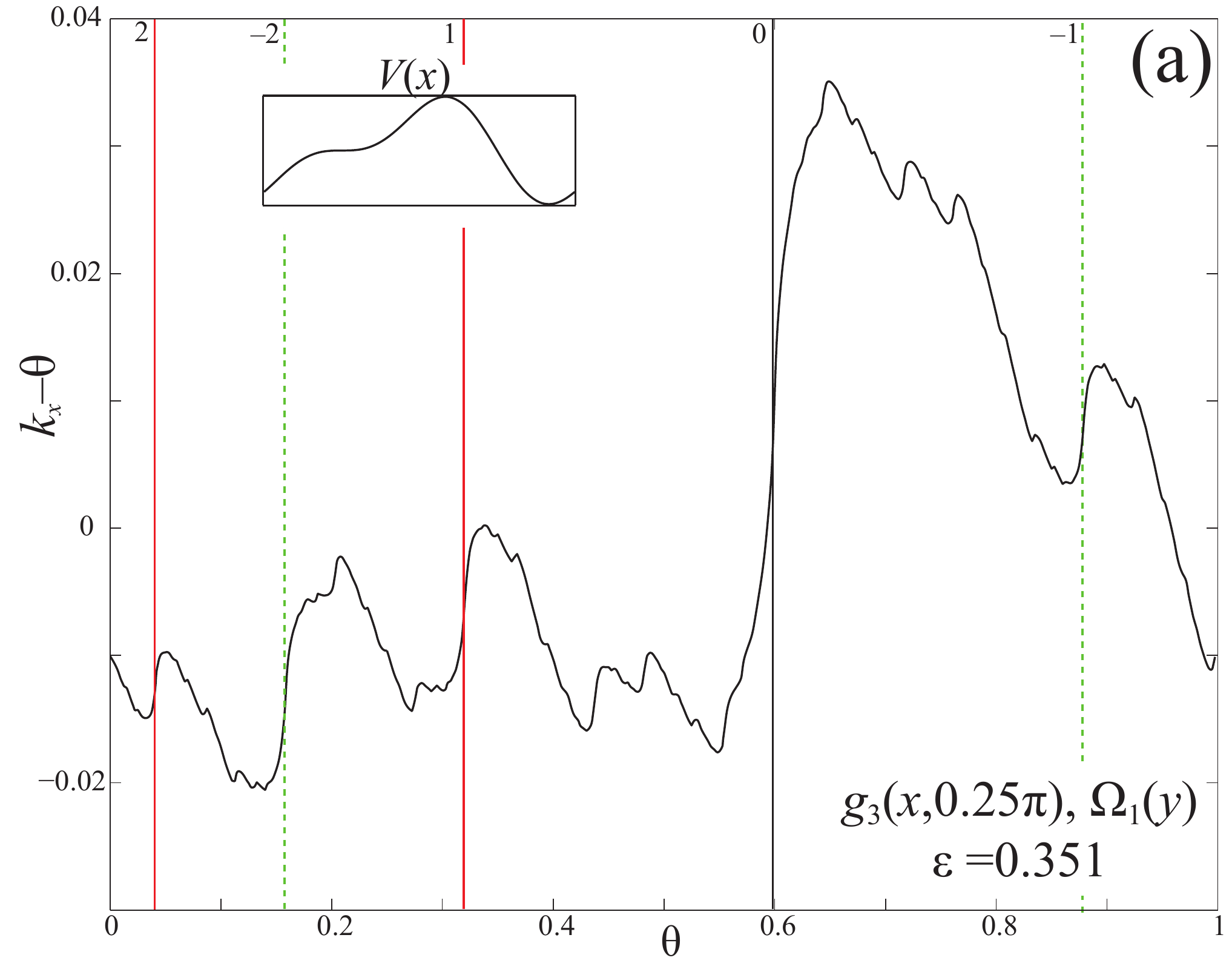}{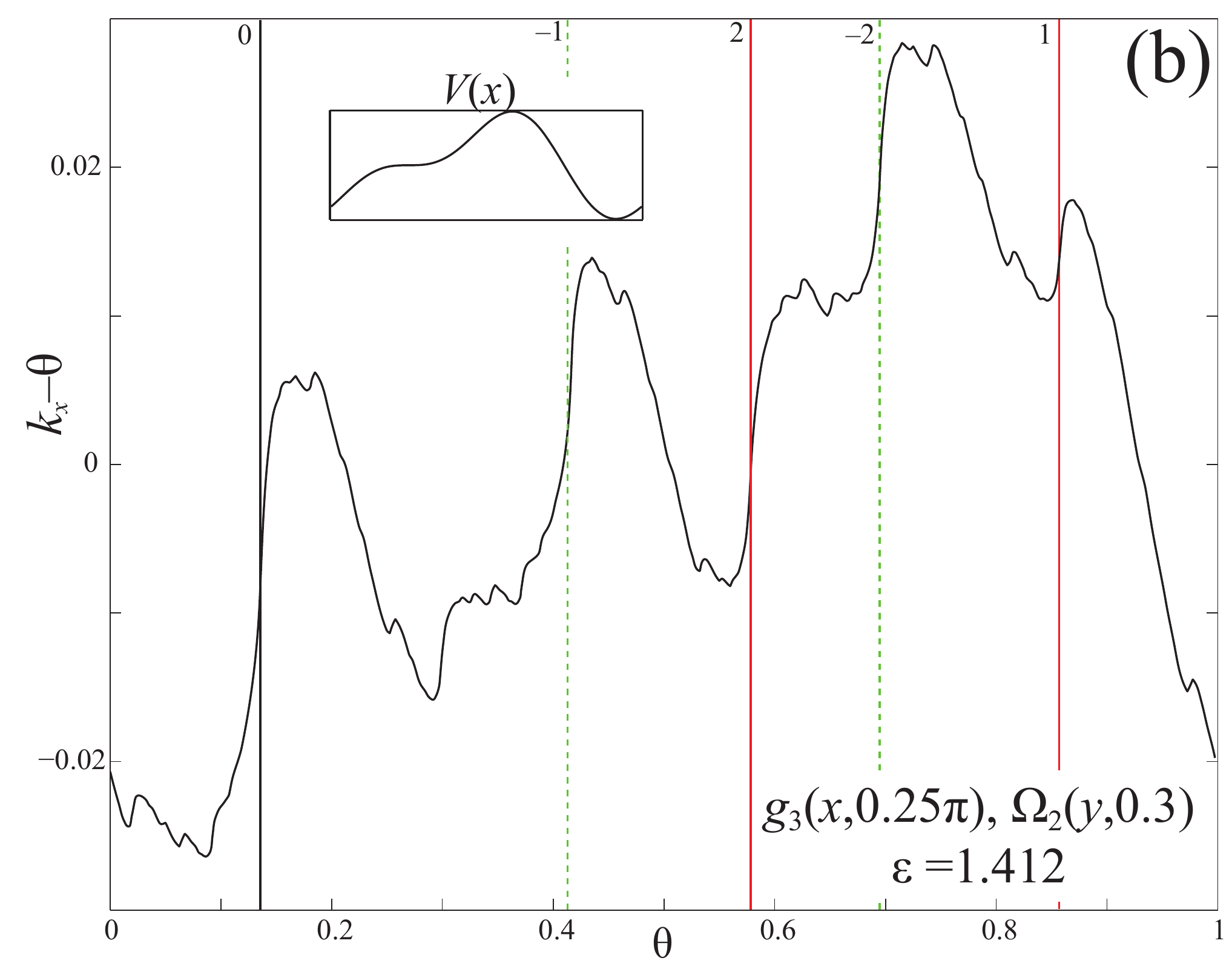}{The conjugacy $k_x$ for a near-critical invariant circle for $g_3(x;\tfrac{\pi}{4})$ with $\omega = -(12+19\gamma)/(43+68\gamma)$ and $\Omega_1$ (a) and $\Omega_2$ (b). The invariant circle for the nontwist case (b), crosses the $y$-axis.}{BothTwist}{3in}

Finally, we explored a map that we conjecture has no reversors. For this map we use the frequency map introduced by Rannou \cite{Rannou74}
\beq{Rannou}
	\Omega_3(y)=-y +\frac{1}{2\pi}(\cos(2\pi y)-1),
\eeq
that is neither even nor odd.
This frequency map, like $\Omega_2$, does not satisfy the twist condition. 
Rannou used this in conjunction with an even force that 
that has a nonzero integral so that the map has nonzero net flux and no rotational invariant circles when $\eps \neq 0$. To permit invariant circles, but to eliminate additional symmetry, we use instead a force that is neither even nor odd, has a zero average, and a full Fourier spectrum:
\beq{G4}
	g_4(x;\psi)=\frac{1}{2\pi}\sin\left(2\pi x + \psi \cos(4 \pi x) \right).
\eeq
The generalized standard map with $\Omega_3$ and $g_4$ is, to the best of our knowledge, completely nonreversible and has only a trivial symmetry group. 

An example of a near-critical golden circle for this map is shown in \Fig{RannouConj}. This circle also appears to have one hole; the principal gap forms near $\theta= 0.408$, close to the (single) maximum of the potential, as shown in the inset.
\InsertFigTwo{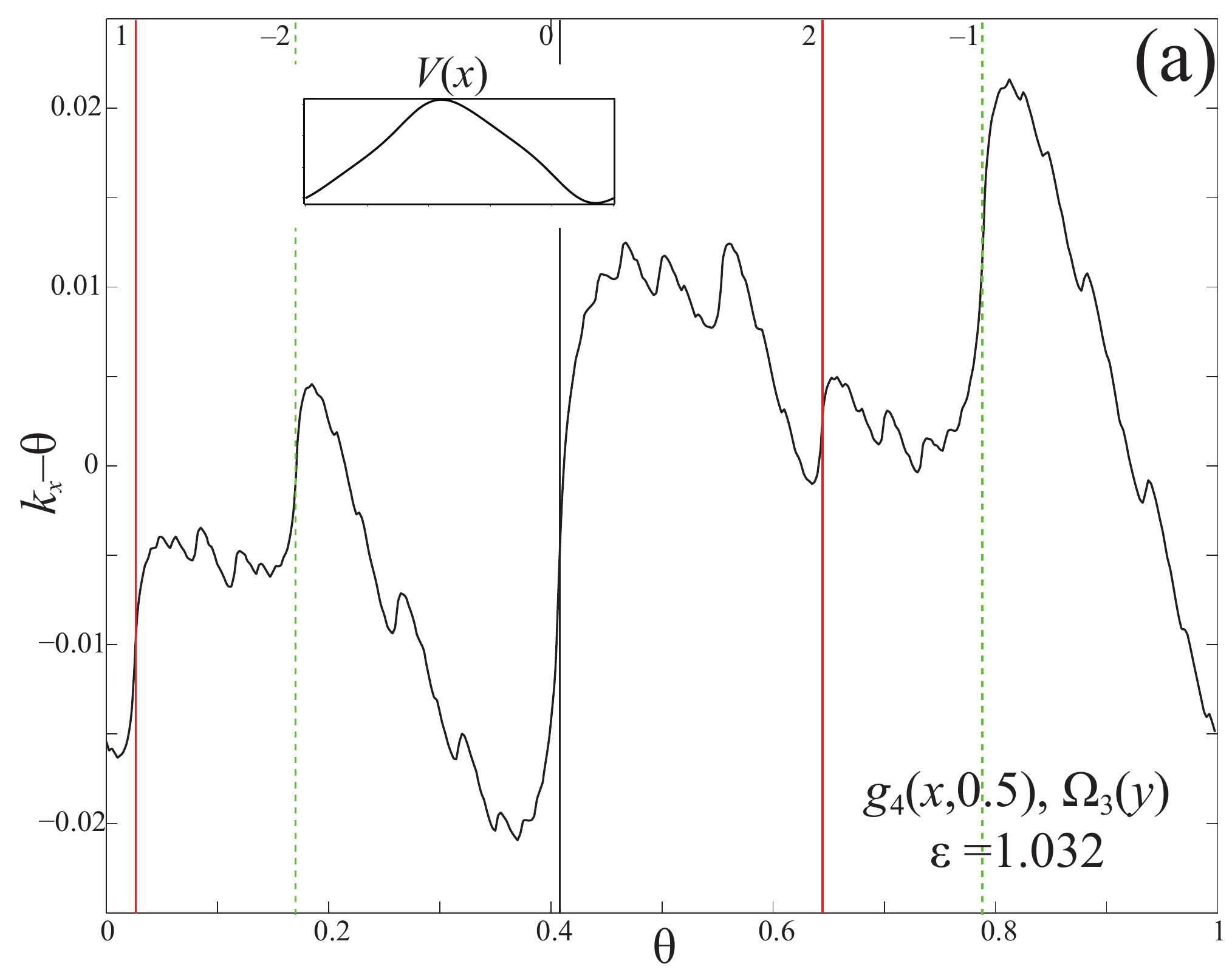}{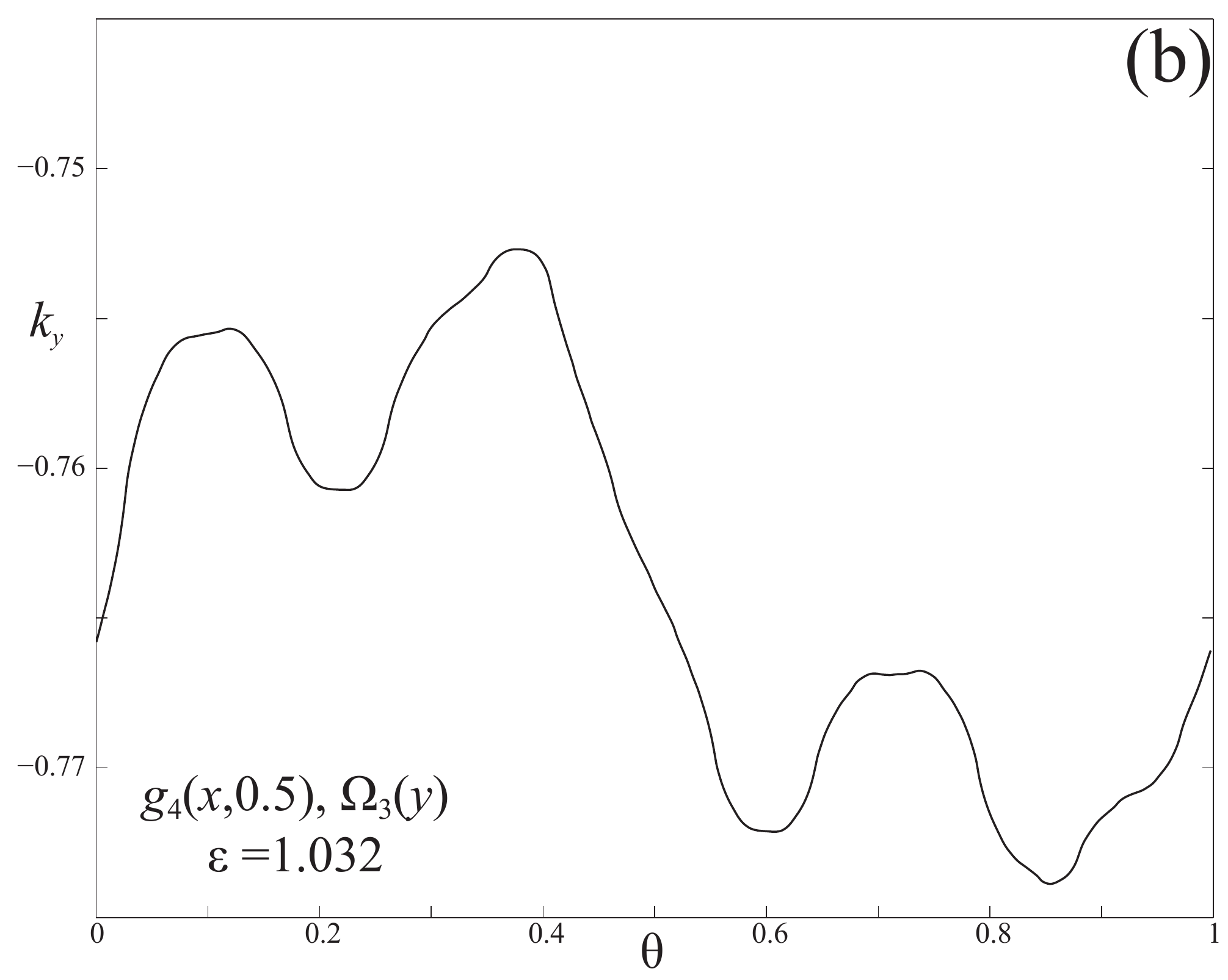}{Angle (a) and action (b) components of the conjugacy for a near-critical invariant circle with $\omega = \gamma^{-1}$ of a generalized standard map with no known symmetries.}{RannouConj}{3in}

In conclusion, for all of the cases that we investigated, the near-critical invariant circles show the formation of a single family of gaps, the largest of which is usually associated with the global maximum of $V(x)$. This presumably leads, as in \Con{Cantorus}, to the formation of a cantorus for $\eps > \eps_{cr}$, even when the twist condition is not satisfied.

\section{Robustness of Noble Invariant Circles}\label{sec:Nobles}
An invariant circle is \emph{locally} most robust in a one-parameter family $f_\eps$ if it exists for larger $\eps$ than any circle in some neighborhood in phase space. The Greene-MacKay renormalization theory \cite{MacKay93} leads to the conjecture that the locally most robust circles for twist maps have ``noble" rotation numbers, that is
\beq{FrequencyRatio}
	\omega(M,\gamma) = \frac{\nu_1}{\nu_2}, \quad \nu = (1,\gamma) M,
\eeq
for a unimodular matrix $M$.
In this case $\nu$ is an integral basis for the ring $\bZ(\gamma)$, and the collection of nobles is projectively equivalent to the set of integral bases of $\bZ(\gamma)$, recall \App{Farey}. 

It is not unexpected that the nobles are robust; for example, the size of the perturbation allowed in KAM theory is proportional to the Diophantine constant $c(\omega)$ \Eq{Diophantine}, and any noble has the maximal value $c(\gamma) = 5^{-1/2}$. The most conclusive evidence for the robustness of the nobles is the numerical study of MacKay and Stark \cite{MacKay92c} for the standard map. 

Here we investigate the relative robustness of nobles relative to other quadratic fields for the various generalized standard maps introduced in \Sec{CriticalCircle}. The six real rings $\bZ(r)$ with the smallest discriminants $D(r)$, or equivalently the largest Diophantine constants, are shown in \Tbl{Rings}. To characterize the computational accuracy, we first compute the critical parameter value of Chirikov's map for circles in each ring using both the residue and the seminorm methods. \Tbl{Rings} shows that the golden mean is, of course, the clear winner with the largest critical value, \Eq{GoldenEpsCR}.

\begin{table}
\centering
\begin{tabular}{c|c|c|l|l|l}
$r$ &$D(r)$&CF&$\eps_{cr}^{R}(r)$&$\eps_{cr}^{S}(r)$&$\langle\Delta \eps_{cr}\rangle_{rms}$\\[1.2ex]
\hline

$\gamma$ 	&5				&$[\overline{1}]$		&0.9716358(1)	&0.971638(1)	&$1.6(10)^{-6}$\\
1+$\sqrt{2}$&8				&$[\overline{2}]$		&0.957447(6)		&0.95744(7)	&$3.0(10)^{-6}$\\
1+$\sqrt{3}$&12				&$[\overline{2,1}]$	&0.87608(1)		&0.8756(6)	&$4.3(10)^{-4}$\\
$\tfrac12(3+\sqrt{13})$&13	&$[\overline{3}]$		&0.89086(2)		&0.8905(3)	&$3.3(10)^{-3}$\\
$\tfrac12(3+\sqrt{17})$&17	&$[\overline{3,1,1]}$	&0.91573762(4)	&0.9143(3)	&$1.0(10)^{-4}$\\
$\tfrac12(3+\sqrt{21})$&21	&$[\overline{3,1}]$	&0.77242(5)		&0.7720(8)	&$2.1(10)^{-4}$\\
\end{tabular}
\caption{\footnotesize The six quadratic fields $\bQ(r)$ with the smallest discriminants, $D$. Each is generated by the irrational $r$ with the periodic continued fraction (CF) shown. The $4^{th}$ and $5^{th}$ columns give $\eps_{cr}(r)$ for Chirikov's map computed by the residue method (for orbits up to period $30,000$) and the seminorm method, respectively. The parentheses indicate the estimated extrapolation error in last digit. The last column is the root mean square difference between the two computations for $256$ rotation numbers in each field.}
\label{tbl:Rings}
\end{table}

The error in each computation of $\eps_{cr}$ was estimated by comparing the extrapolated values from the last two iteration steps. The computations indicate that $\eps_{cr}$ is about an order of magnitude less accurate using the seminorm method; however, the error is occasionally overestimated by this calculation as discussed in \Sec{FindEpsCr}. It is not surprising that error for both methods tends to increase with discriminant: circles from fields with larger $D$ tend to be closer to low-order resonances. For the residue method, each periodic orbit is nearer to those with lower period, making the orbits more difficult to compute, and for the conjugacy method, a larger discriminant implies that the denominators in the solution of the cohomology equation \Eq{CohoSoln} will be smaller, inducing larger errors in the Fourier coefficients.

To study local robustness, we construct a set of nearby rotation numbers from each ring, using the $256$ rationals in $[0,\tfrac12]$ from the Farey tree up to level eight, see \App{Farey}. Each pair of neighboring rationals, $(p_1/q_1, p_2/q_2)$, gives a basis vector 
$\nu = (p_1 +p_2 r, q_1 + q_2 r)$ for each of the six rings. The six corresponding frequency ratios $\omega(M,r)$ \Eq{FrequencyRatio} lie the interval between the neighbors. These interval widths range from $0.0002$ to $0.1$ and define the ``local" neighborhoods for the robustness test. The last column of \Tbl{Rings} shows the rms difference between the residue and seminorm computations for the $256$ irrationals in each ring for Chirikov's map. 

For a given pair of neighboring rationals, we compare the values of $\eps_{cr}$ for each of the six rings to determine which is locally most robust.
Since the computations of $\eps_{cr}$ are uncertain, we are unable to determine a winner when the difference between the most robust and other circles is too small. The approximation to $\eps_{cr}$ combined with the estimated error generates a critical interval, in which the true $\eps_{cr}$ likely resides. A circle is deemed to be locally most robust if its entire critical interval is above the critical intervals of the other five circles. If the most robust interval overlaps with any other interval, we cannot conclusively declare a winner. We also discounted any result for which the width of the critical interval for the most robust circle was greater than $10^{-4}$. 

\Tbl{CompareRings} shows the number of times each ring was most robust in a comparison with the other rings. For the residue computation the nobles are most robust $69\%$ of the time, and for the seminorm method, $67\%$. Errors caused 2-3\% of the trials to be discarded. These results are similar to those of MacKay and Stark \cite{MacKay92c}. There are cases in which numbers in $\bQ(\sqrt{2})$ appear more robust than a noble, but MacKay and Stark found that when the noble at a particular level on the Farey tree was not most robust, there was always a noble at a deeper level that was. Thus the cases in which a noble does not win the contest appear to result from the neighborhood being too large.

If we remove the field $\bQ(\gamma)$, then a comparison of the remaining five fields (the second and seventh rows of the table) shows that $\bQ(\sqrt{2})$ is most robust. Continuing this comparison for the remaining fields shows that robustness seems to be a monotone function of the field discriminant. However, for the seminorm method, the uncertainties are too large to determine a winner amongst the four largest discriminant fields.

\begin{table}
\centering
\begin{tabular}{c|c|c|c|c|c|c|c}
Method &$\gamma$ &1+$\sqrt{2}$&1+$\sqrt{3}$&$\tfrac12(3+\sqrt{13})$&$\tfrac12(3+\sqrt{17})$&$\tfrac12(3+\sqrt{21})$&Unsure\\[1.2ex]
\hline
residue
&173		&67	&10	&0	&0	&0	&6\\
&	&240	&10	&0	&0	&0	&6\\
&	&	&255	&0	&0	&0	&1\\
&	&	&	&180	&44	&0	&32\\
&	&	&	&	&256	&0	&0\\
\hline
seminorm
&166	&73	&6	&1	&0	&1	&9\\
&	&228	&9	&2	&1	&1	&15\\
&	&	&34	&29	&6	&2	&185\\
&	&	&	&33	&7	&3	&213\\
&	&	&	&	&7	&5	&244\\
\end{tabular}
\caption{\footnotesize The relative robustness of invariant circles of Chirikov's standard map with rotation numbers in different algebraic rings. The numbers in each column indicate the number of times a circle with rotation number in that ring was most robust for $256$ trials. The final column is number of trials for which a winner could not be conclusively determined.}
\label{tbl:CompareRings}
\end{table}

It has long been conjectured that noble invariant circles should be locally most robust for maps with more general forces, with nonmontone frequency maps, and without reversing symmetries. We investigated this using the seminorm method for the maps of \Sec{CriticalCircle} and found that in each case the nobles appeared to be locally most robust to nearly the same degree that they are for Chirikov's map. For example, \Tbl{RobustDataG3} shows the results for the generalized standard map with $g_3$ and $\Omega_1$ for $\psi \in [0, \pi/2]$. Regardless of the choice of $\psi$, the noble circles were once again locally most robust at least $60\%$ of the time. We suspect that---just like for Chirikov's map---whenever one of the six non-nobles was a winner, there is a noble below it on the Farey tree that wins, but we did not examine this in detail. 

\begin{table}
\centering
\begin{tabular}{c|c|c|c|c|c|c|c}
 $\tfrac{202}{\pi}\psi$&$\gamma$ & 1+$\sqrt{2}$ & 1+$\sqrt{3}$ & $\tfrac12(3+\sqrt{13})$& $\tfrac12(3+\sqrt{17})$ & $\tfrac12(3+\sqrt{21})$&Unsure \\[1.2ex]
\hline
$1$	 		&165		&67		&6		&1		&2		&2		&13\\
$11$ 		&156		&80		&2		&1		&1		&6		&10\\
$21$ 		&154		&76		&5		&1		&2		&4		&14\\
$31$ 		&160		&72		&6		&2		&2		&2		&12\\
$41$ 		&160		&71		&2		&1		&3		&3		&16\\
$51$ 		&158		&73		&5		&1		&2		&5		&12\\
$61$ 		&152		&76		&3		&1		&2		&8		&14\\
$71$ 		&164		&72		&0		&0		&2		&7		&11\\
$81$ 		&163		&71		&3		&0		&1		&5		&13\\
$91$ 		&163		&68		&7		&2		&2		&4		&10
\end{tabular}
\caption{\footnotesize Relative robustness of circles from six quadratic fields for the generalized standard map with $g_3(x;\psi)$, $\Omega_1(y)$ and $10$ values of $\psi$ using the seminorm method. The number in each column corresponds to the number of times each of the six rings of \Tbl{Rings} was most robust from $256$ trials. The final column is number of trials for which a winner could not be conclusively determined.}
\label{tbl:RobustDataG3}
\end{table}

As a second example, \Tbl{RobustDataO2G3} shows the robustness data for the nontwist map with frequency map $\Omega_2$. In this case the symmetries of the map that allow us to restrict to $\omega \in (0,\tfrac12)$ are gone. We therefore studied the $512$ rotation numbers generated from the rationals on level $9$ of the Farey tree in the interval $(0,1]$. We once again found that circles with noble rotation numbers are locally most robust at least $55\%$ of the time. For this case the errors were larger, and up to $9\%$ of the comparisons had to be discarded.

In conclusion, our results support \Con{Noble}, that invariant circles with noble rotation numbers are generically, locally most robust for the generalized standard map \Eq{StdMap}.

\begin{table}
\centering
\begin{tabular}{c|c|c|c|c|c|c|c}
 $\tfrac{202}{\pi}\psi$&$\gamma$ & 1+$\sqrt{2}$ & 1+$\sqrt{3}$ & $\tfrac12(3+\sqrt{13})$& $\tfrac12(3+\sqrt{17})$ & $\tfrac12(3+\sqrt{21})$&Unsure \\[1.2ex]
\hline
$1$		&314		&116		&19		&16		&1		&4		&42\\
$11$		&292		&149		&24		&6		&0		&4		&37\\
$21$		&296		&133		&27		&11		&2		&1		&42\\
$31$		&291		&151		&19		&2		&1		&4		&44\\
$41$		&287		&152		&17		&4		&0		&4		&48\\
$51$		&283		&142		&26		&8		&1		&5		&47\\
$61$		&291		&145		&24		&5		&1		&3		&43\\
$71$		&289		&138		&25		&11		&1		&1		&47\\
$81$		&297		&145		&23		&7		&0		&3		&37\\
$91$		&302		&146		&15		&6		&0		&5		&38

\end{tabular}
\caption{\footnotesize Relative robustness of circles from six quadratic fields for \Eq{StdMap} with $g_3(x;\psi)$ and $\Omega_2(y,0.3)$ and $10$ values of $\psi$ using the seminorm method. The number in each column corresponds to the number of times each of the six rings of \Tbl{Rings} was most robust from $512$ trials. The final column is number of trials for which a winner could not be conclusively determined.}
\label{tbl:RobustDataO2G3}
\end{table}


\section{The Critical Function}\label{sec:Critical}
The destruction of the last rotational invariant circle is a significant event in the dynamics of area-preserving maps since this corresponds to the transition to global chaos and unbounded orbits. In this section we will identify which of the nobles is globally most robust for the some of the generalized standard maps studied in \Sec{CriticalCircle}. Since, as we confirmed in \Sec{Nobles}, the noble circles appear to be locally most robust, we will confine our search to these rotation numbers.

The last invariant circle is the global maximum of the \emph{critical function} $\eps_{cr}(\omega)$. The conjecture that the golden circle is  most robust for Chirikov's map has been supported by many computations using periodic orbits \cite{Schmidt82, Buric90, Murray91}. 
It is also supported using the seminorm technique, as can be seen in \Fig{G1CritFun}(a). Here we computed $\eps_{cr}(\omega)$ for $256$ noble frequencies in $(0,\tfrac12)$ using level $8$ of the Farey tree as in \Sec{Nobles}. The resulting picture is indistinguishable from that obtained using the residue method and has maximum value $\eps_{cr}(2-\gamma) = \eps_{cr}(\gamma)$ of \Eq{GoldenEpsCR}.\footnote
{Symmetries of Chirikov's map, recall \App{Symmetries}, imply that $\eps_{cr}(\omega) = \eps_{cr}(n \pm \omega)$ for $n \in \bZ$, so it is sufficient to consider $\omega \in [0,\tfrac12]$.}

\InsertFigTwo{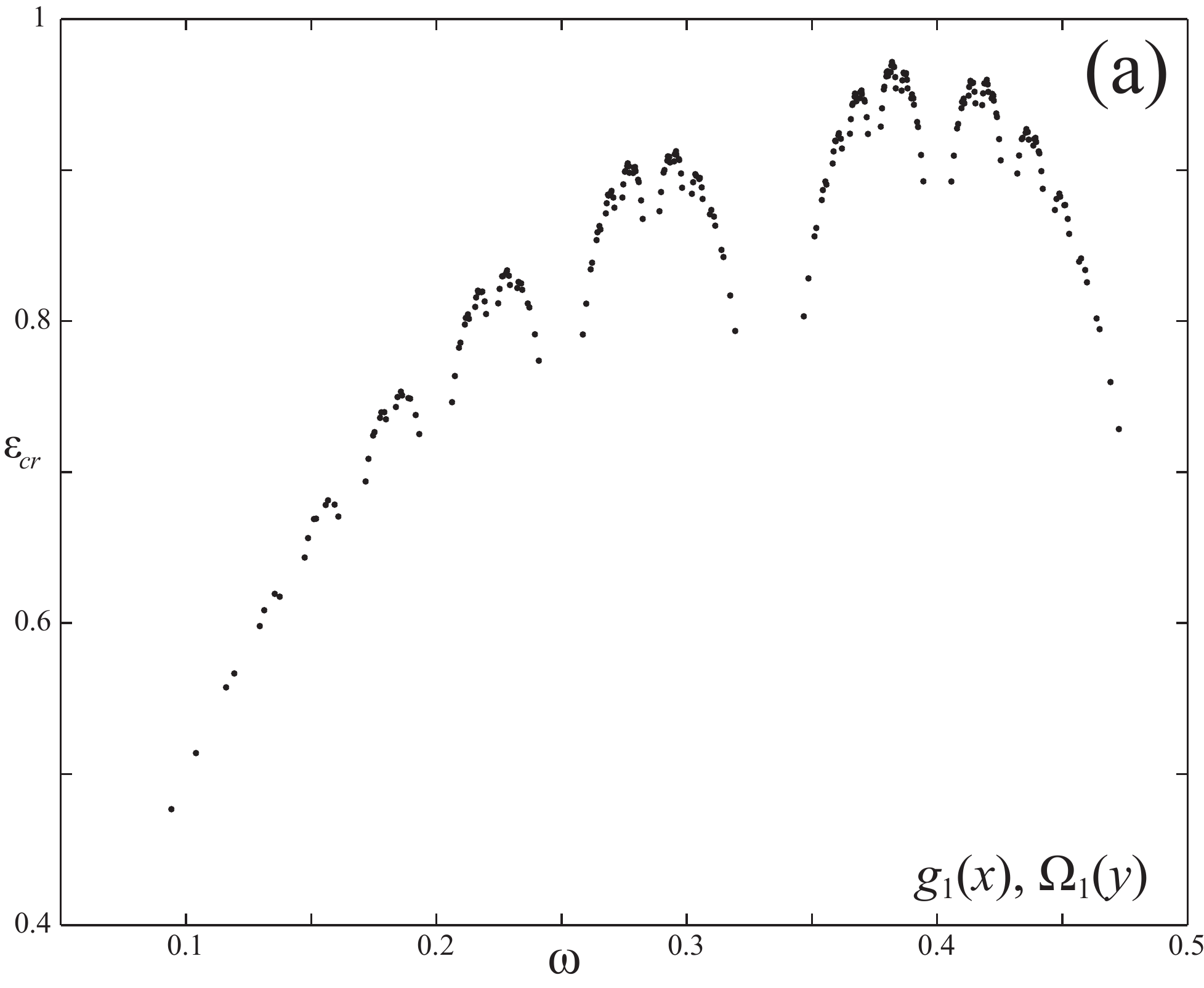}{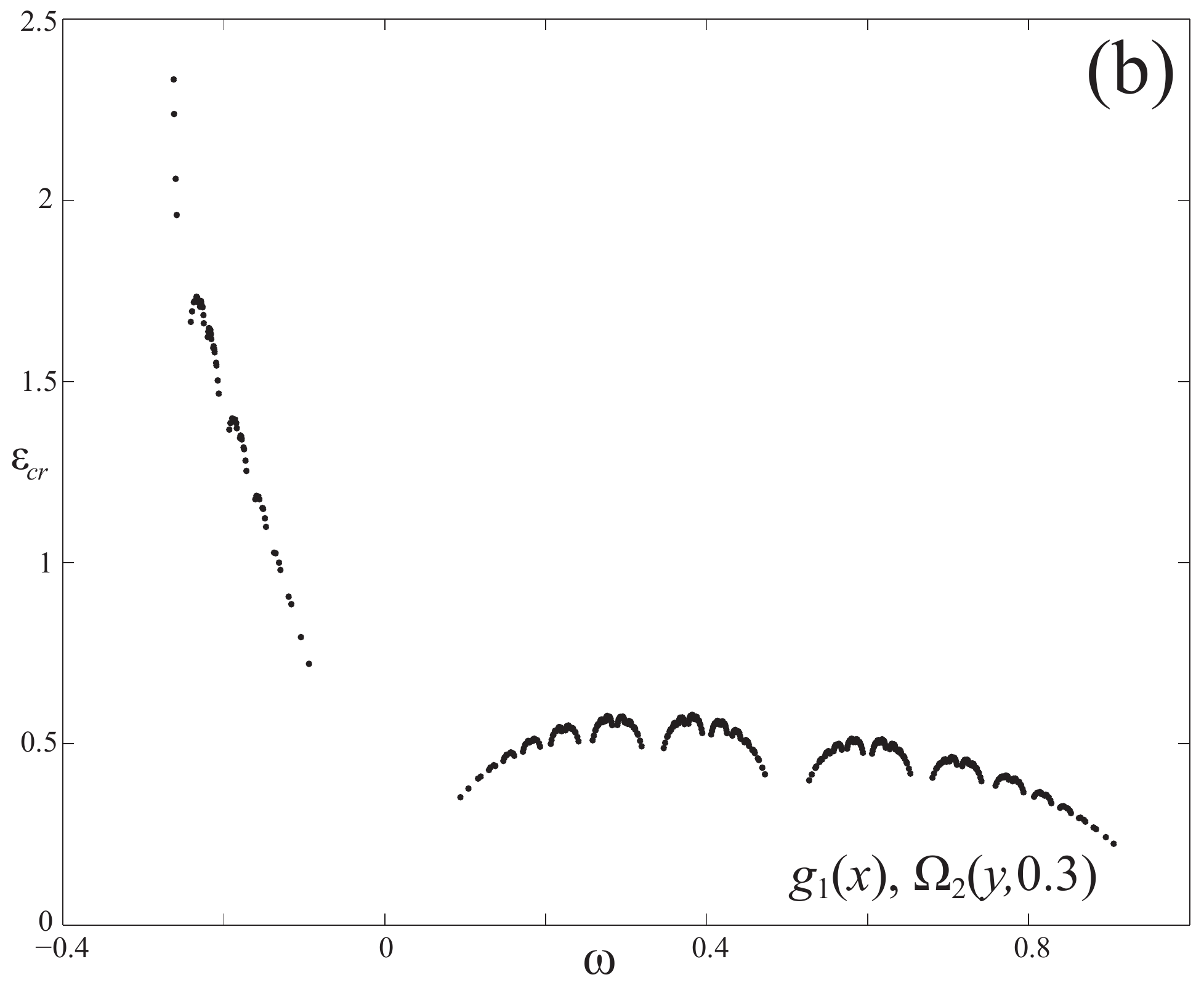}{Critical functions for (a) Chirikov's standard map and (b) the standard nontwist map for $\delta = 0.3$, using the seminorm method.}{G1CritFun}{3.2in}

Determining the globally most robust circle for the standard nontwist map, with $\Omega_2$ and $g_1$, is complicated because of the lack of periodicity in frequency space. Figure \ref{fig:G1CritFun}(b) shows the critical function for $-\delta < \omega < 1$ for $512$ nobles generated from level $9$ of the Farey tree with root $[\frac01,\frac11]$ and $112$ negative nobles greater than $-\delta$. When $\omega >0$, this critical function has a similar structure to Chirikov's map; however---as has been often observed \cite{Morrison00, Szezech09}---the circles become increasingly more robust as $\omega \to -\delta$. Indeed we found the most robust circle to be that with the smallest $\omega$. Unfortunately, the extrapolation errors for the estimate of $\eps_{cr}$ tend to become larger in this limit and we discarded the estimates for $44$ circles for which the error was larger than $10^{-4}$. A similar picture of the critical function was obtained for the nontwist map with the force $g_3$: the most robust circle was always found to be the one with rotation number closest to $-\delta$. 

The critical function is considerably more complex for the multiharmonic twist map, with $\Omega_1$ and $g_2$. Figure~\ref{fig:O1G2Crit}(a) shows that the graph of $\eps_{cr}(\omega)$ resembles that for Chirikov's map when $\psi \in [\tfrac{\pi}{2},\pi]$; however, when $\psi \in [0,\tfrac{\pi}{2}]$, $\eps_{cr}(\omega)$ is much less regular, see \Fig{O1G2Crit}(b). Recall that this region of parameters corresponds to the Cantor set of cusps in the  critical set of the golden circle, e.g., \Fig{CritCurves}(a). We also found that $\eps_{cr}(\omega)$ is especially erratic when two gaps collide, as we observed near $\psi = 0.47\pi$ in \Fig{O1G24Frame}. 

\InsertFigTwo{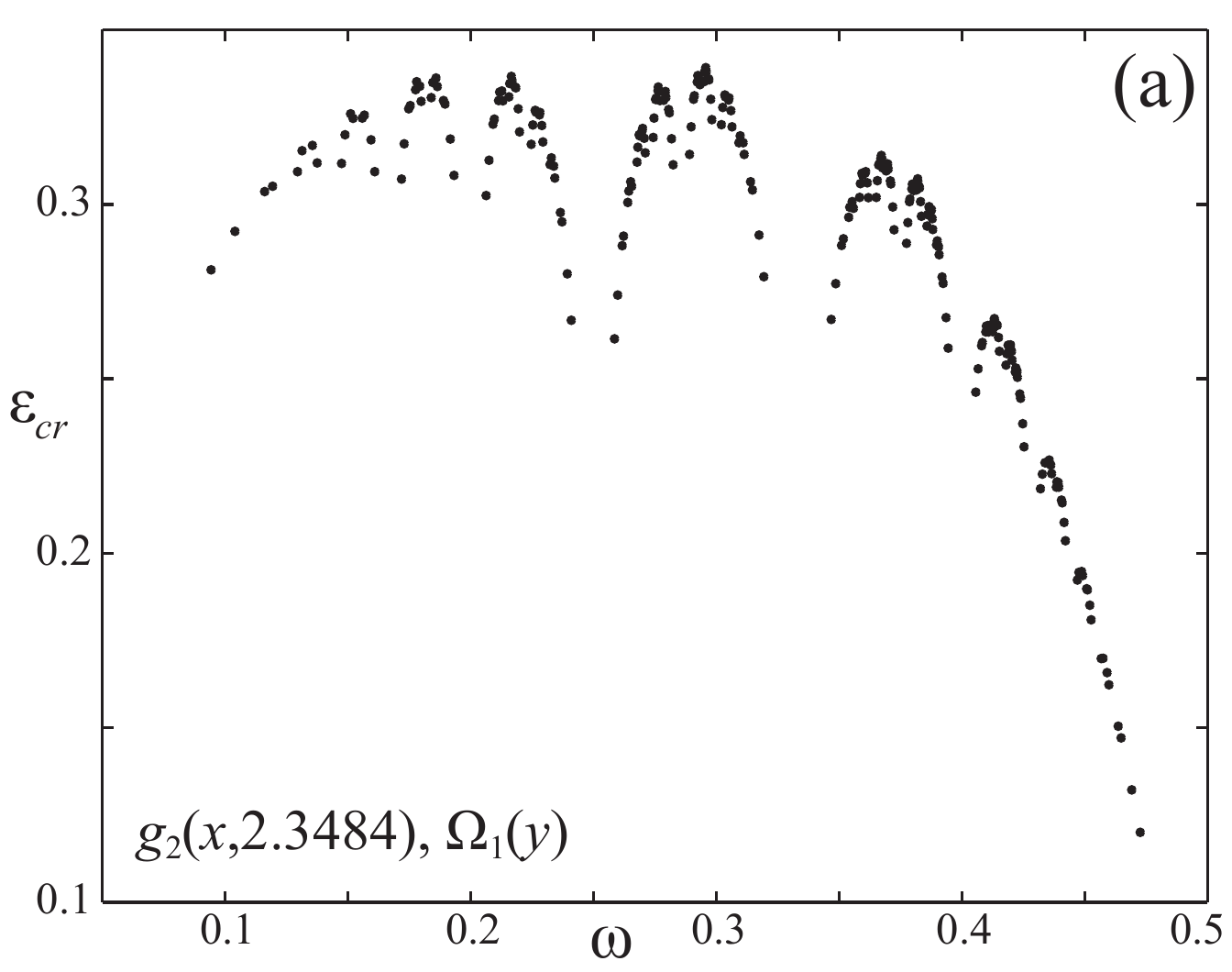}{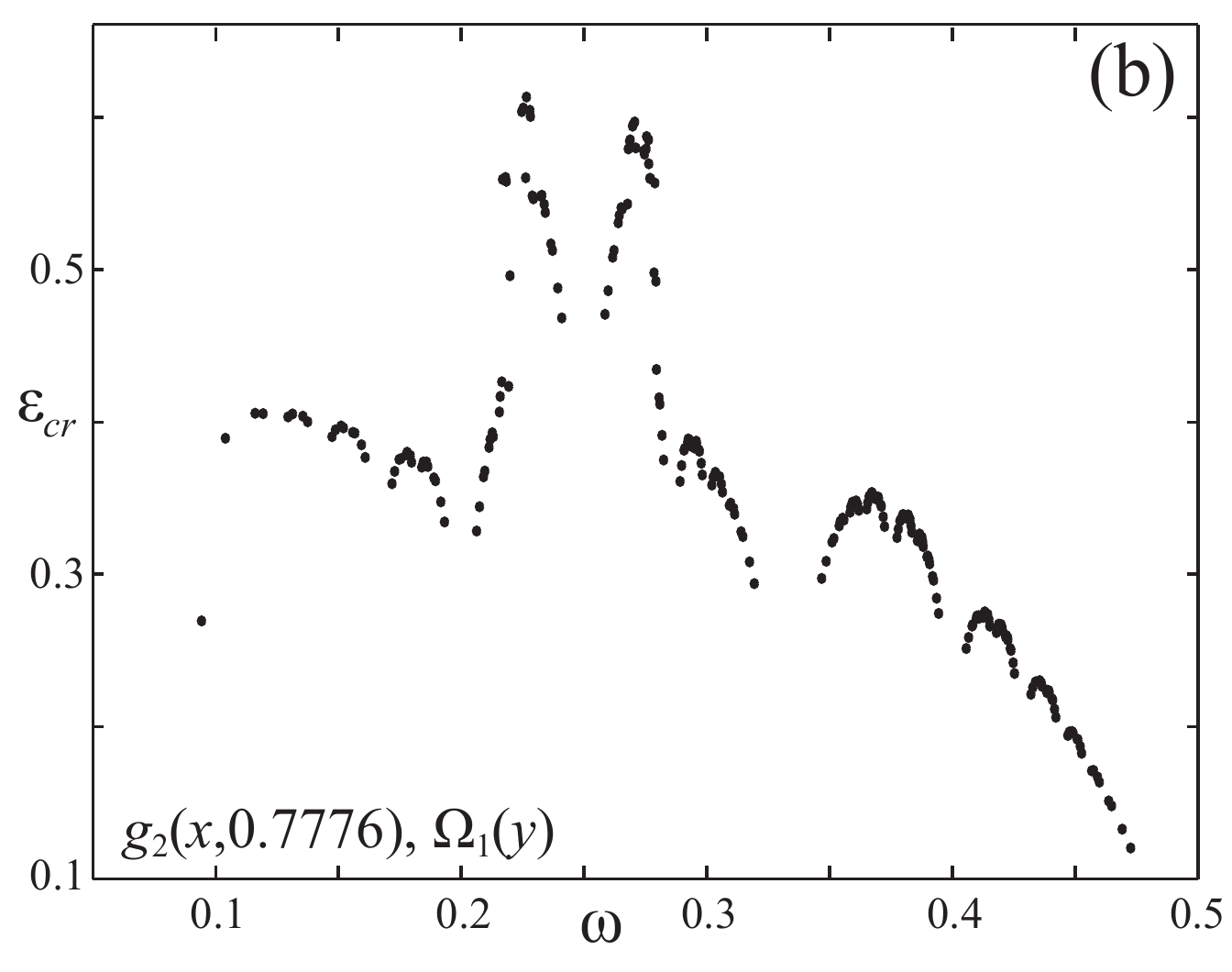}{Critical function for the generalized standard map with  $\Omega_1$ and (a) $g_2(x; 2.3484)$, or (b) $g_2(x;0.7776)$ for $256$ noble rotation numbers.}{O1G2Crit}{3.2in}

\section{Most Robust Circle}\label{sec:GlobalRobust}

To provide evidence for \Con{Piecewise}, we studied the rotation number of the most robust circle
\[
	\omega_{max} = \arg\max \{ \eps_{cr}(\omega)\}
\]
for two parameter families of generalized standard maps. In particular, we studied the dependence of $\omega_{max}$ on the parameter $\psi$ in our various force models. To do this, we first obtained a rough estimate for the interval of rotation numbers that contains the most robust circles over a range of $\psi$. Then, for each value of $\psi$, we select the most robust circle from the  $256$ nobles on level $8$ of the Farey tree whose root corresponds to this interval. 

For example, for the frequency map $\Omega_1$, periodicity  allows us to restrict $\omega$ to the interval $[0,\tfrac12]$. The dependence of $\omega_{max}$ on the parameter $\psi$ for the force $g_3$ is is shown in \Fig{01G3MostRobust}(a). Surprisingly, of the $256$ nobles examined, only five are most robust over most of the range of $\psi$; the continued fraction expansions of their rotation numbers are shown in the figure. The critical $\eps$ for each of these circles varies smoothly with $\psi$, as shown in \Fig{01G3MostRobust}(b).

\InsertFigTwo{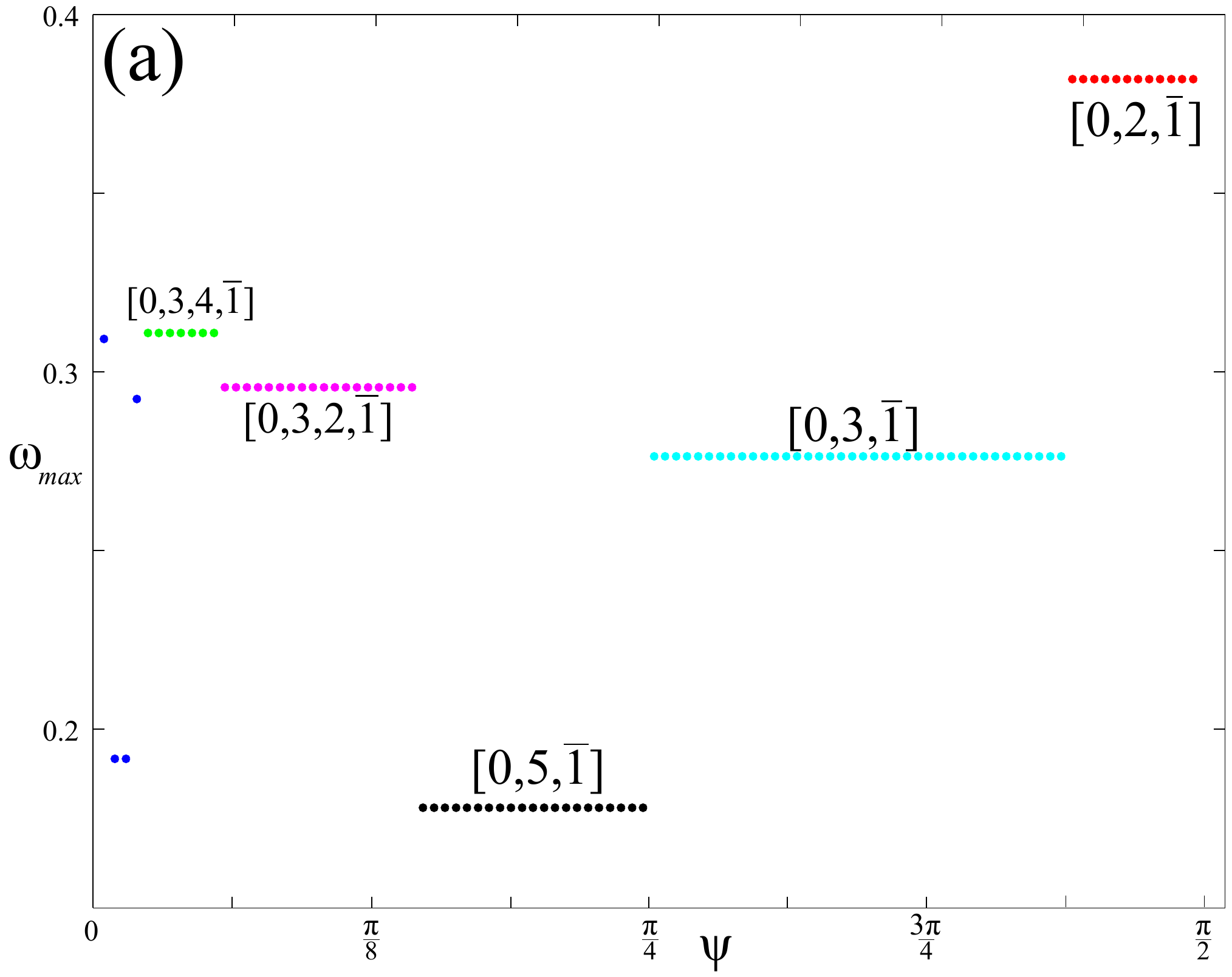}{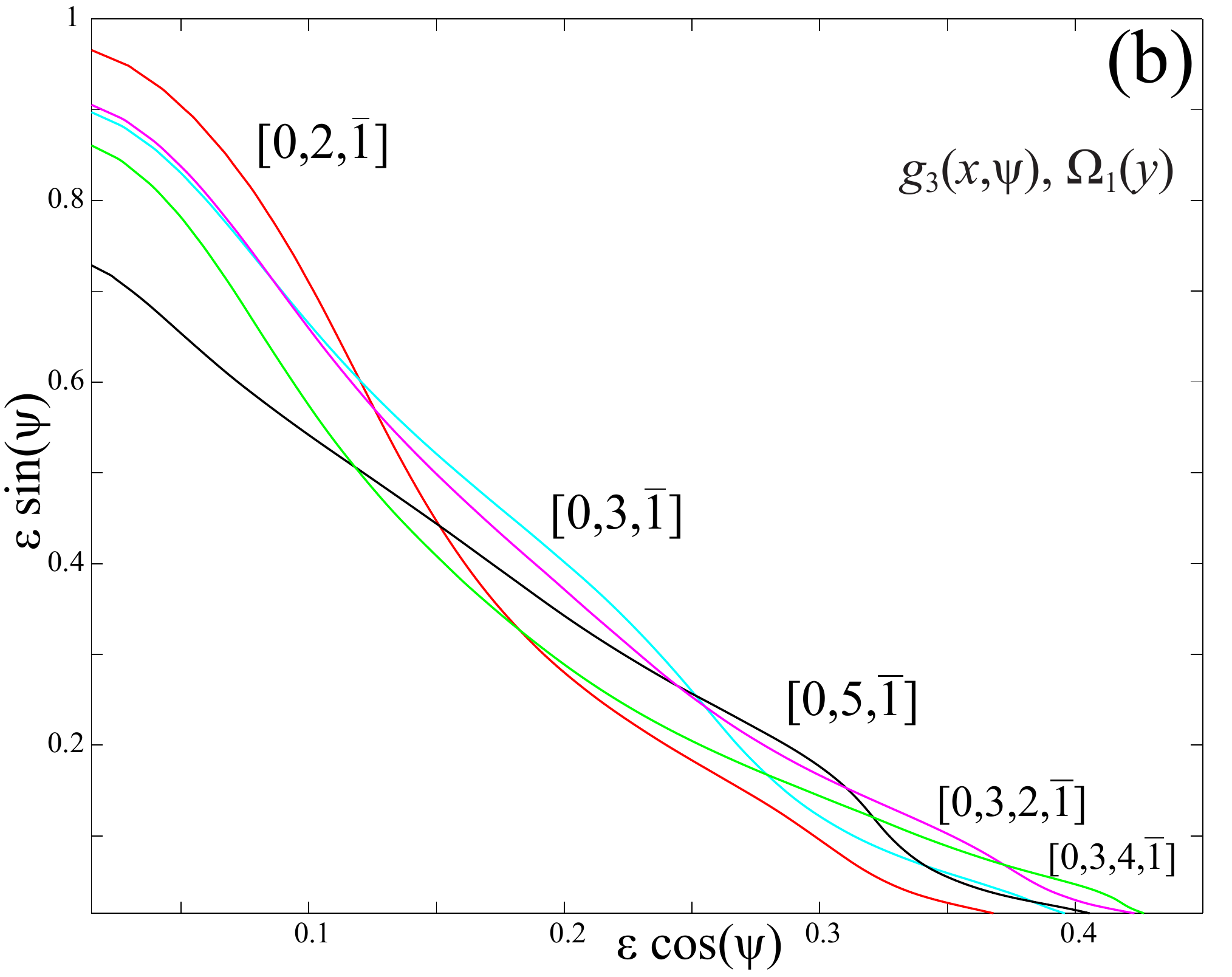}{Globally most robust circles for $\Omega_1$ and $g_3(x;\psi)$ for $100$ values of $\psi \in [0,\tfrac{\pi}{2}]$. Labels show the continued fraction expansions of the rotation numbers for five most robust circles. (a) The rotation number of the most robust circle. (b) The critical sets, $\{(\cos \psi,\sin \psi)\eps_{cr}(\omega)\}$ for the five most robust circles of (a).} {01G3MostRobust}{3in}

The values of $\omega_{max}(\psi)$ shown in \Fig{01G3MostRobust}(a) follows a predictable pattern. As noted in \Sec{CriticalCircle}, this map reduces to Chirikov's map at $\psi=\frac{\pi}{2}$, and to the same map on half the spatial scale at $\psi = 0$. Thus the circle with rotation number $2-\gamma$, whose continued fraction expansion is $[0,2,\overline{1}]$, is most robust at $\frac{\pi}{2}$. When $\psi=0$ the critical function has equal global maxima at the ``half-noble" rotation numbers $\tfrac12(n \pm \gamma)$. For the range $[0,\tfrac12]$ there are maxima at
\bsplit{HalfNobles}
	\tfrac12(\gamma-1) &= [0,3,\overline{4}] ,\\
	\tfrac12(2-\gamma) &= [0,5,\overline{4}]. 
\esplit
For moderate $\psi$, the observed values of $\omega_{max}$ move logically from $2-\gamma$  towards one of the half-nobles, in the sense that their continued fraction expansions become increasingly similar to those of \Eq{HalfNobles}. This behavior persists for small $\psi$, as shown in the vicinity of $\tfrac12(\gamma-1)$ in \Fig{O1G3SmallPsi}. The evolution is not monotonic---indeed, the rotation number of the most robust circle appears to oscillate about $\tfrac12(\gamma-1)$.  For $\psi < 10^{-3}$ the rotation number of the most robust circle in the figure is $[0,3,4,4,2,\overline{1}]$, the noble closest to $\tfrac12(\gamma-1)$ amongst the $256$ in the sample.  

\InsertFigTwo{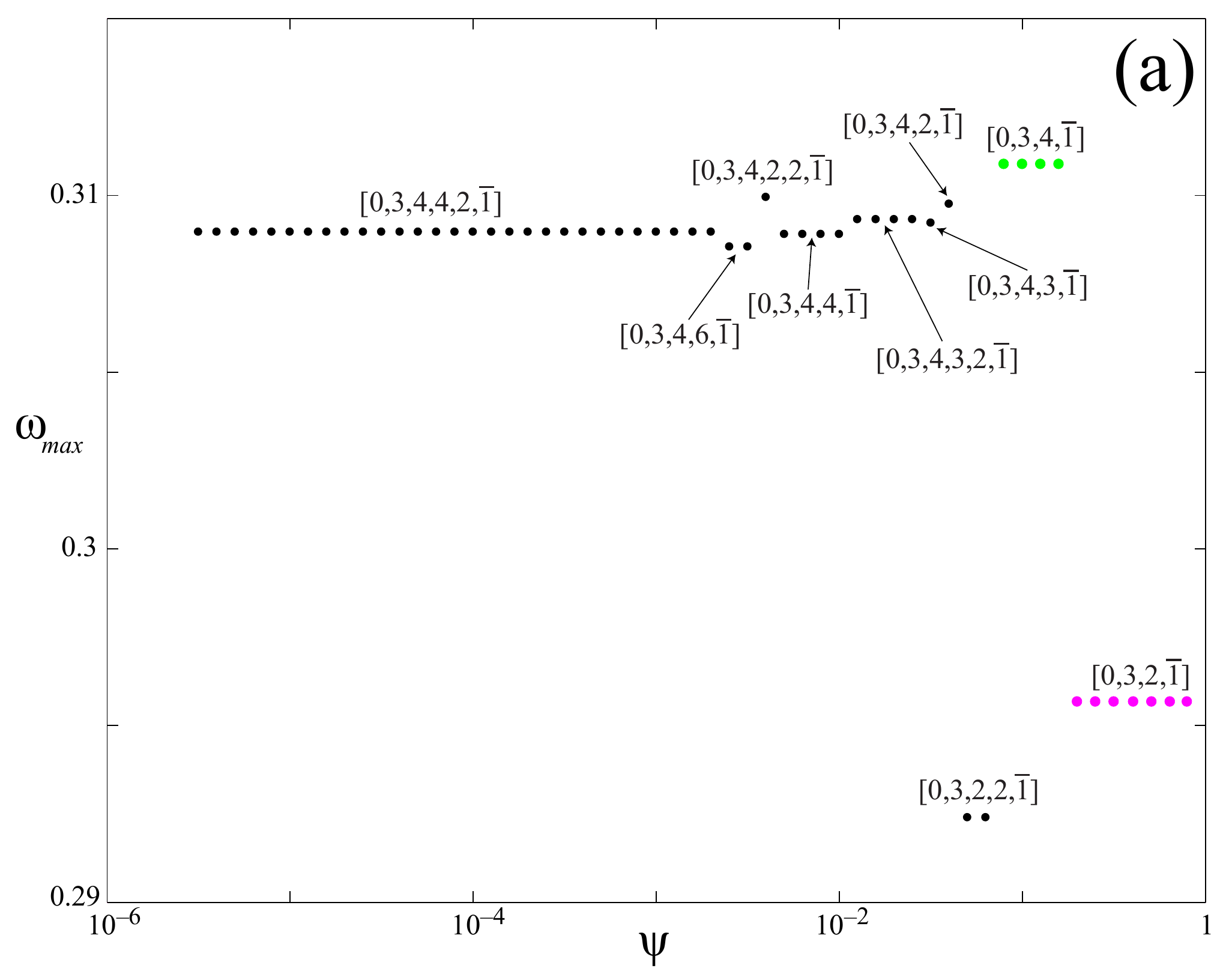}{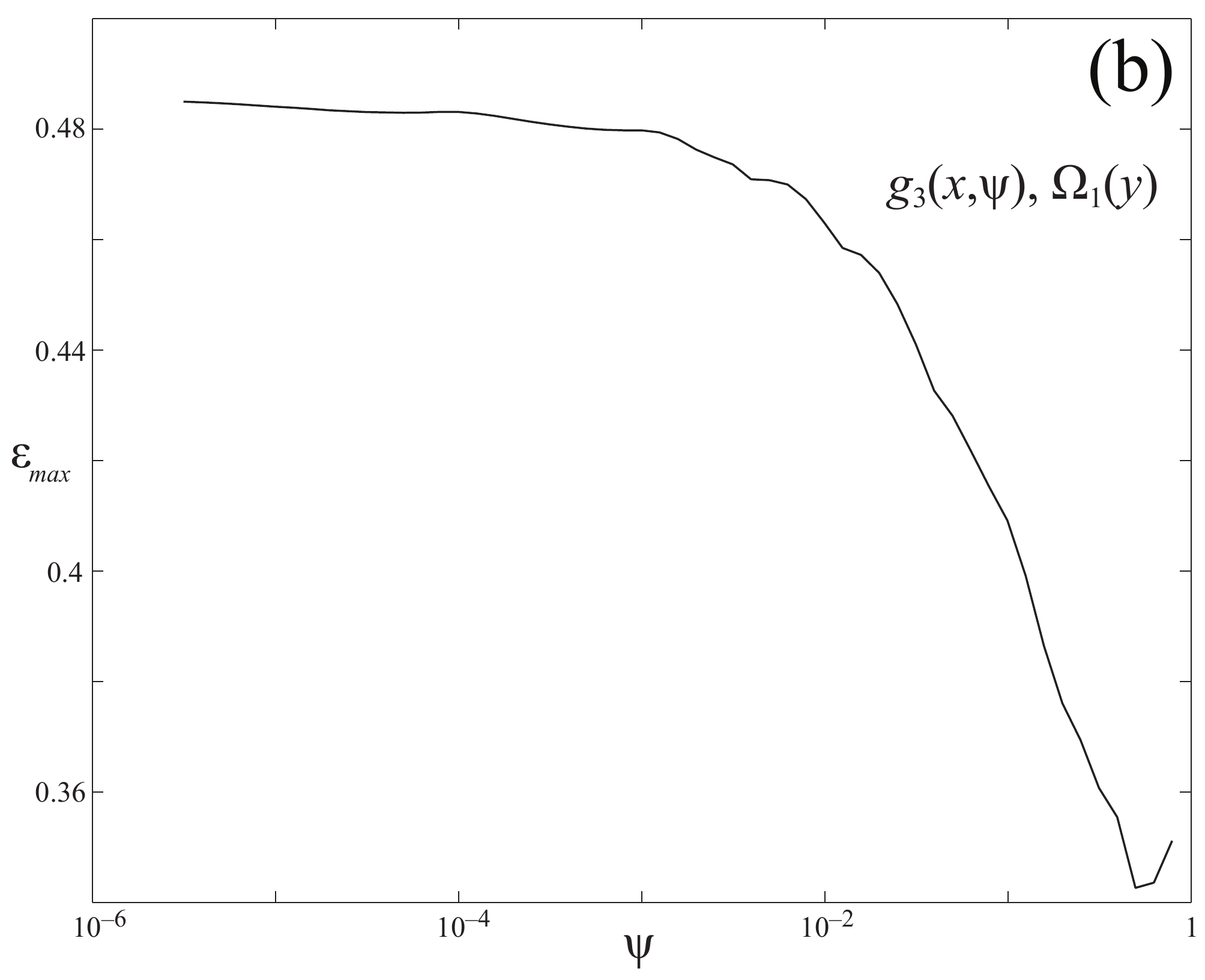}{Most robust circles for $\Omega_1$ and $g_3(x;\psi)$ for $50$ values of $\psi \in [10^{-6},1]$ sampled from $256$ circles with noble rotation numbers in the interval $[\tfrac27,\tfrac13]$. (a) The rotation number of the most robust circle labeled by their continued fraction expansions. (b) The critical curve, $\eps_{cr}(\omega_{max}(\psi))$ for the most robust circles of (a).}{O1G3SmallPsi}{3in}

Thus we see that the rotation number of the most robust circle seems to remain constant on intervals of the parameter $\psi$, at least for the twist map with force $g_3$. We examined several other multiharmonic twist maps in order to determine if this behavior was typical. Figure~\ref{fig:O1G4G5Crit} shows $\omega_{max}(\psi)$ for the generalized standard maps with frequency map $\Omega_1$ and forces $g_4$ and 
\beq{G5}
	g_5(x;\psi)=\frac{\sin(2\pi x)}{2\pi(1 + \psi \sin(4 \pi x))} .
\eeq
Since both maps reduce to Chirikov's map when $\psi=0$, the golden circle is most robust at this point. As $\psi$ grows, the rotation number of the most robust circle grows monotonically in steps approaching $\tfrac12$. It is interesting that the rotation numbers of the most robust circles for the two maps are identical, though the range of $\psi$ is shifted. It seems surprising that the circles with rotation numbers near $\tfrac12$ are most robust as $\psi \to 1$.  A partial explanation for this is that the Fourier expansions of both $g_4$ and $g_5$ have only odd modes, so in particular the $(1,2)$ resonance is suppressed.

\InsertFigTwo{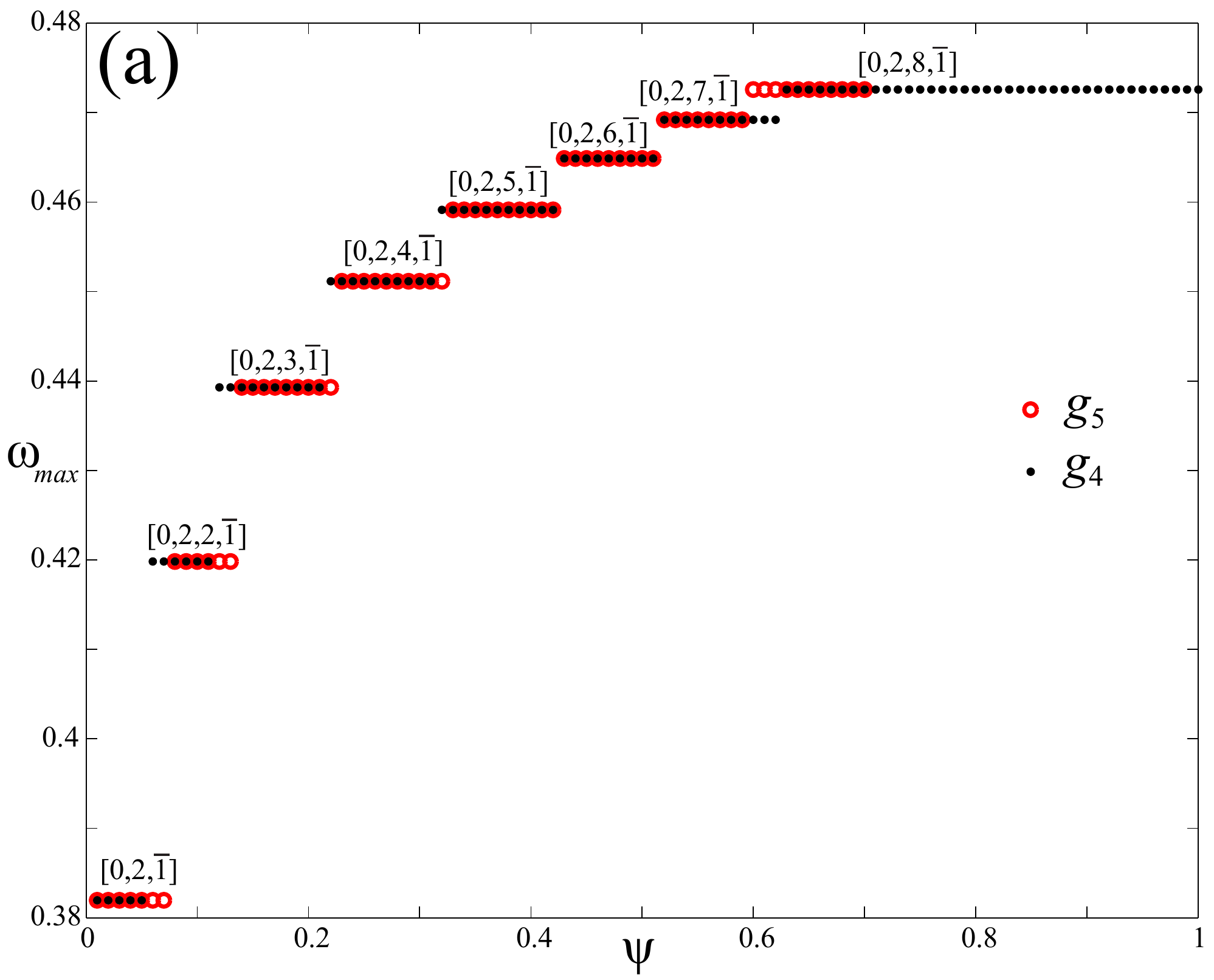}{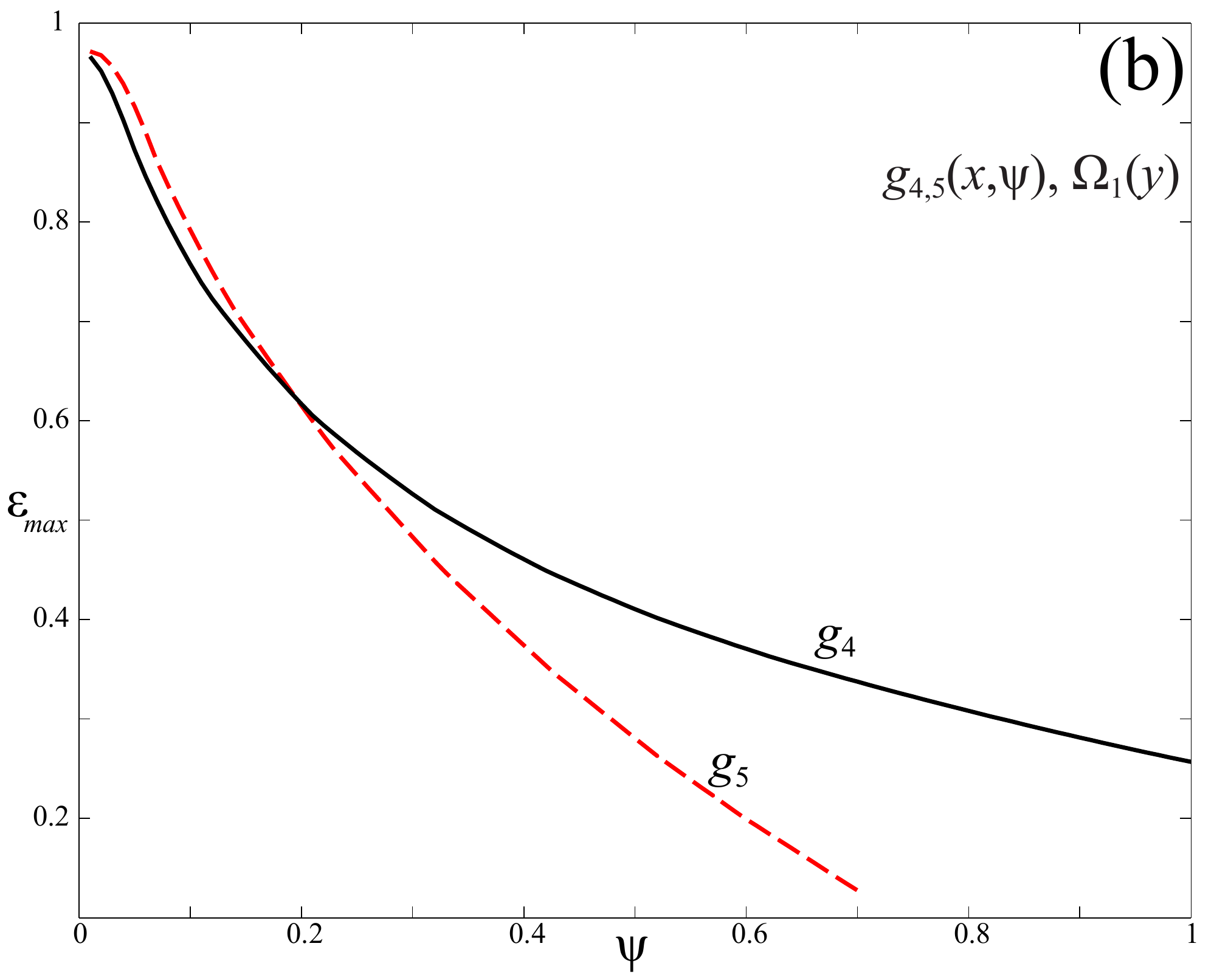}{(a) $\omega_{max}(\psi)$ and (b) $\eps_{cr}(\omega_{max})$ amongst $256$ nobles in $[0,\tfrac12]$ for the twist map with forcing $g_4(x;\psi)$ \Eq{G4} (solid black) for $\psi \in [0,1]$ and $g_5(x;\psi)$ \Eq{G5} (red open circles and dashed line) for $\psi \in [0,0.7]$.} {O1G4G5Crit}{3in}

The behavior of the most robust circles for the fully nonreversible case, with $\Omega=\Omega_3$ and $g=g_4$, is similar to the twist case, see \Fig{RannouRobust}. For this map, the globally most robust circle always had a small rotation number and the proximity of these circles to resonance led to larger errors. To focus on the robust region, we restricted to the $64$ nobles from level $6$ of the Farey tree with the root $[0, \tfrac17]$, and to improve the accuracy, we used up to $2^{15}$ Fourier modes and smaller step sizes in $\eps$, beginning with $\Delta\eps=.005$, decreasing to $\Delta\eps=.0001$ when that step size failed. The results indicate that $\omega_{max}$ is again constant over intervals of $\psi$. Also, just as for the previous cases, $\eps_{cr}(\omega_{max})$ is a monotone decreasing function of $\psi$.

\InsertFigTwo{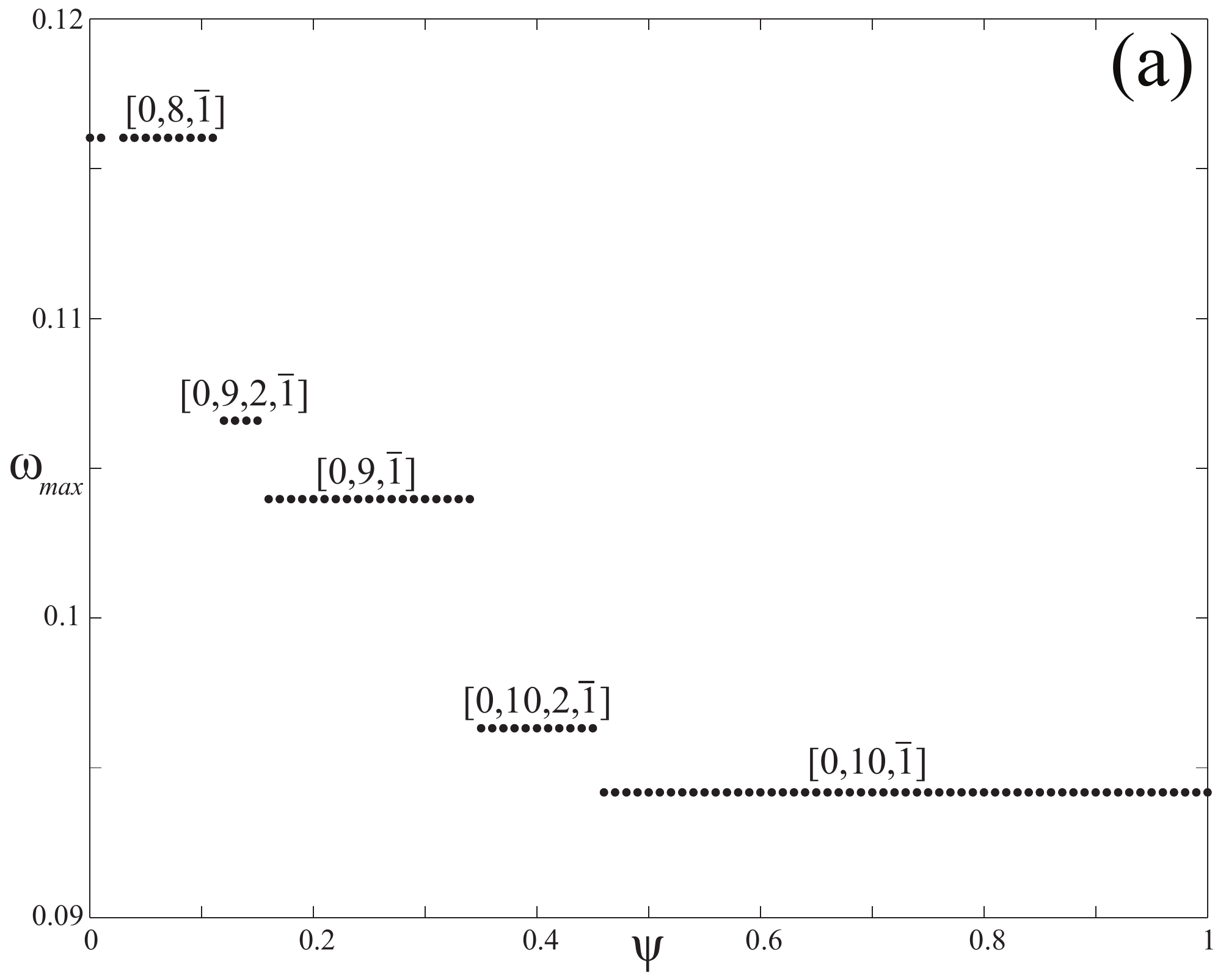}{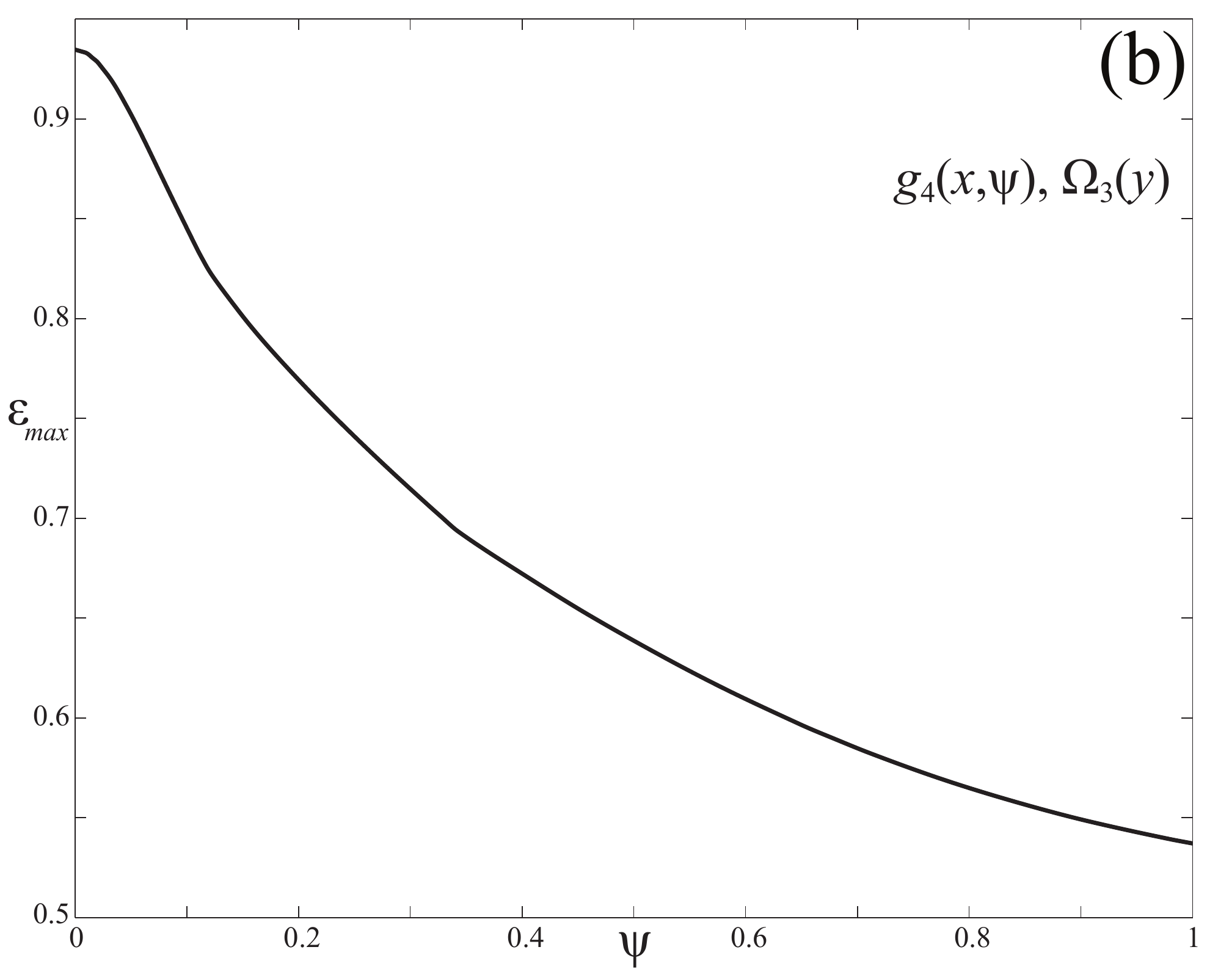}{(a) $\omega_{max}(\psi)$ and (b) $\eps_{cr}(\omega_{max})$  for the generalized standard map with with $\Omega=\Omega_3$, $g=g_4$, estimated using $64$ nobles in $[0,\tfrac{1}{7}]$ for each of $100$ values of $\psi \in [0,1]$.} {RannouRobust}{3in}

In each of the maps that we examined so far, $\eps_{cr}(\omega,\psi)$ is a smooth function of $\psi$ for each fixed noble rotation number. The rotation number $\omega_{max}(\psi)$ appears to be piecewise constant, supporting \Con{Piecewise}. However, for the doubly-reversible, two-harmonic map ($\Omega_1$ and $g_2$), $\eps_{cr}(\gamma,\psi)$ was seen to be a highly irregular function when $\psi \in [0,\tfrac{\pi}{2}]$ (recall \Fig{CritCurves}). A similar irregularity in $\eps_{cr}$ also occurs for other rotation numbers, as suggested by the graph of the critical function \Fig{O1G2Crit}(b). This irregularity also generates erratic behavior in $\omega_{max}$ for the same range of $\psi$, see \Fig{O1G2Best}. Interestingly, when $\psi \in [\frac{\pi}{2},\pi]$, $\omega_{max}$ is piecewise constant and $\eps_{cr}(\omega_{max})$ is smooth, as for the previous examples. Consequently, the irregular behavior of the most robust circles appears to reflect the same complex set of symmetry breaking bifurcations that are responsible for the Cantor set of cusps in the critical set for the golden circle.

Thus we suggest that the conjectured piecewise constant behavior of $\omega_{max}$ will occur generically in families maps that do not have special symmetries that give rise to irregular dependence of $\eps_{cr}$ on a parameter.

\InsertFigTwo{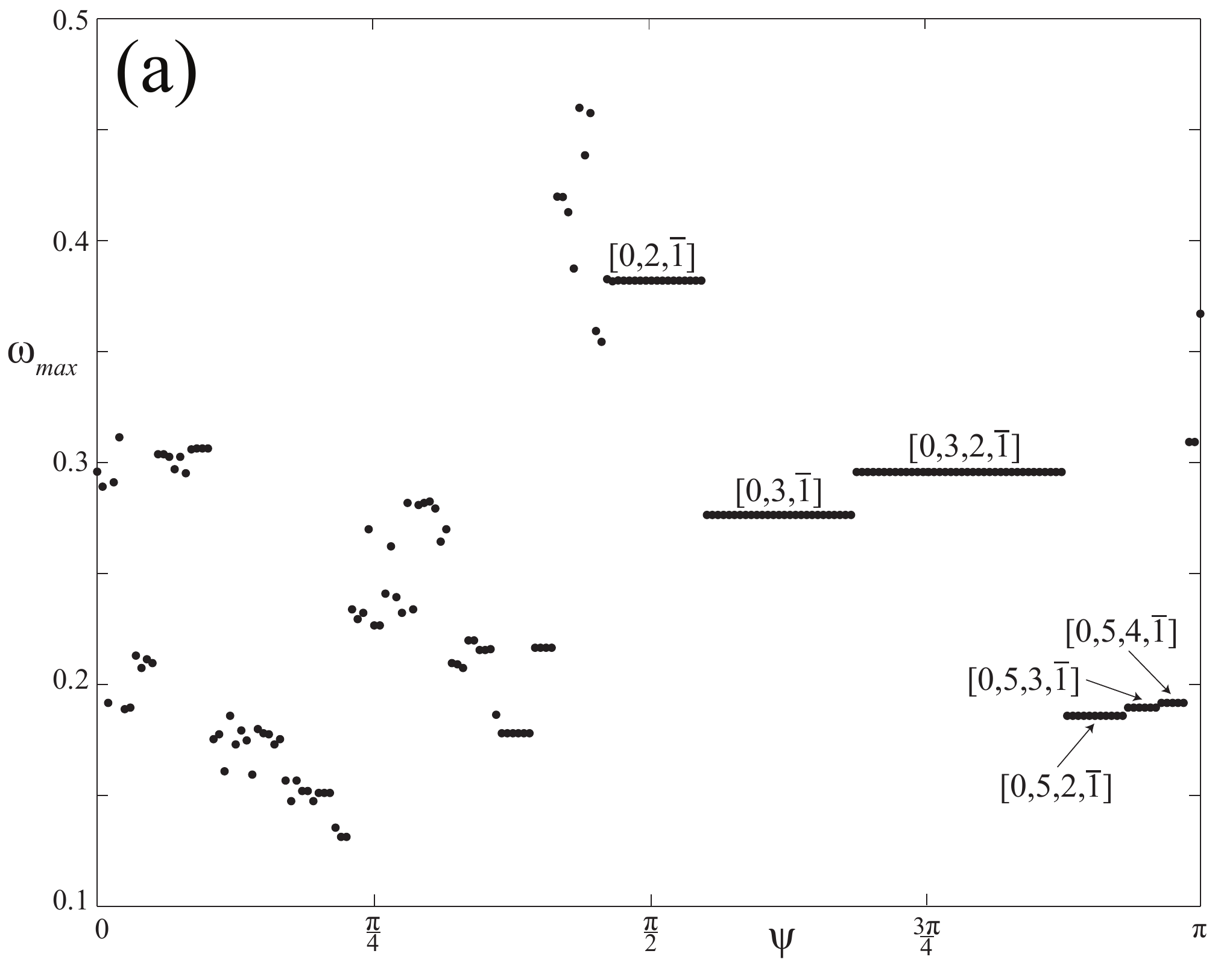}{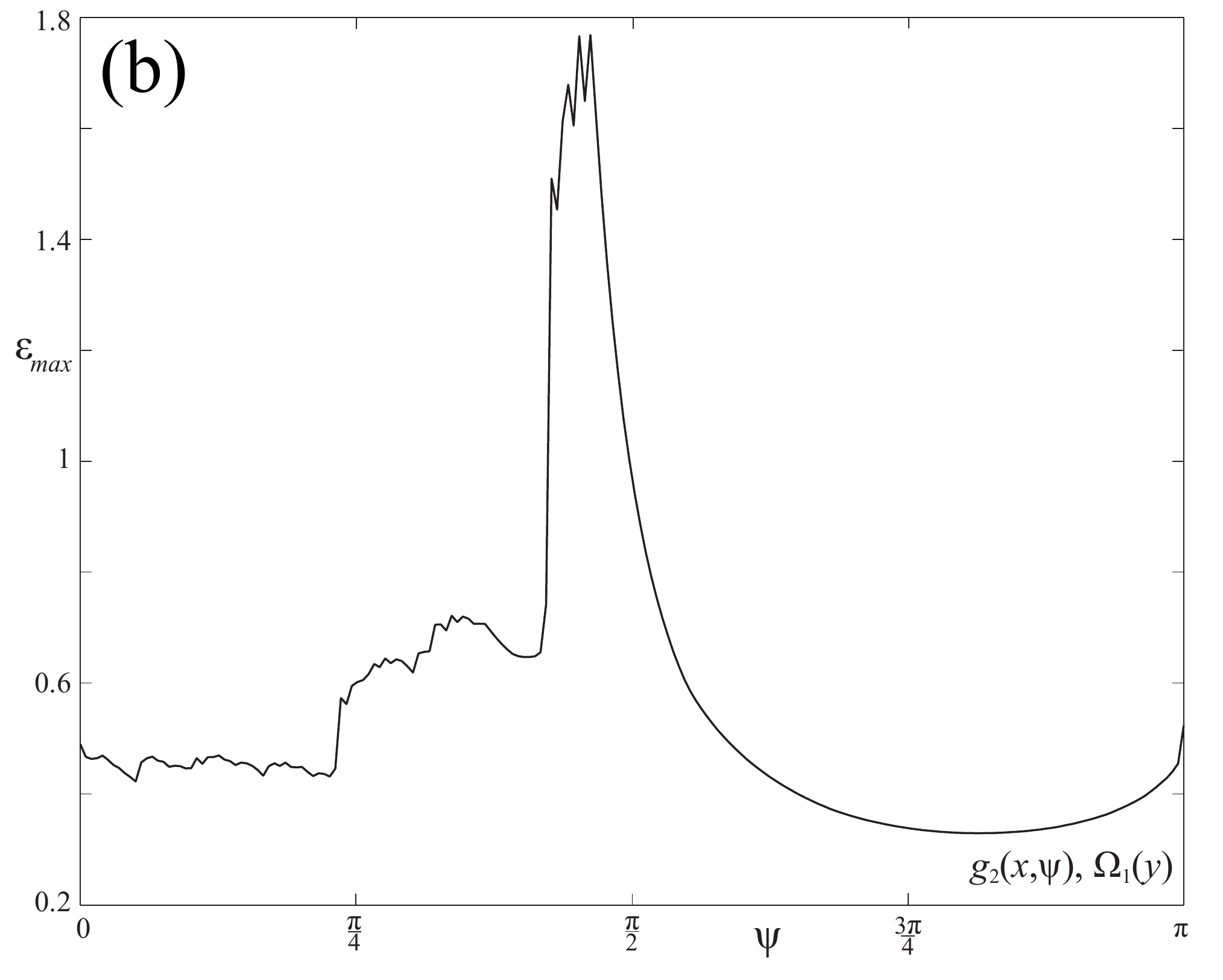}{(a) $\omega_{max}(\psi)$ and (b) $\eps_{cr}(\omega_{max})$ for the generalized standard map with $\Omega_1$ and $g_2$ as a function of $\psi$, estimated using $256$ nobles in $[0,\tfrac12]$ for each of $200$ values of $\psi \in [0,\pi]$.}{O1G2Best}{3in}

\section{Conclusion}
The Fourier method, outlined in \Sec{FindEpsCr}, provides an accurate and efficient approach to find the embedding, $k(\theta)$, for invariant circles of area-preserving maps with a given rotation number. Using the blow-up of a Sobolev seminorm, we obtained reasonably accurate estimates for the parameters at which a Diophantine invariant circle is first destroyed; however, the errors in this computation are larger than that for Greene's residue criterion. Nevertheless, the Fourier method is easily generalizable to maps without the reversing symmetry that makes the computation of periodic orbits for Greene's method relatively easy.

In this paper we have provided some support for the conjectures discussed in \Sec{Intro}. For example, in all of the cases we examined, the near-critical conjugacies show the incipient formation of gaps that signal the mechanism for destruction of the circle and the creation of a remnant cantorus. This provides some evidence that---even when Aubry-Mather theory does not apply (i.e., the twist condition is violated)---cantori form in area-preserving mappings. It would be nice to generalize Aubry's anti-integrable limit to such maps to confirm that they have cantori near $\eps = \infty$.

In \Sec{Nobles} we gave numerical evidence that the relative robustness of a circle depends monotonically on the discriminant of the algebraic ring of its rotation number. Thus the noble invariant circles appear to be locally most robust as in \Con{Noble}.  Again, this conjecture has much support for reversible twist maps, but the seminorm technique has allowed us to extend it to nonreversible and nontwist maps.

Our investigation was limited to rotation numbers in quadratic algebraic rings. We hypothesize that this would extend to more general irrationals, however, to our knowledge, this has never been studied. The extension of this relationship to tori in higher-dimensional symplectic and volume-preserving maps is also an open question. 

Finally, support for \Con{Piecewise} was given in \Sec{GlobalRobust}. There we observed that---except for the doubly-reversible, two-harmonic map---the rotation number of the globally most robust circle seems to be piecewise constant when a second parameter of the map is varied. The exception to this seems to be related to extra symmetry which gives rise to a highly irregular critical function.

\appendix
\newpage
\appendixpage

\section{Quadratic Irrationals and the Farey Tree}\label{app:Farey}

The rotation number $\omega$ of a persistent invariant circle in the standard version of KAM theory satisfies a Diophantine condition, i.e., there is a $c>0$ and an $s>1$ such that
\beq{Diophantine}
			q^s|q\omega-p| > c ,
\eeq
for all $p \in \bZ$ and $q \in \bN$. Numbers that are Diophantine for $s=1$ are called \emph{badly approximable}, and the largest asymptotic value of $c$ for $s=1$ is the \emph{Diophantine constant}:
\beq{DiophantineConsant}
	c(\omega) = \liminf_{q \to \infty} q |q\omega-p| .
\eeq

A number $\omega \in \bC$ is a quadratic algebraic integer if it is the root of a monic quadratic equation with integer coefficients \cite{Hardy79}. The integer $\omega$ generates both an algebraic field, $\bQ(\omega)=\{a + b \omega : a,b \in \bQ\}$, and a corresponding ring of integers, $ \bZ(\omega)=\{a + b \omega: a,b \in \bZ\}$. A basis for this ring is referred to as an \emph{integral basis}.
Every real, quadratic integer is Diophantine, and its constant $c$ can be easily determined. 

\begin{thm}[\cite{Cassels57}]\label{thm:Cassels}
If $\omega$ is a real, irrational root of the monic polynomial $x^2 + b x + c $, with $b,c \in \bZ$, then 
\beq{Dconstant}
	c(\omega) = \frac{1}{\sqrt{D}} ,
\eeq
where $D=b^2-4c$ is the discriminant.
\end{thm}

\noindent
For example, $\gamma$ is a root of $x^2-x-1$, which has discriminant $D=5$. Therefore, $c(\gamma)=1/\sqrt{5}$.  In \Sec{Nobles}, we study robustness of invariant circles for the real quadratic fields with the smallest discriminants, and correspondingly, the largest Diophantine coefficients.

The Farey (or Stern-Brocot) tree is a recursive algorithm that can be used to generate the rationals between a given initial pair. These rationals can in turn be used to generate irrationals in a given algebraic ring. For convenience we write the rational number $\frac{m}{n}$ as the vector $(m,n)$, and will always assume that the greatest common divisor of the components is $1$.

The root, or \emph{level} zero, of the tree is a pair of rationals $(m_1,n_1)$ and $(m_2,n_2)$ that are neighbors, i.e.,
\beq{neighbors}
	m_1n_2-n_1m_2=\pm1 .
\eeq
We typically begin with the neighbors $(0,1)$ and $(1,0)$, though sometimes we use other neighbors to zoom-in to some interval.  The next level is obtained by the mediant operation $\oplus$ defined as 
\beq{MedOp}
	(m_1,n_1) \oplus (m_2,n_2)=(m_1+m_2,n_1+n_2) .
\eeq
So, level one of the tree consists of the single vector $(1,1)=(0,1)\oplus(1,0)$, the \emph{daughter} of the \emph{parent} vectors $(0,1)$ and $(1,0)$. Note that the daughter is a neighbor to each of its parents.

The $(\ell+1)^{st}$ level is constructed by applying the mediant operation to each level $\ell$ vector and its two parents; thus there are $2^{\ell-1}$ rationals on each level. A binary tree is obtained by connecting each vector at level $\ell \geq 1$ to its two daughters at level $\ell + 1$, see \Fig{FareyTree}. There is a unique, finite path, the \emph{Farey path}, from level one to every rational. This path is denoted by a sequence of $l$'s and $r$'s denoting the left and right turns, respectively. Thus the path of $(1,1)$ is empty and,  for example, $p(5,3) = rlr$. Every irrational between the root rationals is uniquely the limit of an infinite path on the tree; for example, $p(\gamma,1) = rlrlrl\ldots = (rl)^\infty$, and $p(1,\gamma) = p(\gamma^{-1},1) = (lr)^\infty$. According to a theorem of Lagrange, the Farey path for every quadratic irrational is eventually periodic and conversely every eventually periodic path converges to a quadratic irrational.

\InsertFig{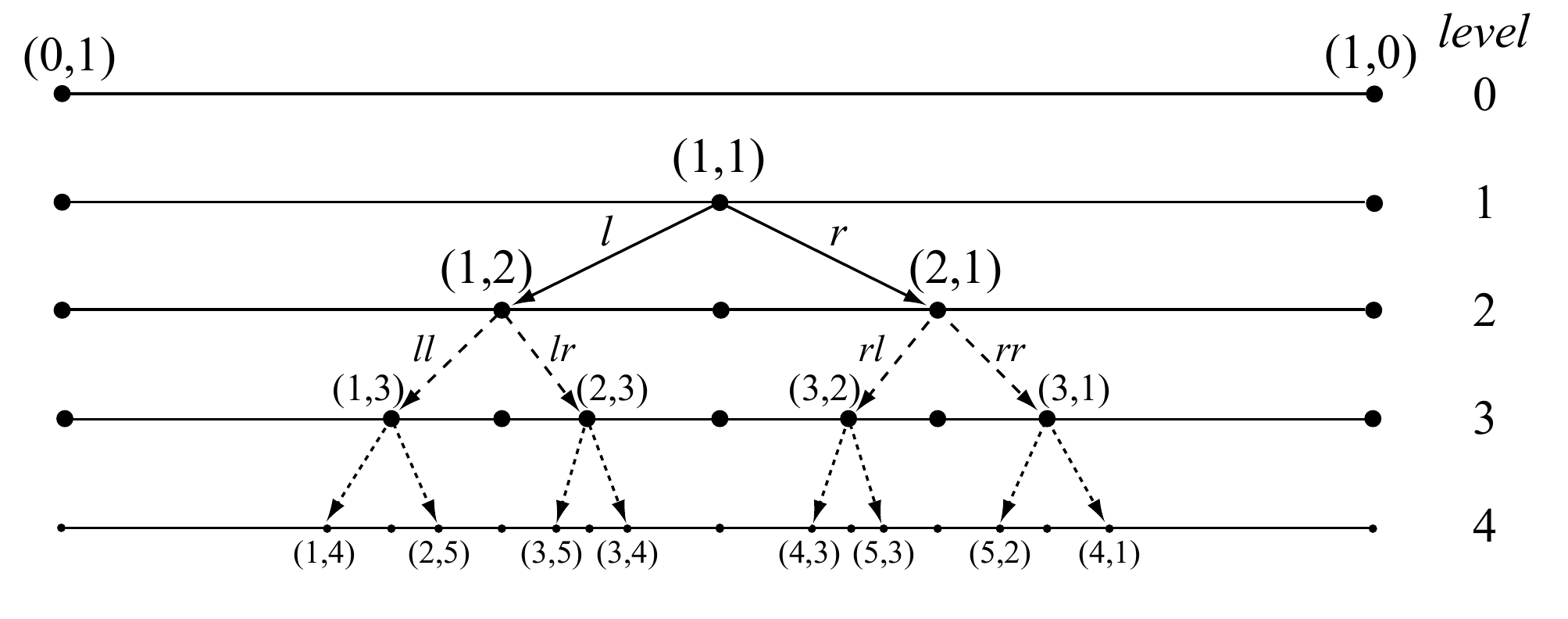}{First five levels of the Farey tree construction.}{FareyTree}{4.5in}

There is a simple relation between the continued fraction expansion of a number and its Farey path.
Suppose that a positive rational has the continued fraction
\[
	\frac{m}{n}=[a_0,a_1,\ldots, a_k] \equiv [a_0,a_1,\ldots, a_k-1,1],
\]
where w.l.o.g. $a_k>1$ for $k \ge 1$. Then
\[
	p(m,n) = \left\{ \begin{array}{ll}
					r ^{a_0}l^{a_1}r^{a_2}\ldots r^{a_k -1} , \quad k \mbox{ even}\\
					r ^{a_0}l^{a_1}r^{a_2}\ldots l^{a_k -1}, \quad k \mbox{ odd}
			\end{array} \right. .
\]
The continued fractions for the two daughters of $\frac{m}{n}$ are obtained by incrementing the last entry by one in the two equivalent representations: 
$[a_0,a_1,\ldots, a_k+1]$ and $[a_0,a_1,\ldots, a_k-1,2]$.

Neighboring rational vectors can be used to generate irrational numbers in the ring $\bZ(\omega)$. To see this, note that $(1,\omega)$ is an integral basis for the ring, and every other integral basis will be of the form $(\eta_1,\eta_2)=(1,\omega)M$ for some unimodular matrix $M$. Denote the rows of $M$ by a pair of neighboring vectors $(m_1,n_1)$, $(m_2,n_2)$; then the irrational
\beq{NewElement}
	\omega' = \frac{\eta_1}{\eta_2}=\frac{m_1 + m_2 \omega}{n_1+n_2 \omega}
\eeq
is in $\bZ(\omega)$. The Farey path for $\omega'$ is essentially the concatenation of the paths of the rational and of $\omega$. More precisely, if $(m_2,n_2)$ is a daughter of $(m_1,n_1)$, and its path ends in an $l$, then $p(\omega',1) = p(m_2,n_2) r p(1,\omega)$, while if it ends in an $r$, then $p(\omega',1) = p(m_2,n_2) l p(\omega,1)$.

In \Sec{Nobles}, the $2^l+1$ rational vectors up to level $l$ on the Farey tree are used to construct $2^l$ irrationals in the ring $\bZ(\omega)$ by applying the relation \Eq{NewElement} to each pair of neighbors.

\section{Greene's Criterion}\label{app:Greene}
Greene conjectured that periodic orbits in the neighborhood of an invariant circle should be stable. Indeed, a sequence of periodic orbits should limit upon a circle only if their Lyapunov exponents converge to zero. Conversely, if the limit of a family of periodic orbits has nonzero Lyapunov exponents, then the invariant circle should no longer exist. This is known as \emph{Greene's residue criterion}. Although initially presented as a conjecture, some aspects of it have been proven \cite{MacKay92a, Falcolini92, Delshams00}.

A type-$(m,n)$ periodic orbit of $f$ is an orbit of the lift,
$
	(x_{t+1},y_{t+1})=F(x_t,y_t), 
$
such that $(x_n,y_n)=(x_0+m,y_0)$. When $f$ is area-preserving, the product of the multipliers of any periodic orbit must be 1. The stability of a periodic orbit is therefore completely determined by the trace of the Jacobian $Df$. Greene characterized this stability through a quantity he called the \emph{residue}
\[
	R=\tfrac{1}{4}(2-\tr(Df^n(x_0,y_0)) .
\]
If $0 < R < 1$, the orbit is linearly stable. 

Greene conjectured that an invariant circle with rotation number $\omega$ will exist if and only if the residues $R_j$ of a sequence of type $(m_j,n_j)$ periodic orbits that approximate the circle,
\[
	\lim_{j \to \infty} \frac{m_j}{n_j}=\omega,
\]
are bounded,
\[
	\limsup_{j \to \infty} |R_j| < \infty .
\]

Rational sequences that approximate $\omega$ can be obtained by truncating either the Farey tree or continued fraction expansion of $\omega$, recall \App{Farey}. To implement the criterion, we compute each periodic orbit using a root-finding method and continuation in $\eps$. This computation is simplified when the map is reversible, see \App{Symmetries}, since the computation can be reduced to a one-dimensional secant method.
A threshold residue, $R^{th}$, is selected, and the parameter value $\eps_{j}^{th}$, such that $|R_j|=R^{th}$ is found, also using the secant method.  The choice of $R^{th}$ is arbitrary; however Greene found that for the golden mean,  convergence was fastest with the choice $R^{th} = 0.25$ \cite{Greene79}. This is the value we use. 

The approximation of $\eps_{cr}$ can be improved by extrapolation. The renormalization theory implies that for noble circles of twist maps,
\[
	\eps_j^{th} \sim \eps_{cr}-A\delta^{-j} ,
\]
for some constants $A$ and $\delta$ \cite{MacKay93}. This leads to a three-point extrapolation scheme: given three sequential approximants, we can compute the parameters $\delta$, $A$, and $\eps_{cr}$.
Extrapolation significantly improves the accuracy of the estimation of the critical parameter. For orbits up to period $30000$, the values of $\eps^{th}_j$ generally converged to $3-4$ digits, but after extrapolation, the last several estimates of $\eps_{cr}$ agreed to $4-8$ digits. 

\section{Symmetries and Reversors}\label{app:Symmetries}
Recall that a homeomorphism $S$ is a symmetry of $f$ if $f \circ S = S \circ f$, and is a reversor if $f \circ S = S \circ f^{-1}$. The collection of reversors and symmetries of $f$ is its reversing-symmetry group $\cG$ \cite{Lamb98}. 

An orbit is \emph{symmetric} under $S$ if it is invariant under $S$. One special class of reversors are orientation reversing involutions, and systems possessing such reversors have been called ``R-reversible" \cite{Lamb98}. Such a reversor has a one-dimensional fixed set $\fix{S}=\{z : S(z)=z\}$, and any symmetric periodic orbit has a point on $\fix{S}$ \cite{MacKay93}. Symmetric orbits can then be computed by a one-dimensional root-finding method along $\fix{S}$ \cite{Meiss08}. It is for this reason that most of the studies of the breakup of invariant circles have been done for R-reversible maps and symmetric invariant circles.

In this appendix, we recall some of the symmetries and reversors of the generalized standard map \Eq{StdMap}. 
Of course, any map $f$ is a symmetry of itself. In addition, 
when $f$ is a map on the cylinder, it commutes with integral rotations. Specifically, if 
\[
	T_{m,n}(x,y)=(x+m,y+n), \quad m,n \in \bZ ,
\]
then $T_{1,0}$ is a symmetry of $f$, or more properly of its lift, $F$, to the universal cover obtained by letting $x\in\bR$. Since the composition of two symmetries is also a symmetry, $F^k \circ T_{m,0}$ is also a symmetry of $F$.

Certain special cases of \Eq{StdMap} have additional symmetries and reversors.
We say that a function $g$ is odd about $\alpha$ if
\[
	g(\alpha+x) = - g(\alpha - x),\quad\quad \mbox{ (odd about } \alpha \mbox{)}.
\]
If the force in \Eq{StdMap} is odd, then this map is reversed by the involution
\beq{S1Reversor}
	S_1(x,y) = (2\alpha-x, y+\eps g(x)) .
\eeq
Note that when $g$ is periodic and odd about $\alpha$, then it is also odd about $\tfrac12 +\alpha$; the corresponding reversor is $T_{1,0} \circ S_1$. Every rotational invariant circle intersects $\fix{S_1} =\{(\alpha,y): y \in \bR\}$. The consequences for the conjugacy \Eq{conjugacy} are discussed in \App{Symmetry}.
Cases studied in this paper with this reversor are shown in \Tbl{SymmetryGroups}.

Similarly, if the frequency map $\Omega$ is odd about a point $\beta$, then \Eq{StdMap} is reversible under the involution
\beq{S2Reversor}
	S_2(x,y) = (x, 2\beta -y- \eps g(x)) .
\eeq
This reversor maps an invariant circle with rotation number $\omega$ onto one with rotation number $-\omega$. A consequence is that these circles have the same critical parameter set: $\eps_{cr}(-\omega) = \eps_{cr}(\omega)$.

Chirikov's standard map, \Eq{StdMap} using \Eq{Twist} and \Eq{Chirikov}, has both sets of reversors since $\Omega$ and $g$ are both odd about zero ($\alpha = \beta = 0$): it is ``doubly reversible." Another way of viewing this is to note that this map has the inversion $I_1(x,y)= (2\alpha-x,2\beta-y)$ as a symmetry. This is not independent of the reversors since $I = S_1 \circ S_2$.

An additional symmetry of maps, like Chirikov's map, with $\Omega(y) = y$ is a translation symmetry in the momentum direction: $f$ commutes with the vertical translation $T_{0,1}$. A consequence is that $\eps_{cr}(\omega) = \eps_{cr}(\omega+n)$ for integer $n$. Combining this with the inversion implies that for Chirikov's map, $\eps_{cr}(\omega) = \eps_{cr}(1-\omega)$; thus, one can limit the rotation numbers studied to $\omega \in [0,\frac12]$. The known symmetries of Chirikov's map are summarized in the first line of \Tbl{SymmetryGroups}.

The map \Eq{StdMap} also has a reversor when \emph{both} $g$ and $\Omega$ are even, e.g.,
\[
	g(\alpha+x) = g(\alpha -x), \mbox{ (even about } \alpha \mbox{)}.
\]
When $g$ is even about $\alpha$ \textbf{and} $\Omega$ is even about $\beta$, \Eq{StdMap}
has the reversor
\beq{S3Reversor}
	S_3(x,y)= (2\alpha-x,2\beta-y- \eps g(x)).
\eeq
Note that this reversor is orientation preserving, and its fixed set is a point. 
This reversor maps an invariant circle with rotation number $\omega$ and positive twist onto a circle with the same rotation number but negative twist.
Examples that we study, recall \Tbl{SymmetryGroups}, include the nontwist maps with $\Omega_2$, which is even about $0$, and forces $g_1$ or $g_3$, which are even about $\tfrac14$ and $\tfrac34$.

Finally, when $\Omega$ is even and $g$ has the odd-translation symmetry, $g(x+\frac12) = -g(x)$, the map \Eq{StdMap} commutes with the symmetry
\beq{HalfShiftSymmetry}
	I_2(x,y) = (x+\tfrac12, -y). 
\eeq
This is the case for the standard nontwist map, i.e., using the Chirikov's force \Eq{Chirikov}, and the frequency map \Eq{NonTwist}. This symmetry was exploited in many studies of the breakup of shearless tori \cite{Shinohara98}. 

When none of the above symmetries of $g$ and/or $\Omega$ hold, then the map \Eq{StdMap} is, as far as we know, not reversible, and its complete symmetry group is $\langle f, T_{1,0}\rangle$. One example of this is due to Rannou \cite{Rannou74}; however, while it is conjectured that this map is not reversible \cite{MacKay93, Roberts92}, as far as we know this question is still open.

\begin{table}[htbp]
 \centering
 \begin{tabular}{@{}lcc|lcc|cc @{}} 
 \multicolumn{6}{c}{\quad Map} &\multicolumn{2}{c}{$\cG$} \\
 Force & Odd & Even & $\Omega$ & Odd & Even &Reversors & Symmetries\\
 \midrule
 $g_1$ & $0$, $\tfrac12$ & $\tfrac14$, $\tfrac34$ & $\Omega_1$ & 0& & $S_1$, $S_2$ & $I_1$, $T_{0,1}$ \\
 & & & $\Omega_2$& & $0$ & $S_1$, $S_3$ & $I_2$\\
 \midrule
 $g_2$ & $0$, $\tfrac12$ & & $\Omega_1$ & 0& & $S_1$, $S_2$ & $I_1$, $T_{0,1}$ \\
 & & &$\Omega_2$ && 0 & $S_1$ & \\
 \midrule
 $g_3$ & & $\tfrac14$, $\tfrac34$ &$\Omega_1$ & 0 & & $S_2$ & $T_{0,1}$ \\
 & & &$\Omega_2$ && 0 & $S_3$ & \\
 \midrule
 $g_4$ & & &$\Omega_1$ &0 &&$S_2$&$T_{0,1}$ \\
 & & &$\Omega_3$ & &&& \\
 \midrule
 $g_5$ & & &$\Omega_1$ &0 &&$S_2$&$T_{0,1}$ \\
 \end{tabular}
 \caption{Known generators of the reversing symmetry groups for various maps of the standard form \Eq{StdMap},
 with forces $g_{1-5}$, \Eq{Chirikov}, \Eq{G2}, \Eq{G3}, \Eq{G4}, and \Eq{G5}, and frequency maps $\Omega_{1-3}$, \Eq{Twist}, \Eq{NonTwist}, and \Eq{Rannou}. The columns labeled ``odd" and ``even"  give the points about which the functions are odd or even, if any. The ``trivial" symmetries $f$ and $T_{1,0}$ are omitted.}
 \label{tbl:SymmetryGroups}
\end{table}

\section{Symmetries of the Conjugacy}\label{app:Symmetry}

Symmetries or reversors act on the conjugacy $k$, \Eq{conjugacy}, to produce additional conjugacies \cite{Apte05, Olvera08}. 

\begin{lem}\label{lem:RevConj} 
Suppose that $k: \bS \to M$ is a conjugacy \Eq{conjugacy} for $f: M \to M$ for a circle with rotation number $\omega$. Then if $S$ is a symmetry of $f$, $\tilde k = S \circ k$ also solves \Eq{conjugacy}, and if $S$ is a reversor, $\tilde k = S \circ k \circ R$ solves \Eq{conjugacy}, where $R(\theta) = -\theta$. 
\end{lem}

\begin{proof} When $S$ is a symmetry, \Eq{conjugacy} implies that
\[
	f \circ S \circ k = S \circ f \circ k = S \circ k \circ T_\omega .
\] Thus $\tilde k = S \circ k$ is also a solution to \Eq{conjugacy}. 
When $S$ is a reversor, then the inverse of \Eq{conjugacy} implies
\[
	f \circ S \circ k \circ R = S \circ f^{-1} \circ k \circ R 
	 = S \circ k \circ T_{-\omega} \circ R = S \circ k \circ R \circ T_\omega .
\]
Thus $\tilde k = S \circ k \circ R$ solves \Eq{conjugacy}.
\end{proof}

A symmetry or a reversor may map one invariant circle onto another, but if it maps a circle onto itself, then \Lem{RevConj} and \Lem{Uniqueness} imply that the conjugacy itself is symmetric up to a shift. One simple reversible example corresponds to the map \Eq{StdMap} when $g$ is odd, see \App{Symmetries}.

\begin{cor}\label{cor:SymConj}
When \Eq{StdMap} is a twist map, $g$ is odd about $\alpha$, $k$ is a continuous conjugacy, and $\omega$ is irrational, then 
there is a $\vphi$ such that $k_x-\alpha$ is odd about $\vphi$, and $k_y$ is even about $\vphi + \tfrac12 \omega$.
\end{cor}
\begin{proof}
For this case we use the reversor \Eq{S1Reversor}.
The new conjugacy given by \Lem{RevConj},
\beq{newConj}
	\tilde k(\theta) = S_1 \circ k(-\theta) = 
			\left(2\alpha-k_x(-\theta), k_y(-\theta)+\eps g(k_x(-\theta)\right),
\eeq
is continuous and has degree-one, so that $\tilde k(\bS)$ is an invariant circle of $f$ with rotation number $\omega$. When \Eq{StdMap} is a twist map, it has at most one rotational invariant circle for each $\omega$, therefore $\tilde k(\bS) = k(\bS)$, and \Lem{Uniqueness} thus implies that $\tilde k(\theta) = k(\theta+2\vphi)$ for some $\vphi$. 
 
Consequently, the $x$-coordinate of \Eq{newConj} gives
\[
	k_x(\theta + 2\vphi) = 2\alpha - k_x(-\theta) .
\]
Setting $\xi=\theta+\vphi$, then gives 
$
	k_x(\vphi+\xi) -\alpha = -(k_x(\vphi-\xi)-\alpha) ;
$
therefore, $k_x -\alpha$ is odd about $\vphi$. 

Since the composition of a symmetry and a reversor is a reversor, a second reversor for \Eq{StdMap} is $ f(S_1(x,y)) = (2\alpha - x + y, y)$.
Composing \Eq{newConj}, with $f$ and using \Eq{conjugacy} this gives
\[
	\tilde k (\theta+ \omega) = f(\tilde k(\theta)) = (f \circ S_1) (k(-\theta)) .
\]
Again using $\tilde k(\theta) = k(\theta+2\vphi)$, the $y$-component of the above equation gives
$
	 k_y(\theta+\omega+2\vphi) = k_y(-\theta) ,
$
which implies the evenness assertion.
\end{proof}

We will use this result as a numerical check on our solutions of \Eq{conjugacy}.
While is seems difficult to check that a function is ``odd" or ``even" about some unknown point $\vphi$, the Fourier coefficients of such a conjugacy, \Eq{KFourier}, must obey
\bsplit{FourierPhase}
	\hat k_{x0} &= \alpha-\vphi ,\\
	\Re(\hat k_{xj} e^{ 2i \pi j \vphi}) &= 0 ,\\
	\Im(\hat k_{yj} e^{2 i\pi j (\vphi+\tfrac12 \omega)}) &= 0 ,
\esplit
i.e., $\vphi$ is determined by the average and the phases of the Fourier coefficients are related.

\bibliographystyle{alpha}
\bibliography{MultiHarmonic}{}

\end{document}